\documentclass[11pt]{article}

\usepackage[margin=1in]{geometry}
\usepackage[utf8]{inputenc}
\usepackage[T1]{fontenc}
\usepackage{graphicx}
\usepackage[usenames,dvipsnames]{xcolor}
\usepackage{mathrsfs}
\usepackage{amsmath,amssymb,graphics,amsthm}
\usepackage{dsfont}    % for identity
\usepackage{bbm}
\usepackage{bm}
\usepackage{braket}
\usepackage{stmaryrd}
\usepackage{comment}
\usepackage{enumitem}
\usepackage{makecell}
\usepackage{enumitem}

\usepackage[colorlinks=true, urlcolor=violet, linkcolor=blue, citecolor=red, hyperindex=true, linktoc=section]{hyperref}
\allowdisplaybreaks

\newtheorem{thm}{Theorem}
\numberwithin{thm}{section}
\newtheorem{cor}[thm]{Corollary}
\newtheorem{lem}[thm]{Lemma}

\newtheorem{prop}[thm]{Proposition}
\newtheorem*{prop*}{Proposition}

\newtheorem{defn}[thm]{Definition}

\makeatletter
\renewcommand{\p@subsection}{}
\renewcommand{\p@subsubsection}{}
\makeatother

\usepackage{xcolor}
\usepackage{mathtools}
\definecolor{brightred}{rgb}{1,0,0}
\definecolor{red}{rgb}{0.76, 0.23, 0.13}
\definecolor{grey}{rgb}{0.3686, 0.5255, 0.6235}
\definecolor{blue}{rgb}{0.0, 0.4, 0.65}
\definecolor{green}{rgb}{0.0, 0.5, 0.0}
\definecolor{ink}{RGB}{42,45,53}
\definecolor{plane}{RGB}{232,238,244}
\definecolor{planeedge}{RGB}{108,128,151}
\definecolor{xred}{RGB}{205,48,57}
\definecolor{xredlight}{RGB}{252,231,232}
\definecolor{zblue}{RGB}{42,104,171}
\definecolor{zbluelight}{RGB}{207,226,245}
\definecolor{blobfill}{RGB}{250,247,237}
\definecolor{blobedge}{RGB}{137,127,103}
\colorlet{lightcyan}{white!70!cyan}
\colorlet{pink}{white!30!red}
\usepackage{pifont} % http://ctan.org/pkg/pifont
\newcommand{\cmark}{\textcolor{green}{\ding{51}}}
\newcommand{\xmark}{\textcolor{red}{\ding{55}}}

\DeclareMathAlphabet{\mathpzc}{OT1}{pzc}{m}{it}

\def\i{\text{i}}

\newcommand{\Aut}{\operatorname{Aut}}

\newcommand{\rank}{\operatorname{rank}}
\newcommand{\rs}{\operatorname{rs}}
\newcommand{\wt}{\operatorname{wt}}

\newcommand{\ZSZ}{\mathrm{ZSZ}}

\newcommand{\Ring}{\mathsf{R}}

\newcommand{\abs}[1]{\lvert{#1}\rvert}

\newcommand{\nm}{\mathrm{n}}

\newcommand{\F}{\mathbb{F}}

\newcommand{\transp}{\mathsf{T}}

\newcommand{\ident}[0]{\mathds{1}}

\usepackage{tikz}
\usetikzlibrary{arrows.meta,positioning,calc,decorations.pathreplacing,backgrounds}
\usepackage{quiver}

\usepackage[ruled,linesnumbered]{algorithm2e}

\SetCommentSty{mycommfont}
\SetKwInOut{Input}{Input}
\SetKwInOut{Return}{Return}

\usepackage{authblk}
\usepackage{appendix}
\usepackage{indentfirst}  % indent first paragraph after section header
\usepackage{booktabs}  % for table formatting
\usepackage{caption}
\hypersetup{
    pdftitle={ZSZ-LP codes}
}

\title{Quantum LDPC codes with design rate 1/5 and good performance below 1000 physical qubits}
\author{Yifan Hong\thanks{Email: \href{mailto:yihong@nvidia.com}{yihong@nvidia.com}}}
\affil{\small NVIDIA}
\date{\today}

\begin{document}
\maketitle

\begin{abstract}
Constant-rate quantum low-density parity-check (LDPC) codes promise fault-tolerant quantum computation with constant spatial overhead in the asymptotic limit. Nonetheless, discovering finite-length code instances with good practical performance remains challenging. We introduce a new family of quantum LDPC codes with design rate $1/5$ and check weight 9 that approaches the teraquop memory regime per qubit-round with several hundred physical qubits, under idling-free circuit-level noise of strength $0.1\%$ and GPU-accelerated Relay-belief-propagation (Relay-BP) decoding with average latencies around 1--2\,ms, a regime relevant to trapped-ion and neutral-atom processors. The construction involves the balanced product of classical LDPC codes with design rate $1/2$ that share non-abelian $\mathbb{Z}_\ell \rtimes \mathbb{Z}_m$ group symmetries, which may be of independent interest for classical error correction. We build syndrome extraction circuits tailored to reconfigurable atom arrays using a simple greedy scheduler, with single-round rearrangement times around 30--60\,ms using present hardware specifications, and substantial room for future improvements. We also construct logical Pauli bases that are equivariant with respect to the group symmetry, which can significantly compress the design space for code surgery. Together, these results further advance the practicality of constant-rate quantum LDPC codes for near-term, fault-tolerant quantum computers.
\end{abstract}

\tableofcontents

\section{Introduction}

Fault-tolerant quantum computing (FTQC) is widely believed to be required for the reliable execution of large-scale quantum algorithms. Central to FTQC, quantum error correction (QEC) is a fundamental primitive that removes decoherence faster than it accumulates \cite{Shor_1995, Steane_1996, Knill_2000, gottesman1997thesis}. The original threshold theorems for FTQC relied on error-correcting codes such as concatenated small codes \cite{Aharonov_1997, Knill_1998} and topological codes \cite{Kitaev_2003} and have polylogarithmic space and time overheads. Later on, the necessary spatial overhead for concatenated codes was reduced to a constant by using codes with rate$\,\rightarrow\!1$ \cite{Yamasaki_2024_FT}. On the other front, theoretical innovations in constant-rate quantum low-density parity-check (qLDPC) codes \cite{HGP, Breuckmann_2021, DLV_2024} have also reduced the spatial overhead to a constant \cite{Gottesman_2014, Fawzi_2020_constant, Tamiya_2026_FT, zhang2026accelerating}, with the current state-of-the-art scheme also achieving a (nearly) logarithmic time overhead \cite{Nguyen_2025_FT}, thus bridging the gap with classical fault tolerance \cite{vonNeumann_1956}.

The above results constitute theoretical breakthroughs in the asymptotic limit and ultimately guide our understanding in realizing FTQC in practice. Estimations of the actual constants and overheads require specifications of underlying QEC codes/mechanisms, their fault-tolerant gadgets, as well as noise models and target fidelities in which thresholds and overheads can be derived. Depending on how applicable one wants the numbers to be, various assumptions such as connectivity and hardware compatibility can be tightened or relaxed. In this realm of practically motivated FTQC, there has been much work done on estimating physical footprints with 2D topological codes, such as the surface \cite{bravyi1998surface} and color codes \cite{Bombin_2006}, on superconducting \cite{Gidney_2021_RSA, gidney2025rsa, babbush2026ecc} and neutral-atom architectures \cite{Zhou_2025_arch}. Similar resource estimates using more general qLDPC codes, which need not be geometrically constrained, have recently also emerged \cite{yoder2025tour, webster2026pinnacle, yang2026rascql, cain2026shor, tripier2026cat}.

The design of practical FTQC architectures typically revolves around the choice(s) of error-correcting code(s) specifically tailored for different parts of the FTQC stack. Homogeneous architectures usually employ one or two ``Jack of all trades'' codes that simultaneously store and compute the logical qubits. Examples include the earlier surface-code architectures, the bicycle architecture \cite{yoder2025tour} with bivariate bicycle (BB) codes \cite{BB_codes}, and the Pinnacle architecture \cite{webster2026pinnacle} with generalized bicycle (GB) codes \cite{GB_codes}. On the other hand, heterogeneous architectures choose codes optimized for specific tasks such as memory, Clifford computation and resource-state distillation. \cite{cain2026shor} is such an example which utilizes quasi-cyclic lifted product (QC-LP) codes \cite{LP_codes} for memory and BB/surface codes for computation. Though it may at first seem like a heterogeneous architecture is superior due to its more fine-grained optimization, allowing for instance extremely efficient memories \cite{cain2026shor, zhao2026ultra}, bridging different error-correcting codes in a fault-tolerant way with low overhead is highly nontrivial. Current approaches include methods based on code homomorphisms \cite{Huang_2023, Xu_2025_fast} and code surgery \cite{Cohen_2022, Williamson_2026_gauging, cross2025improved, Swaroop_2026_adapters, he2025extractors, zheng2025high, cowtan2025fast, yuan2026parsimonious}; the former requires the involved codes to have some shared structure, whereas the latter is more general purpose.

BB codes are the current paragon for high-rate, implementation-friendly and computationally effective qLDPC codes owing to their weight-6 stabilizer checks and two-dimensional translational symmetry. The low check weight enables compact syndrome extraction with relatively high pseudothreshold of $0.7\%$ under uniform circuit-level noise, while the translational symmetry enables a compact biplanar layout for superconducting chips \cite{BB_codes}. This translational symmetry, which is also shared by GB and more generally quasi-cyclic LP codes, also makes them amenable to reconfigurable atom arrays \cite{Xu_2024_constant, Viszlai_2025_GB, cain2026shor}. In addition, it enables logical qubits to be permuted in the code block through cyclic translations, which substantially lowers the design space of surgery gadgets required to augment these codes from memory to Pauli-based computational blocks \cite{Litinski_2019, he2025extractors, yoder2025tour, webster2025explicit}. Cyclic translational symmetries have similarly been used to lower the computational overhead for hypergraph product (HGP) codes \cite{Xu_2025_fast}. Finally, on the decoding side, the translational symmetry has been leveraged to port over fast matching decoders for topological codes \cite{sahay2026matching, tan2026matching}.

A recurring lesson from these prior works is that algebraic structure can make qLDPC codes simultaneously sparse, high-rate, and implementation-friendly. In particular, the translational symmetries $\mathbb{Z}_\ell$ (GB and QC-LP) or $\mathbb{Z}_{\ell_1} \times \mathbb{Z}_{\ell_2}$ (BB) can be thought of as residual symmetries or automorphisms of the lifted \cite{LP_codes} or balanced product \cite{balanced_product}, of which these constructions are special cases, stemming from the fact that these groups are abelian. From the theoretical side, we know that large automorphism groups, despite enabling efficient logical gates \cite{Grassl_2013, SHYPS_codes, phantom_codes}, present an inherent obstacle to achieving asymptotically good $\llbracket n,k,d \rrbracket$ parameters \cite{guyot2026aut, holmes2026logic}, which denote block length ($n$), number of logical qubits ($k$) and minimum distance ($d$) respectively. In practice, this tradeoff is explicitly present in subsystem HGP simplex codes \cite{SHYPS_codes} and ``phantom'' codes \cite{phantom_codes} where $k=O(\log n)$, as well as in clustered cyclic codes \cite{gu2026qgpu} where $kd\leq n$. A natural question to explore, then perhaps, is whether sacrificing the automorphism symmetry can lead to improved code parameters of product codes at finite lengths.

In this work, we construct finite-length qLDPC codes with design rate 1/5 using the lifted/balanced product over non-abelian, metacyclic ``ZSZ'' groups $\mathbb{Z}_{\ell_1} \rtimes \mathbb{Z}_{\ell_2}$, where $\rtimes$ denotes the semidirect product with the normal subgroup on the left. Depending on the choice of FTQC architecture and hardware implementation, these codes can act as memory modules or be used directly for logical computation with aforementioned gadgets such as code surgery. We target systems with movable qubits, such as trapped ions \cite{Kielpinski_2002, Chen_2024_IonQ, quantinuum_helios} and neutral atoms \cite{Saffman_2010_Rev, Bluvstein_2023_logical, Norcia_2024_Atom, rines2026sqale}, and aim for near-term system sizes (e.g. $\lesssim1000$ physical qubits) and physical error rates (e.g. $\lesssim0.1\%$). As we will see, the non-abelian twist ($\rtimes$) introduces both advantages and disadvantages from a practical standpoint. The primary advantage is that, by being non-abelian and sacrificing its potential for quantum code automorphisms, it circumvents an abelian low-distance obstacle when trying to construct a simple, constant-rate generalization of quantum two-block group-algebra (2BGA) codes \cite{2BGA_codes}, which have design rate zero. Our candidate codes have encoding efficiency ratios ($kd^2/n$) ranging from one to two orders of magnitude greater than the surface code, all at modest code lengths below a thousand. The benefits of smaller code lengths are twofold: \emph{(i)} their decoding complexities are simpler, and \emph{(ii)} logical qubits are distributed across more disjoint code blocks which allows for increased computational parallelism. The downside is that the physical connectivity to realize these codes becomes more complicated than their abelian analogues, which increases the hardware requirements for experimental feasibility. Previously, the ZSZ group was used to construct 2BGA codes \cite{2BGA_codes} with favorable single-shot and passive QEC properties compared to BB codes \cite{ZSZ_codes}. We continue this avenue of investigation and show that the same group can be used to construct constant-rate, five-block qLDPC codes with good finite-length performance. This five-block form can be thought of as the simplest lifted/balanced product family with nonzero design rate. Accordingly, we call these codes ``ZSZ-LP'' codes.

Our emphasis is therefore finite-size and implementation-oriented. We search for ZSZ-LP instances below roughly a thousand physical qubits per memory block with guided heuristics (Section \ref{sec:code construction}), identify equivariant logical Pauli bases for computation (Section \ref{sec:logical basis}), identify $ZX$-duality structures for fold-transversal gates (Section \ref{sec:ZX dualities}), build explicit syndrome-extraction schedules inspired by reconfigurable atom arrays (Section \ref{sec:syndrome extraction}), and finally benchmark circuit-level memory performance with GPU-accelerated decoding (Section \ref{sec:numerics}). We conclude with some open directions for future work (Section \ref{sec:outlook}). We emphasize that our code search is nowhere near exhaustive, and codes with even better performance may likely exist; the aim of our candidate codes is simply to showcase the limits of what is possible. The main text serves as a pedagogical guide that establishes motivation and intuitive arguments while trying to avoid too much technical jargon, and the appendices supply rigorous mathematical statements for support and implementation details for reproducibility. Our main results are summarized below.

\subsection{Summary of main results}

\subsubsection{Quantum LDPC codes with design rate 1/5}

\textbf{qLDPC construction.} We introduce a new family of quantum LDPC codes with encoding rate at least 1/5 whose code-capacity memory efficiency, measured in terms of $kd^2/n$ for a $\llbracket n,k,d \rrbracket$ quantum stabilizer code, exceeds all prior constructions under block length $n<1000$ and is also competitive with several others above $n>1000$; see the left panel of Figure \ref{fig:zsz-lp vs other kd2n}. All $X$ and $Z$ parity checks have weight 9, and each qubit participates in at most 6 $X$-checks and 6 $Z$-checks. The code parameters in addition to other useful quantities are listed in Table \ref{tab:ZSZ-LP codes intro}.

We build our quantum ZSZ-LP codes by first choosing pairs of candidate rate-1/2 classical LDPC codes sharing common ZSZ group symmetries, ensembled and filtered by minimum distance and girth. We then take the balanced product \cite{balanced_product} of these classical code pairs to obtain our rate-1/5 quantum LDPC codes. When the classical parity-check matrices have full rank, the logical dimensions of the classical and quantum codes are the same. We also find that for all code instances in Table \ref{tab:ZSZ-LP codes intro}, the (exact or estimated) minimum distances of the quantum codes are the same as those of their classical input codes. We also discuss how to build $ZX$-symmetric ZSZ-LP codes based on automorphism-fold symmetries, which can lead to fold-transversal Hadamard and phase gates. We provide theoretical upper bounds on Tanner-graph girths and minimum distances in Appendix \ref{app:ZSZ-LP codes}, which guide our design choices in constructing our candidate codes. In particular, we show in Theorem \ref{thm:abelian-canonical-syzygies} that an abelian group results in a low (constant) minimum distance, which is our primary motivation for choosing the ZSZ group.

\begin{figure}[t]
    \centering
    \includegraphics[width=0.49\textwidth]{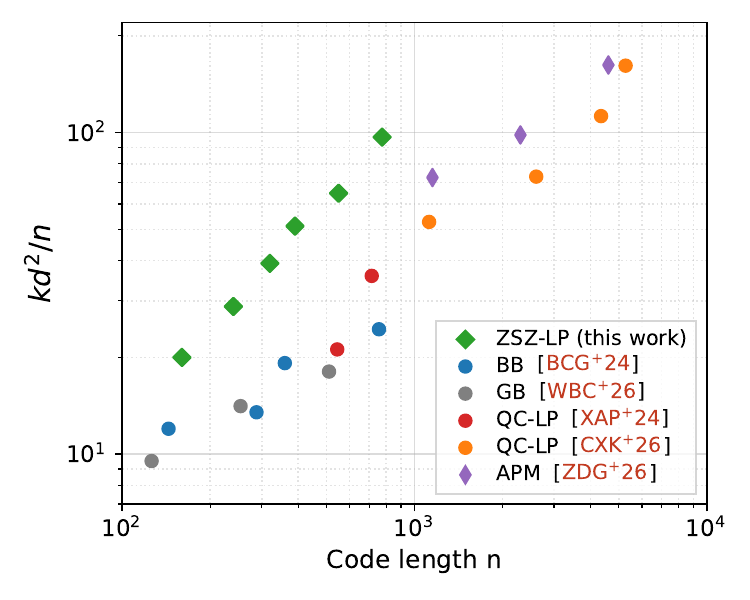}
    \includegraphics[width=0.5\textwidth]{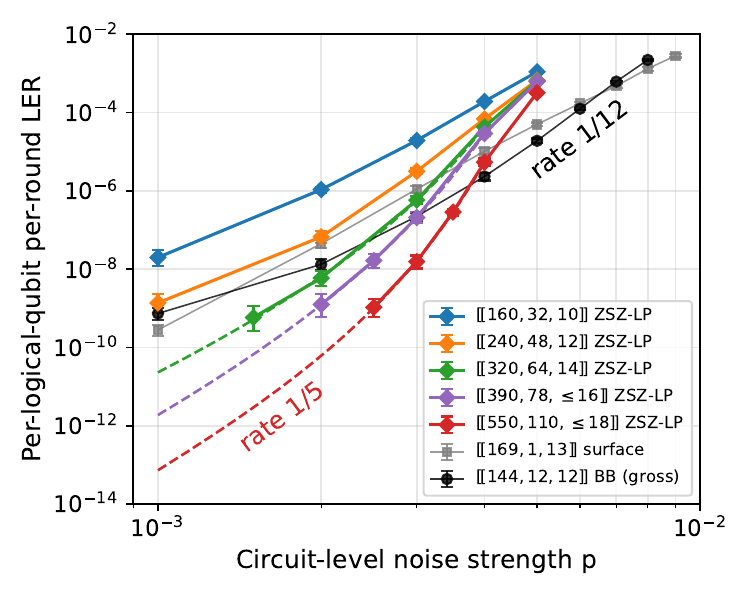}
    \caption{\textbf{Left:} The memory efficiency improvement $kd^2/n$ over the surface code is plotted as a function of the code length $n$ for several high-rate qLDPC codes proposed for near-term implementation, including bivariate bicycle (BB) \cite{BB_codes}, generalized bicycle (GB) \cite{webster2026pinnacle}, quasi-cyclic lifted product (QC-LP) \cite{Xu_2024_constant, cain2026shor} and Kasai's construction \cite{kasai2026breaking} with affine permutation matrices (APM) \cite{zhao2026ultra}. Note that $kd^2/n=1$ for the (rotated) surface code. \textbf{Right:} The $XZ$-averaged logical error rate (LER) per qubit-round (over $d$ rounds) is plotted as a function of the idling-free circuit-level noise strength $p$ for several candidate ZSZ-LP codes as well as the gross code and $d=13$ surface code. The high-rate qLDPC codes are decoded with a GPU-implementation \cite{cudaq_qec} of the Relay-BP message-passing algorithm \cite{Relay_BP}, and the surface code is decoded with minimum-weight matching \cite{Higgott2025sparseblossom}. Error bars denote 95\% confidence intervals. Low-error-rate extrapolation (dashed curves) is performed with the phenomenological fitting function $\bar{p} = p^{d_{\rm circ}/2}{\rm e}^{c_0+c_1p+c_2p^2}$ with fit parameters $c_0,c_1,c_2$ \cite{BB_codes}.}
    \label{fig:zsz-lp vs other kd2n}
\end{figure}

\begin{table}[t]
\centering\renewcommand{\arraystretch}{1.25}
\resizebox{\textwidth}{!}{
\begin{tabular}{cccccccc}
\toprule
\textbf{Code name} & $\llbracket n,k,d \rrbracket$ & $n+m$ & $kd^2/n$ & \textbf{\makecell{Check \\ weight}} & \textbf{\makecell{OGM\\ $\bar{Z},\bar{X}$}} & \textbf{\makecell{OGS\\ $\bar{Z},\bar{X}$}} & \textbf{\makecell{LER/$k$ at \\ $p=0.1\%$}} \\ \midrule
ZSZ-LP-160 & $\llbracket 160,32,10 \rrbracket$      & 288  & 20 & 9 & \cmark, \cmark & \cmark & $2.0(3)\times10^{-8}$ \\
ZSZ-LP-240 & $\llbracket 240,48,12 \rrbracket$      & 432  & 28.8 & 9 & \cmark, \cmark & \cmark & $1.4(2)\times10^{-9}$ \\
ZSZ-LP-320 & $\llbracket 320,64,14 \rrbracket$      & 576  & 39.2 & 9 & \cmark, \cmark & \cmark & $2.3(7)\times10^{-11}$ \\
ZSZ-LP-390 & $\llbracket 390,78,\leq16 \rrbracket$  & 702  & 51.2 & 9 & \cmark, \cmark & \cmark & $2(1)\times10^{-12}$ \\
ZSZ-LP-550 & $\llbracket 550,110,\leq18 \rrbracket$ & 990  & 64.8 & 9 & \cmark, \cmark & \xmark & $7(8)\times10^{-14}$ \\
ZSZ-LP-775 & $\llbracket 775,155,\leq22 \rrbracket$ & 1395 & 96.8 & 9 & $140,129$ & \xmark &   \\ [0.2em]
\Xhline{0.10pt}
APM \cite{zhao2026ultra} & $\llbracket 2304,1156,\leq14 \rrbracket$ & 3456 & 98.3 & 12 & \xmark & \xmark & $\approx 1.3\times10^{-13}$ \\
$\textsf{lp}^{3,7}_{20}$ \cite{cain2026shor} & $\llbracket 4350, 1224, \leq20 \rrbracket$ & 7500 & 112 & 10 & \xmark & \xmark & $\approx 8\times10^{-14}$ \\
\bottomrule
\end{tabular} }
\caption{Several candidate ZSZ-LP codes are tabulated including their $\llbracket n,k,d \rrbracket$ code parameters, physical (data + ancilla) memory footprint $n+m$, efficiency ratio $kd^2/n$ and (maximum) check weight. Exact and upper bounds on code distances are computed using the \textsf{pySATDist} and \textsf{QDistEvol} ($10^6$ iterations) methods respectively \cite{webster2026distance}. In the sixth column, we indicate whether a complete basis of one-generator-minimal (OGM) logical $\bar{X}$ and $\bar{Z}$ operators with weight $d$ has been found; otherwise we report the maximum number of linearly independent OGM operators found. The seventh column indicates whether a one-generator-symplectic (OGS) basis exists. The final column reports the (actual or extrapolated) logical error rate (LER) per qubit-round at $p=0.1\%$. For comparison, we also list two other constant-rate qLDPC codes previously proposed for the teraquop regime whose reported LERs were benchmarked under similar noise models but different decoders.}
\label{tab:ZSZ-LP codes intro}
\end{table}

\textbf{Circuit-level performance.} We numerically benchmark the memory performance of the first five ZSZ-LP codes in Table \ref{tab:ZSZ-LP codes intro} under idling-free circuit-level depolarizing noise with two-qubit gate and state-preparation and measurement (SPAM) errors; the per-qubit-round logical error rates (LER) are plotted in Figure \ref{fig:zsz-lp vs other kd2n}. These first five codes have a combined (data + ancilla) physical qubit footprint under 1000. Using reconfigurable-array-inspired syndrome extraction schedules and a GPU-accelerated Relay-BP decoder \cite{cudaq_qec}, we observe a pseudothreshold around $\approx0.5\%$, below which the LERs drop steeply with increasing code distance. The gigaquop (teraquop) regime, relevant to large-scale applications such as cryptanalysis \cite{Shor_1994_alg}, is defined as $10^9$ ($10^{12}$) reliable quantum operations and so, as a baseline, requires the per-qubit-round LER for memory to be below $10^{-10}$ ($10^{-13})$. Based on the numerical simulations, ZSZ-LP-320 and ZSZ-LP-390 can potentially reach the gigaquop regime, and ZSZ-LP-550 can potentially reach the teraquop regime at $p\leq0.1\%$ under our noise model, decoder and projected extrapolations. We cautiously say ``potentially'' for several reasons: \emph{(i)} the Relay-BP decoder can behave suboptimally at low noise due to trapping sets and may not be able to correct all errors of weight below half the circuit distance, which our extrapolation suggests; \emph{(ii)} the actual per-logical-qubit LER in practice is not really the same as dividing the block LER by the number of logical qubits, since a block failure could be correlated among many logical qubits; \emph{(iii)} we only benchmark memory, and achieving the aforementioned quop regimes also requires a complete architecture across a universal logical gate set. At $p=0.1\%$, we also benchmark the latency of the GPU Relay-BP decoder and find average latencies of around 1-2\,ms for ZSZ-LP-320 through ZSZ-LP-550, suggesting that our decoding configuration may already be sufficient for real-time decoding in trapped-ion and neutral-atom architectures. To the best of our knowledge, ZSZ-LP-320 through ZSZ-LP-550 are the first quantum LDPC codes with rate $\geq10\%$, and under 1000 physical (data + ancilla) qubits per block, to approach the gigaquop and teraquop memory regimes under idling-free circuit-level noise.

\textbf{Equivariant logical Pauli bases.} We leverage the ZSZ group symmetry of the classical seed codes to construct two kinds of logical Pauli bases: one-generator minimal (OGM) and one-generator symplectic (OGS). The Pauli operators in an OGM basis all have minimum weight $d$ but may not be symplectic, while those in an OGS basis are symplectic but may have extensive weight. ``One-generator'' refers to the fact that all basis operators of one Pauli type are constructed through the ZSZ group orbit of one seed codeword and are thereby equivariant under the ZSZ group. OGM and OGS bases thus represent a tradeoff between sparsity and utility, with predilection hinging on the desired application of the code. For each candidate code other than ZSZ-LP-775, we find a full-rank, non-symplectic OGM logical basis; see Table \ref{tab:OGM generators} for a presentation in terms of the ZSZ group algebra. For ZSZ-LP-160 through ZSZ-LP-390, we also find full-rank OGS logical bases at the expense of larger operator weights; see Table \ref{tab:OGS generators}. These group-equivariant logical Pauli operators can enable the construction of low-overhead surgery gadgets for logical I/O or direct Pauli-based computation on the code block \cite{webster2025explicit}. For one-generator codes in particular, owing to the classical automorphism symmetry, a single $X$-surgery gadget can be used to measure any $\bar{X}$-basis operator, and likewise for the $\bar{Z}$-basis.\footnote{Since the quantum ZSZ-LP code does not retain the full ZSZ group as an automorphism, it is slightly weaker than the case described in \cite{webster2025explicit}. In particular, automatic distance preservation of a surgery gadget between equivariant operators is not necessarily guaranteed anymore.} We present two examples of surgery gadgets for ZSZ-LP-390 in Table \ref{tab:surgery main}.

\subsubsection{Classical LDPC codes with design rate 1/2}

Along the way of constructing the quantum ZSZ-LP codes, we build classical $(3,6)$-regular LDPC codes with design rate 1/2 based on the same ZSZ group algebra that may be of independent interest for classical forward error correction (FEC). These codes enjoy symmetries governed by free actions of the ZSZ group, which we use for both the balanced product as well as constructing equivariant codewords. We find that these classical ZSZ-2BGA codes have statistically higher minimum distances than other $(3,6)$-regular ensembles such as the bipartite configuration model \cite{Bender_1978, Bollobas_1980, ModernCodingTheory} and Progressive Edge Growth (PEG) \cite{PEG_LDPC}; see the left panel of Figure \ref{fig:classical zsz vs random dist}. Our classical ZSZ-2BGA code pairs used to construct our candidate quantum ZSZ-LP codes can be found in Table \ref{tab:classical ZSZ code params}. The equivariant OGM codewords for each code are constructed from automorphism orbits of a seed weight-$d$ codeword; all orbits except those for the last ZSZ-2BGA code have full rank. In addition, the classical seed codes of ZSZ-LP-160 through ZSZ-LP-390 admit one-generator-systematic (OGS) generator matrices, which trade the minimum-weight property for systematic encoding. To the best of our knowledge, these classical ZSZ-2BGA codes are the first $(3,6)$-regular LDPC codes to possess such distance and codeword properties at these block lengths. We provide theoretical upper bounds on minimum distances and girths of abelian-2BGA and ZSZ-2BGA codes in Appendix \ref{app:classical 2BGA codes}.

\begin{figure}[t]
    \centering
    \includegraphics[width=0.47\textwidth]{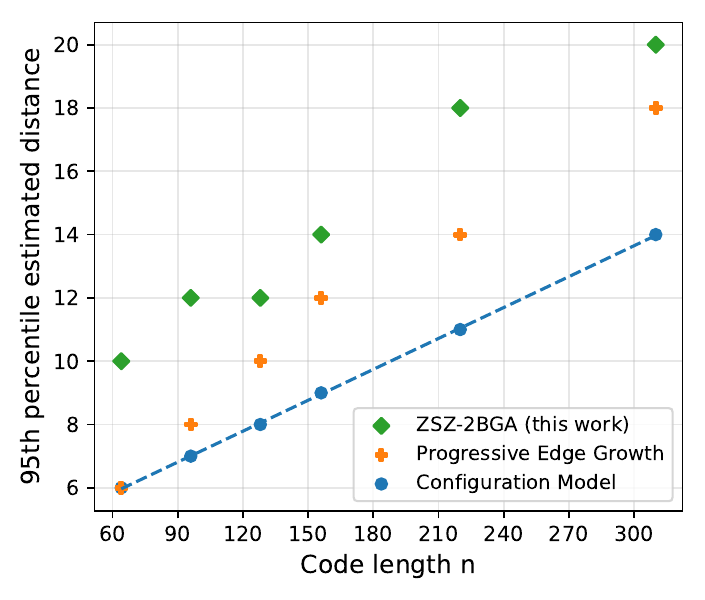} \hfill
    \includegraphics[width=0.505\textwidth]{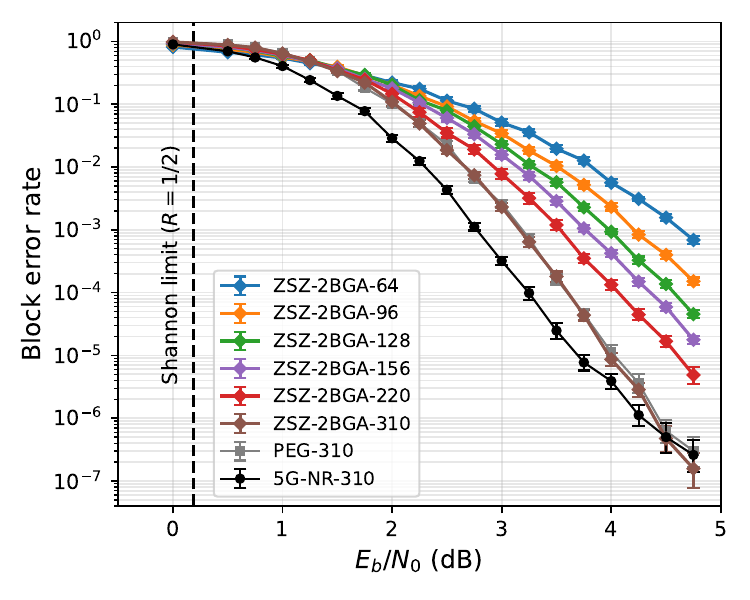}
    \caption{\textbf{Left:} The 95th percentile distance as a function of the code length $n$ is plotted for three random $(3,6)$-regular LDPC ensembles: classical two-block ZSZ (ZSZ-2BGA) codes, the bipartite configuration model \cite{ModernCodingTheory}, and the progressive-edge-growth (PEG) algorithm \cite{PEG_LDPC}. 1000 codes are sampled from each ensemble, and minimum distances are estimated with 10,000 iterations of \textsf{QDistEvol}. A linear fit is applied to the configuration model since its $(3,6)$-regular ensemble is statistically known to be an asymptotically good family \cite{ModernCodingTheory}. \textbf{Right:} The block error rate under the additive white Gaussian noise (AWGN) channel with binary phase shift keying (BPSK) modulation is plotted for the classical left ZSZ-2BGA codes in Table \ref{tab:classical ZSZ code params} as well as a PEG code and a degree-irregular 5G-NR LDPC code at block length 310. $E_b/N_0$ is the normalized signal-to-noise ratio through the channel. The vertical dashed line indicates the channel capacity, or Shannon limit, of AWGN-BPSK for information rate 1/2. Error bars denote 95\% confidence intervals.}
    \label{fig:classical zsz vs random dist}
\end{figure}

We also numerically benchmark the FEC performance of our classical ZSZ-2BGA codes under a standard additive white Gaussian noise (AWGN) channel with binary phase shift keying (BPSK) modulation, and decoded with GPU-accelerated sum-product BP \cite{sionna_github}; see the right panel of Figure \ref{fig:classical zsz vs random dist}. We observe a pseudothreshold near $1.25\,$dB, which is roughly within $1.1$\,dB of the AWGN-BPSK channel capacity $(E_b/N_0)_\star\approx 0.19\,$dB for information rate $R=1/2$. Notably, ZSZ-2BGA-310 (the left seed code of ZSZ-LP-775) performs nearly identically to the best girth-8 PEG instance from a 1000-code sample. Compared to an irregular-degree-optimized LDPC code under the fifth generation new radio (5G NR) specification \cite{Richardson_2018}, it has worse performance in the 5G code's waterfall regime $\lesssim4\,$dB but gains traction in its error-floor regime $\gtrsim4\,$dB. Such strong FER performance combined with group-algebra structure could make ZSZ-2BGA codes appealing for practical applications in classical error correction.

\subsection{Related works}

Guo et al. \cite{ZSZ_codes} introduced the ZSZ group for constructing quantum 2BGA codes, including their implementation in reconfigurable atom arrays, and is the primary motivator for this work. We present several significant improvements over this earlier work. First, the earlier codes had design rate zero, whereas our codes have design rate 1/5, which allows our memory efficiencies ($kd^2/n$) to be several times higher. The earlier work also employed random-coloration syndrome extraction for simulations, while we develop a greedy scheduler specifically tailored for optical-tweezer arrays as well as provide physical syndrome-extraction rearrangement times. Finally, our construction allows for easy construction and analysis of minimum-weight and symplectic logical operator bases, which would be difficult in the previous construction due to the zero design rate. Willenborg et al. \cite{Willenborg_2025_dihedral} studied non-LDPC lifted product codes over dihedral groups, which are ZSZ groups with $\ell_2=2$ and $q=-1$. Unfortunately, when restricting the classical seed codes to LDPC 2BGA form, we show in Corollary \ref{cor:dihedral dichotomy} that dihedral groups lend to either poor distance or poor girth.

Other complementary works on constant-rate qLDPC codes targeting the gigaquop and teraquop regimes include quasi-cyclic lifted product codes with rate $\approx29\%$ \cite{cain2026shor} and non-product codes based on affine permutation matrices (APM) with rate $\approx50\%$ \cite{zhao2026ultra}. On one hand, the qLDPC codes in both these works require physical qubit counts in the thousands, whereas our codes (other than ZSZ-LP-775) stay under a thousand. On the other hand, their logical cycle speeds in reconfigurable atom arrays are about an order of magnitude faster than ours, due to their efficient co-design using simple array rearrangements, e.g. cyclic translations, whereas our codes are mostly optimized for pure decoding performance and currently require increasing rearrangement overhead with increasing block length. We anticipate that the time overheads for syndrome extraction can be further improved in future work.

On the classical coding side, the closest LDPC precedents to our ZSZ-2BGA codes are the Margulis codes \cite{Margulis_1982} and their Ramanujan-graph variants by Rosenthal and Vontobel \cite{Rosenthal_Vontobel_2000}: $(3,6)$-regular LDPC codes over $\F_2[{\rm SL}_2(\F_p)]$ (and ${\rm PGL}_2(\F_q)$) whose parity-check matrices take a similar two-block form. Relatedly, Bazzi and Mitter proved that non-LDPC two-block codes over dihedral group algebras with random \emph{dense} generators are asymptotically good \cite{Bazzi_Mitter_2006}, lending theoretical support to non-abelian two-block constructions; to our knowledge, sparse two-block codes over metacyclic (ZSZ) group algebras have not been previously studied classically.

\textbf{Concurrent works.} While preparing this manuscript, we
became aware of concurrent and independent works \cite{mitten_codes, zheng2026lp} with similar results and have coordinated a joint public release. The first concurrent work \cite{mitten_codes} also studies qLDPC codes with the same five-block structure. While our work focuses on code construction and decoding, their work includes a full suite of logical gadgets for FTQC, and so we encourage readers interested in the practical FTQC application to also explore their manuscript. The second concurrent work \cite{zheng2026lp} explores FTQC with \emph{abelian} lifted product codes. While our codes are built from the non-abelian lifted product, it shares some similarity with the abelian case, and it would be interesting whether their abelian gadgets can be ported over to the non-abelian setting. In addition to these concurrent works, another work \cite{okada2026cpm} appeared that introduces \emph{non-product} qLDPC codes of similar block lengths and encoding efficiencies using circulant permutation matrices (CPM). For instance, their $\llbracket 472,122,14\rrbracket$ CPM code has comparable $kd^2/n\approx51$ with our $\llbracket 390,78,\leq16 \rrbracket$ ZSZ-LP code but with $(3,8)$-LDPC check matrices compared to our $(6,9)$-LDPC check matrices. Loosely speaking, non-product constructions allow for lower degrees and larger girths compared to product constructions at the same block length. At the same time, product constructions enable code properties, such as logical bases, to be inherited from smaller classical seed codes that are simpler to analyze. We anticipate the interplay between product and non-product constructions to be an exciting frontier in practical qLDPC code discovery.

%%%%%%%%%%%%%%%%%%%%%%%%%%%%%%%%%%%%%%%%%%%%%%%%%%%%%%%%%%

\section{Code construction}\label{sec:code construction}

\subsection{General 5-block form}

Let $A,B,C,D$ be $\ell\times\ell$ binary matrices with bounded row and column weights, i.e. LDPC. Our quantum LDPC codes have CSS parity-check matrices in block form:
\begin{subequations}\label{eq:5-block H_X,H_Z}
\begin{align}
    H_X &= \begin{pmatrix}\,
        A & 0 & B & 0 & C^\transp \\
        0 & A & 0 & B & D^\transp \,
    \end{pmatrix}  \\
    H_Z &= \begin{pmatrix}\,
        C & D & 0 & 0 & A^\transp \\
        0 & 0 & C & D & B^\transp \,
    \end{pmatrix} \, ,  \label{eq:5-block H_Z}
\end{align}
\end{subequations}
where `0' denotes the $\ell\times\ell$ zero matrix, and $A^\transp$ denotes the transpose of $A$. The CSS orthogonality condition $H^{}_X H^\transp_Z=0$ requires that
\begin{align}\label{eq:5-block CSS orth}
    H^{}_X H^\transp_Z = 
    \begin{pmatrix}
        AC^\transp+C^\transp A & BC^\transp+C^\transp B \\
        AD^\transp+D^\transp A & BD^\transp+D^\transp B
    \end{pmatrix} = 
    \begin{pmatrix}\,
        [A,C^\transp] & [B,C^\transp] \\
        [A,D^\transp] & [B,D^\transp] \,
    \end{pmatrix} = 0  \, ,
\end{align}
where in the second equality we used the fact that $+$ and $-$ are the same under binary arithmetic ($\F_2$). This condition is to ensure that the $X$-checks and $Z$-checks commute and can thus be simultaneously measured. If the CSS orthogonality condition is satisfied, then the usual CSS formula for the code dimension gives us
\begin{align}
    k = n - \rank{H_X} - \rank{H_Z} \geq 5\ell - 2\ell - 2\ell = \ell \, ,
\end{align}
or alternatively, $k/n \geq 1/5$; i.e. design rate 1/5.

The problem now boils down to finding binary matrices $A,B,C,D$ that satisfy
\begin{align}\label{eq:A,B,C,D comm conditions}
    [A,C^\transp] = [B,C^\transp] = [A,D^\transp] = [B,D^\transp] = 0  \, .
\end{align}

\subsection{ZSZ group primer}

One method to build such $A,B,C,D$ is with the help of a group algebra, where they will become matrix representations of group elements, and their commutator structure \eqref{eq:A,B,C,D comm conditions} will arise from group properties. Short for ``Z semidirect Z'', a \textbf{ZSZ} group with parameters ($\ell_1,\ell_2,q$) satisfying $q^{\ell_2} = 1\;\text{mod $\ell_1$}$, of size $\ell=\ell_1\ell_2$, is the semidirect product of two cyclic groups with the presentation:
\begin{align}\label{eq:ZSZ presentation}
    \mathbb{Z}_{\ell_1} \rtimes_q \mathbb{Z}_{\ell_2} := \big\langle x,y \;\big|\; x^{\ell_1} = y^{\ell_2} = yxy^{-1}x^{-q} = 1 \big\rangle \, .
\end{align}
The $q^{\ell_2} = 1\;\text{mod $\ell_1$}$ condition is to ensure that the group presentation is closed. In math literature, \eqref{eq:ZSZ presentation} is sometimes referred to as a metacyclic group.
Denote the ZSZ group by $\ZSZ_{\ell_1,\ell_2,q}$ or ZSZ for short when the parameters are implicit. $q$ is the non-abelian ``twist'' factor and determines how an element $x^a\in\mathbb{Z}_{\ell_1}$ changes when conjugated by $y\in\mathbb{Z}_{\ell_2}$. When $q=1$, the twist is trivial, and the presentation reduces to that of the abelian bivariate bicycle (BB) group $\mathbb{Z}_{\ell_1} \times \mathbb{Z}_{\ell_2}$. There are several theoretical and practical motivations for considering the ZSZ group:
\begin{enumerate}
    \item As we will see later, choosing an abelian group for the 5-block construction \eqref{eq:5-block H_X,H_Z} results in poor minimum distance; i.e. constant independent of block length.
    \item The ZSZ groups have more flexible choices in size $\ell=\ell_1\ell_2$ compared to other non-abelian groups, which allows greater flexibility in code construction.
    \item The ZSZ groups have a natural $\ell_1\times\ell_2$ rectangular embedding suitable for systems with movable qubits such as reconfigurable atom arrays \cite{ZSZ_codes}.
\end{enumerate}

A useful way to organize the data of the ZSZ group is to prescribe a canonical form $x^ay^b$ where the $x$ component appears to the left of the $y$ component. Using this canonical form, the twist $q$ induces the following ``push-through'' relations:
\begin{subequations}\label{eq:push-through relations}
\begin{align}
    y^\beta(x^iy^j) &= x^{q^\beta i} y^{j+\beta} \label{eq:push-through y left}  \\
    (x^i y^j)x^\alpha &= x^{i+q^j\alpha} y^j  \label{eq:push-through x right}  \, .
\end{align}
\end{subequations}
The left and right Cayley graphs of the group give a useful geometric picture of these push-through relations; see Figure \ref{fig:ZSZ cayley graphs}.

\begin{figure}[t]
    \centering
    \resizebox{\textwidth}{!}{%
    \begin{tikzpicture}[
      x=0.92cm, y=0.86cm,
      xedge/.style={draw=blue!80, line width=0.5pt, opacity=0.75},
      yedge/.style={draw=red!80, line width=0.5pt, opacity=0.75},
      axlab/.style={font=\footnotesize},
      ptitle/.style={font=\bfseries},
    ]

    %================= LEFT CAYLEY GRAPH =================
    \begin{scope}
      % x-edges: untwisted horizontal shift  x^i y^j -> x^{i+1} y^j
      \foreach \j in {0,...,5}{
        \foreach \i in {0,...,7}{ \draw[xedge] (\i,\j) -- ({\i+1},\j); }
      }
      % y-edges: twisted  x^i y^j -> x^{q i} y^{j+1},  here q=2
      \foreach \j in {0,...,4}{
        \foreach \i in {0,...,8}{
          \pgfmathtruncatemacro{\ti}{mod(2*\i,9)}
          \draw[yedge] (\i,\j) -- (\ti,{\j+1});
        }
      }
      % vertices
      \foreach \j in {0,...,5}{ \foreach \i in {0,...,8}{ \fill (\i,\j) circle (2pt); } }
      % axes
      \foreach \i/\t in {0/1,1/x,2/{x^2},3/{x^3},4/{x^4},5/{x^5},6/{x^6},7/{x^7},8/{x^8}}
        \node[axlab] at (\i,-0.55) {$\t$};
      \foreach \j/\t in {0/1,1/y,2/{y^2},3/{y^3},4/{y^4},5/{y^5}}
        \node[axlab] at (-0.6,\j) {$\t$};
      \node[ptitle] at (4,-1.5) {Left Cayley graph};
    \end{scope}

    %================= RIGHT CAYLEY GRAPH =================
    \begin{scope}[xshift=10.5cm]
      % y-edges: untwisted vertical shift  x^i y^j -> x^i y^{j+1}
      \foreach \i in {0,...,8}{
        \foreach \j in {0,...,4}{ \draw[yedge] (\i,\j) -- (\i,{\j+1}); }
      }
      % x-edges: twisted, row-dependent step x^i y^j -> x^{i+q^j} y^j (one representative per row)
      \foreach \j/\s/\b in {0/1/30,1/2/28,2/4/22,3/8/12,4/7/14,5/5/18}
        \draw[xedge] (0,\j) to[bend left=\b] (\s,\j);
      % vertices
      \foreach \j in {0,...,5}{ \foreach \i in {0,...,8}{ \fill (\i,\j) circle (2pt); } }
      % axes
      \foreach \i/\t in {0/1,1/x,2/{x^2},3/{x^3},4/{x^4},5/{x^5},6/{x^6},7/{x^7},8/{x^8}}
        \node[axlab] at (\i,-0.55) {$\t$};
      \foreach \j/\t in {0/1,1/y,2/{y^2},3/{y^3},4/{y^4},5/{y^5}}
        \node[axlab] at (-0.6,\j) {$\t$};
      \node[ptitle] at (4,-1.5) {Right Cayley graph};
    \end{scope}

    % legend
    \draw[blue!80, line width=1.4pt] (6.7,-2.3) -- (7.4,-2.3);
    \node[right, font=\footnotesize] at (7.4,-2.3) {$x$ edges};
    \draw[red!80, line width=1.4pt] (10.2,-2.3) -- (10.9,-2.3);
    \node[right, font=\footnotesize] at (10.9,-2.3) {$y$ edges};

    \end{tikzpicture}%
    }
    \caption{The left and right Cayley graphs of the ZSZ group $\ZSZ_{9,6,2}$ drawn on the $\ell_1\times\ell_2$ grid of group elements $x^iy^j$ (black dots), with the $x$-exponent $i$ as the horizontal coordinate and the $y$-exponent $j$ as the vertical coordinate. Edges denote multiplication by the generators $x$ (blue) and $y$ (red). \textbf{Left:} The left Cayley graph. Left multiplication by $x$ is an untwisted horizontal shift $x^iy^j\mapsto x^{i+1}y^j$, while left multiplication by $y$ is twisted $x^iy^j\mapsto x^{qi}y^{j+1}$ via \eqref{eq:push-through y left}, scaling the $x$-coordinate by $q$ on each upward step. \textbf{Right:} The right Cayley graph. Here right multiplication by $y$ is the untwisted vertical shift $x^iy^j\mapsto x^iy^{j+1}$, whereas right multiplication by $x$ carries a twisted, row-dependent horizontal shift $x^iy^j\mapsto x^{i+q^j}y^j$ via \eqref{eq:push-through x right}; for visual clarity only one representative $x$-edge is drawn per row, the others being its horizontal translates. Edges that wrap around the periodic boundaries are omitted.}
\label{fig:ZSZ cayley graphs}
\end{figure}
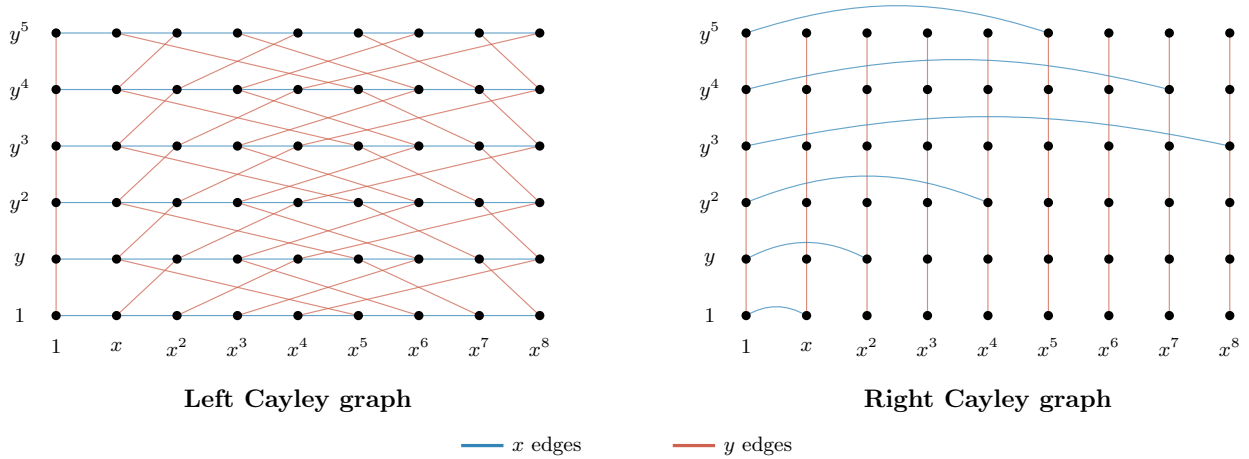

To get permutation matrices from the group elements, we use their left-regular and right-regular representations, which essentially encode the group's left and right Cayley multiplication tables. The left matrices $L[\cdot]$ encode the action of left multiplication on group elements, while the right matrices $R[\cdot]$ encode right multiplication. Since group multiplication is associative, e.g. $a(gb)=(ag)b$, all $L[\cdot]$ matrices commute with all $R[\cdot]$ matrices. Define the $m\times m$ cyclic-shift matrix $S_m$ and affine permutation matrix $T^{(q)}_m$ by
\begin{align}
    S_m\ket{i}=\ket{i+1}\, ,\qquad T^{(q)}_m\ket{i}=\ket{qi}
\end{align}
on basis vectors $\ket{i}$, where the indices are understood modulo $m$. Equivalently in terms of matrix elements, $\big(S_{m}\big)_{i+1,i}=1$ and $\big(T^{(q)}_{m}\big)_{qi,i}=1$. Then the left-regular representation of $xy\in\ZSZ$, using the push-through relation \eqref{eq:push-through y left}, takes the form
\begin{align}\label{eq:L[xy] tensor decomp}
    L[xy] = \big(S^{}_{\ell_1}T^{(q)}_{\ell_1}\big) \otimes S^{}_{\ell_2}  \, .
\end{align}
Note that with the trivial twist $q=1$, we have $T^{(1)} = \ident$ and $L[xy]$ reduces to the circulant matrices used in the BB codes \cite{BB_codes}. Let $\Pi_j=\ket{j}\!\bra{j}$ denote the projector onto the $j$th $y$-coordinate, with $j$ understood modulo $\ell_2$. The right-regular representation $R[xy]$ induces $y$-exponent-dependent $x$-translations \eqref{eq:push-through x right}, i.e. a horizontal shear, given by
\begin{align}\label{eq:R[xy] decomp}
    R[xy] = \sum_{j=0}^{\ell_2-1} S_{\ell_1}^{q^j} \otimes S_{\ell_2}\Pi_j \, .
\end{align}

\subsection{Classical two-block ZSZ codes}

The first step in building our quantum LDPC codes is to build a pair of classical LDPC codes. Let $a,b,c,d \in \F_2[\ZSZ]$ be trinomials in the ZSZ group algebra; in other words, they are each sums of three group elements or monomials such as $1+x+xy$.
Define the pair of classical parity-check matrices
\begin{align}\label{eq:H_left,H_right}
    H_{\rm left} = \begin{pmatrix}\,
        L[a] & L[b] \,
    \end{pmatrix} \;,\qquad
    H_{\rm right} = \begin{pmatrix}\,
        R[c] & R[d] \,
    \end{pmatrix} \, .
\end{align}
$H_{\rm left}$ and $H_{\rm right}$ are classical two-block group-algebra (2BGA) codes using the left-regular and right-regular representations respectively. We will append the prefix ``left'' for properties of $H_{\rm left}$ and ``right'' for $H_{\rm right}$; e.g. left codeword, right distance. Since our ZSZ group has size $\ell$, every $L[\cdot],R[\cdot]$ has dimensions $\ell\times\ell$, and so $H_{\rm left}, H_{\rm right}$ both have dimension $\ell\times2\ell$. Hence they both have design rate 1/2, with the true rate being higher if they are rank-deficient. Since $a,b,c,d$ are trinomials, their regular representations are all $(3,3)$-LDPC, and so $H_{\rm left}, H_{\rm right}$ are both $(3,6)$-LDPC. Without loss of generality, we can always normalize the first monomial to be 1 in each of $a,b,c,d$, which corresponds to a relabeling of the columns, i.e. an origin shift, of $H_{\rm left}$ and $H_{\rm right}$.

Notice that $H_{\rm left}$ is built solely from the left-regular representation of the group algebra. Recall that all left matrices $L[\cdot]$ commute with all right matrices $R[\cdot]$ from group associativity. In other words, $H_{\rm left}$ has a natural, free right action by the ZSZ group:
\begin{align}\label{eq:H_left symmetry}
    R[g]\, H_{\rm left} = \big(\,R[g]L[a]\;\;\,R[g]L[b]\,\big) = \big(\,L[a]R[g]\;\;\,L[b]R[g]\,\big) = H_{\rm left}\, R[g]^{\oplus2}
\end{align}
for all $g\in\ZSZ$, where $R[g]^{\oplus2} = \mathrm{diag}(R[g],R[g])$. Conversely, $H_{\rm right}$ has a natural, free left action by the ZSZ group. These symmetries will be useful both for the quantum code construction in Section \ref{sec:balanced product} as well as for constructing equivariant codeword bases and logical Pauli operators in Section \ref{sec:logical basis}.

One may be wondering about the particular choices of \emph{(i)} a non-abelian group such as ZSZ, and \emph{(ii)} having $a,b,c,d$ be weight-3 trinomials. The choice of a non-abelian group comes from the following simple observation. Suppose the underlying group was abelian. Then define the following matrix $G_{\rm left} = (L[b]^\transp\;\;L[a]^\transp)$ which satisfies $H^{}_{\rm left} G^\transp_{\rm left} = L[a]L[b] + L[b]L[a] = L[ab+ba] = 0$ since the group is abelian. Hence, every row of $G_{\rm left}$ is a left codeword, which implies that the left distance $d_{\rm left} \leq \wt(a) + \wt(b) = 6$. Analogously, we will have $d_{\rm right} \leq 6$. Readers familiar with quantum 2BGA codes \cite{2BGA_codes} may notice that $H_{\rm left}\sim H_X$, and $G_{\rm left}$ is nothing more than the corresponding $H_Z$. In the language of abstract algebra, $G_{\rm left}$ is sometimes called a syzygy with respect to $H_{\rm left}$. So if we desire minimum distances for \eqref{eq:H_left,H_right} that are not constant with respect to $\ell$, we must choose a non-abelian group. See Appendix \ref{app:classical 2BGA codes} for more details and formal statements on classical 2BGA codes, including distance upper bounds for ZSZ groups.

The second choice of trinomials comes from the following observations. First, from the perspective of fault tolerance, we want the weights of $a,b,c,d$ to be small because ultimately these classical codes will be used to build the quantum codes, and the quantum check weights will be related to the classical check weights. Smaller check weights typically lead to higher thresholds under circuit-level noise and so are desirable from a practical standpoint. Smaller check weights also typically lead to Tanner graphs with fewer short cycles which are desirable from a BP-decoding standpoint. Now, suppose one of the $a,b$ in $H_{\rm left}$ has unit weight, which without loss of generality we can assume $a=1$. Then $H_{\rm left}$ has the block-form $H_{\rm left} = (\ident\;\;B)$ where $B=L[b]$. This block form is essentially a systematic form for $H_{\rm left}$, and so we can identify $G_{\rm left} = (B^\transp\;\;\ident)$ as its systematic generator matrix, which satisfies $H^{}_{\rm left} G^\transp_{\rm left} = B+B=0$ over $\F_2$. Again, we see that the code distance is constant of $\ell$. Now we move on to the case where $\wt(a)=2$. In this case, $A=L[a]$ is $(2,2)$-LDPC and is a cycle code on a graph that is composed of disjoint ring subgraphs, i.e. direct sums of cyclic repetition codes. On each ring subgraph, any even number of unsatisfied parity checks can become satisfied through perfect matching. Hence, a potentially low-weight codeword can be created through the following manner: spawn low-weight errors on the $B$ subblock with even-weight syndromes which can then be cancelled by matching on the $A$ subblock. This argument is not rigorous but gives an intuition for why $\wt(a)=2$ can be bad. In a numerical search with random $a,b$, we corroborate this intuition and find that the code distances are rather poor when $\wt(a)=2$.

\begin{table}[t]
\centering\renewcommand{\arraystretch}{1.1}
\begin{tabular}{ccccccc}
\toprule
\textbf{Code name} & $[n,k,d]_{\rm left}$ & $[n,k,d]_{\rm right}$ & \textbf{\makecell{left \\ girth}} & \textbf{\makecell{right \\ girth}} & \textbf{\makecell{OGM \\ $\mathbf{c}_{\rm left}, \mathbf{c}_{\rm right}$}} & \textbf{\makecell{OGS \\ $G_{\rm left}, G_{\rm right}$}} \\ \midrule
ZSZ-2BGA-64 & $[64,32,10]$       & $[64,32,10]$       & 6 & 6 & \cmark, \cmark & \cmark, \cmark \\
ZSZ-2BGA-96 & $[96,48,12]$       & $[96,48,12]$       & 6 & 6 & \cmark, \cmark & \cmark, \cmark \\
ZSZ-2BGA-128 & $[128,64,14]$      & $[128,64,14]$      & 6 & 6 & \cmark, \cmark & \cmark, \cmark \\
ZSZ-2BGA-156 & $[156,78,16]$      & $[156,78,16]$      & 6 & 6 & \cmark, \cmark & \cmark, \cmark \\
ZSZ-2BGA-220 & $[220,110,\leq18]$ & $[220,110,\leq18]$ & 8 & 8 & \cmark, \cmark & \xmark, \xmark \\
ZSZ-2BGA-310 & $[310,155,\leq22]$ & $[310,155,\leq22]$ & 8 & 8 & $140,129$ & \xmark, \xmark \\ \bottomrule
\end{tabular}
\caption{The parameters of the classical ZSZ-2BGA codes $H_{\rm left}, H_{\rm right}$ for every candidate ZSZ-LP code are listed as well as the girths of their corresponding Tanner graphs. Every code is $(3,6)$-regular with rate 1/2. The trinomials can be found in Table \ref{tab:ZSZ-LP polynomials}. Exact and upper bounds on code distances are computed using the \textsf{pySATDist} and \textsf{QDistEvol} ($10^5$ iterations) methods respectively \cite{webster2026distance}. The last column indicates whether the classical codes admit one-generator-systematic (OGS) generator matrices.}
\label{tab:classical ZSZ code params}
\end{table}

Another important design characteristic for classical LDPC codes is the girth of the Tanner graph. It is well known that the performance of message-passing decoder such as BP is sensitive to the girth of the Tanner graph. Short cycles introduce statistical dependencies among the messages and are a well-known source of error floors \cite{MacKay_Postol_2003, ModernCodingTheory}. A large girth ensures that no small cycles exist in the Tanner graph that can degrade a BP decoder. The density of short cycles can depend sensitively on the ZSZ group parameters as well as choice of trinomials for $a$ and $b$. Fortunately, the upper bounds on 2BGA Tanner girths from ZSZ groups seem to be fairly generous and provide some heuristic guidelines for constructing Tanner graphs with large girth; see Appendix \ref{app:classical tanner girth} for more details. In addition to the distance-6 obstruction above, it is also known that the Tanner girth is at most 6 for abelian 2BGA codes \cite[Theorem B.3]{ZSZ_codes}.\footnote{\cite[Theorem B.3]{ZSZ_codes} computes girth over variable-variable or check-check adjacency graphs, and hence the Tanner girth is twice the girths that they report.}

We can generate many classical (left) ZSZ-2BGA codes by randomly sampling trinomials for $a$ and $b$. To this end, we randomly sample 1000 instances for several ZSZ groups and estimate their minimum distances using 10,000 iterations of \textsf{QDistEvol} \cite{webster2026distance}.
Figure \ref{fig:classical zsz vs random dist} displays the upper 95th percentiles of estimated distances for three random ensembles of $(3,6)$-LDPC codes: left ZSZ-2BGA, the bipartite configuration model \cite{Bender_1978, Bollobas_1980}, and the Progressive Edge Growth algorithm \cite{PEG_LDPC}. The bipartite configuration model constructs a random Tanner graph by assigning empty sockets on the bit and check nodes according to prescribed degrees $(\Delta_B,\Delta_C)$ and then pairing up sockets to form edges via a random permutation. Theoretically, the configuration model has been proven to generate asymptotically good (constant rate and linear distance) LDPC codes with high probability when $\Delta_C>\Delta_B$ (for constant rate) and $\Delta_B\geq3$ (for linear distance) \cite{ModernCodingTheory}. In Progressive Edge Growth, edges are added progressively in a selected manner that avoids short cycles, and it is one of the leading methods to construct well-performing LDPC codes in practice. Remarkably, we find that the random ZSZ-2BGA codes statistically produce $(3,6)$-LDPC codes with (estimated) minimum distances that exceed those of both the configuration model and Progressive Edge Growth at block lengths $n \in \{64, 96, 128, 156, 220, 310\}$, the lengths that we sample for eventually building the quantum codes. The left side of Figure \ref{fig:classical zsz vs random dist} shows a scatter plot of the 95th percentile estimated distance of 1000 code instances from all three LDPC ensembles at the aforementioned code lengths.

\subsection{Quantum ZSZ-LP codes from the balanced product}
\label{sec:balanced product}

After choosing a pair of $H_{\rm left}, H_{\rm right} \in \F^{\ell\times2\ell}_2$ corresponding to classical ZSZ-2BGA codes, we take their balanced product \cite{balanced_product} to obtain quantum LDPC codes with CSS parity-check matrices
\begin{subequations}\label{eq:H_X,H_Z balanced product}
\begin{align}
    H_X &= \big(\, H^{}_{\rm left} \otimes \ident^{}_2 \;\big|\; \ident^{}_1 \otimes H^\transp_{\rm right} \,\big)  \\
    H_Z &= \big(\, \ident_2 \otimes H^{}_{\rm right} \;\big|\; H^\transp_{\rm left} \otimes \ident^{}_1 \,\big) \, ,
\end{align}
\end{subequations}
where $\ident$ denotes the $j\times j$ identity matrix. Plugging in our definitions of $H_{\rm left},H_{\rm right}$ \eqref{eq:H_left,H_right} into \eqref{eq:H_X,H_Z balanced product}, we arrive at our 5-block form \eqref{eq:5-block H_X,H_Z} with
\begin{align}
    A=L[a] \;,\qquad B=L[b] \;,\qquad C=R[c] \;,\qquad D=R[d] \, .
\end{align}
One can quickly verify that the CSS orthogonality conditions \eqref{eq:5-block CSS orth} are satisfied since left and right representations commute.
The balanced product first involves taking the ordinary hypergraph product \cite{HGP}, which produces an algebraic complex that lives in the Cartesian product space $\ZSZ\times\ZSZ$, which scales in size by $\ell^2$. The balanced product then quotients out a redundant symmetry corresponding to the diagonal subgroup $\mathcal{D} = \{ (g,g)\;|\; g\in\ZSZ \}$ of size $\ell$; in other words it enforces the equivalence relation $v\cdot g\otimes w \sim v\otimes g\cdot w$ that precisely corresponds to the free right and left actions of the ZSZ group on $H_{\rm left}$ and $H_{\rm right}$ respectively \eqref{eq:H_left symmetry}. The size of the remaining quotient space thus scales as $|\ZSZ\times\ZSZ|/|\mathcal{D}| \sim \ell^2/\ell = \ell$, which matches the dimension scaling of \eqref{eq:H_X,H_Z balanced product}. Geometrically, a balanced product code lives on a square Cayley complex that can be thought of as a marriage of the left and right Cayley graph, whose squares correspond precisely to commuting left and right actions.

Alternatively, one can interpret \eqref{eq:H_X,H_Z balanced product} as first taking the hypergraph product of $H_{\rm left}$ and $H_{\rm right}$ but over their $1\times2$ ZSZ group algebra presentations rather than their full binary presentations, which essentially results in the all-lower-case version of \eqref{eq:5-block H_X,H_Z}, with transposes replaced by the antipode map that inverts all group elements. The resulting group-algebra matrix is then ``lifted'' to binary by replacing $a,b$ with their left-regular representations and $c,d$ with the right-regular representations, a construction known as the non-abelian lifted product \cite{PK_good_qLDPC}. To not confuse abbreviations with belief propagation (BP), we call these codes ZSZ lifted product codes, or ZSZ-LP codes for short. See Table \ref{tab:ZSZ-LP polynomials} for a list of the trinomials $a,b,c,d$ used to build each candidate ZSZ-LP code. Note that $H_{\rm left}$ and $H_{\rm right}$ can be inferred from $a,b$ and $c,d$ respectively for each code.

\begin{table}[t]
\centering\renewcommand{\arraystretch}{1.3}
\resizebox{\textwidth}{!}{
\begin{tabular}{cccccc}
\toprule
\textbf{Code name} & $\ell_1,\ell_2,q$ & $a$ & $b$ & $c$ & $d$ \\ \midrule
ZSZ-LP-160 & $16,2,9$ & $1+x^{13}+x^{14}y$ & $1+x^3y+x^{13}y$ & $1+x^{14}+x^{12}y$ & $1+x^7+x^{11}$ \\
ZSZ-LP-240 & $12,4,7$ & $1+x^8y+x^7y^2$ & $1+x^2+x^5y$ & $1+x^4y+x^2y^3$ & $1+x^4+x^9$ \\
ZSZ-LP-320 & $16,4,3$ & $1+x+y$ & $1+x^2+x^{13}y^3$ & $1+x^{12}y+x^2y^3$ & $1+x^{12}+x^{11}y^2$ \\
ZSZ-LP-390 & $26,3,3$ & $1+x^{17}y+x^{14}y^2$ & $1+x^{16}+x^3y$ & $1+x^{24}+x^8y$ & $1+x^{21}+y$ \\
ZSZ-LP-550 & $22,5,3$ & $1+x^{10}+x^5y^2$ & $1+x^8+x^2y^3$ & $1+x^7+x^{16}y$ & $1+x+x^6y^4$ \\
ZSZ-LP-775 & $31,5,2$ & $1+x^{26}y^2+x^{18}y^3$ & $1+x^{13}y^2+x^7y^4$ & $1+xy+x^{28}y^4$ & $1+y+x^{22}y^2$ \\ \bottomrule
\end{tabular}
}
\caption{The ZSZ group parameters $(\ell_1,\ell_2,q)$ and the trinomials $a,b,c,d$ are listed for every candidate ZSZ-LP code.}
\label{tab:ZSZ-LP polynomials}
\end{table}

If $H_{\rm left}$ and $H_{\rm right}$ have full rank, then so do $H_X$ and $H_Z$ from \eqref{eq:H_X,H_Z balanced product}. To see why, observe that $\rank H_X \geq \rank (H_{\rm left}\otimes \ident_2) = 2\rank H_{\rm left}$ and likewise $\rank H_Z \geq 2\rank H_{\rm right}$. At the same time, the number of rows of $H_X$ ($H_Z$) is twice that of $H_{\rm left}$ ($H_{\rm right}$), which bounds its maximum rank. When $H_{\rm left}$ and $H_{\rm right}$ have full rank, then these inequalities coincide, leading to $\rank H_X = 2\rank H_{\rm left}$ and $\rank H_Z = 2\rank H_{\rm right}$. The usual CSS formula for the code dimension then gives us
\begin{align}
    k &= n - \rank{H_X} - \rank{H_Z}  \notag \\
    &= n - 2\rank{H_{\rm left}} - 2\rank{H_{\rm right}}  \notag \\
    &= 5\ell - 2\ell - 2\ell  \notag \\
    &= \ell = k_{\rm left} = k_{\rm right} \, .
\end{align}
In other words, when the classical parity-check matrices have full rank, then the number of logical qubits in the quantum code is the same as the number of logical bits in each classical seed code.

For the lifted and balanced products in general, there is not really a relation between the minimum distances of the classical seed codes and the quantum code like for HGP codes \cite{HGP}, and the quantum distance typically has to be numerically estimated. Nonetheless, for our simple 5-block case \eqref{eq:5-block H_X,H_Z}, we can still derive some useful upper bounds in the case where the underlying group is abelian, essentially extending our low-weight syzygy argument ($d\leq6$) from the classical codes to the quantum codes; see Appendix \ref{app:abelian-low-weight-logicals} for formal details. At first, it may seem like the low-weight syzygies that plagued the classical abelian 2BGA codes could potentially be ``stabilized away'' in the quantum code by becoming stabilizers so that they are not low-weight logical operators. Indeed, this is exactly what happens for quantum 2BGA codes like BB codes: the low-weight syzygy matrix of $H_X=(A\;\,B)$ is precisely $H_Z=(B^\transp\;\,A^\transp)$ and so is stabilized away. Unfortunately, for the 5-block form \eqref{eq:5-block H_X,H_Z}, we show that not all abelian low-weight syzygies can be stabilized away, which results in a low quantum distance. See Appendix \ref{app:abelian-low-weight-logicals} for the formal statements. This distance upper bound is our main motivation for choosing the ZSZ group over an abelian group. One may also wonder why we choose different trinomials for $H_{\rm left}$ and $H_{\rm right}$. It turns out that using the same trinomials, i.e. $a=c$ and $b=d$, results in unavoidable 4-cycles in the Tanner graphs and hence minimizes the Tanner girth to be 4; see Theorem \ref{thm:symmetric-girth-4} for the formal details.

Quite remarkably, our candidate ZSZ-LP codes in Table \ref{tab:ZSZ-LP codes intro} seem to have no apparent loss in distance compared to their classical seed codes in Table \ref{tab:classical ZSZ code params}. This feature is surprising for two reasons: \emph{(i)} appending additional block columns to a parity-check matrix like in \eqref{eq:H_X,H_Z balanced product} usually leads to a lower distance since unsatisfied checks now have additional degrees of freedom to become satisfied; and \emph{(ii)} the logical operators which resemble minimum-weight codewords of the classical codes can now be deformed by stabilizers of the same Pauli type. We note that this distance preservation is extremely special to the candidate ZSZ-LP codes; a typical random ZSZ-LP code has a reduced distance compared to its classical seed codes.

\subsection{Code selection through filtered random search}

We now describe the procedure used to select the trinomials $a,b,c,d$ in Table \ref{tab:ZSZ-LP polynomials} and equivalently $H_{\rm left}$ and $H_{\rm right}$ in Table \ref{tab:classical ZSZ code params}. We first fix a ZSZ group by selecting $\ell_1,\ell_2,q$ that satisfy $q^{\ell_2} = 1\;\text{mod $\ell_1$}$. We choose the ZSZ group parameters and classical 2BGA structure according to both theoretical and practical considerations, which can serve as heuristic guidelines for future code searches:
\begin{itemize}
    \item We favor $\ell_1\gg\ell_2$ to prevent a commutator-subgroup distance bound $d\leq 6\ell_1/\gcd(\ell_1,q-1)$ from becoming small compared with the code length; see Appendix \ref{app:ZSZ distance bounds} and Corollary \ref{cor:ZSZ commutator distance bound} for details.
    \item We choose different trinomials for $H_{\rm left}$ and $H_{\rm right}$. The symmetric choice $a=c$ and $b=d$ creates unavoidable $4$-cycles that may degrade BP decoding; see Appendix \ref{app:tanner graph cycles} and Theorem \ref{thm:symmetric-girth-4} for details.
\end{itemize}
We randomly generate four trinomials for $a,b,c,d$ and apply a cascading series of filters, as illustrated in Figure \ref{fig:code-finding-flowchart}. The filters are either enforcing a minimum girth or a minimum distance. The filters are broadly ordered from least costly to most costly in time complexity for search efficiency. For example, computing the girth of a graph via breadth-first-search is much quicker than estimating a minimum distance of a code via \textsf{QDistEvol}, and computing classical quantities is generally faster than computing quantum ones.

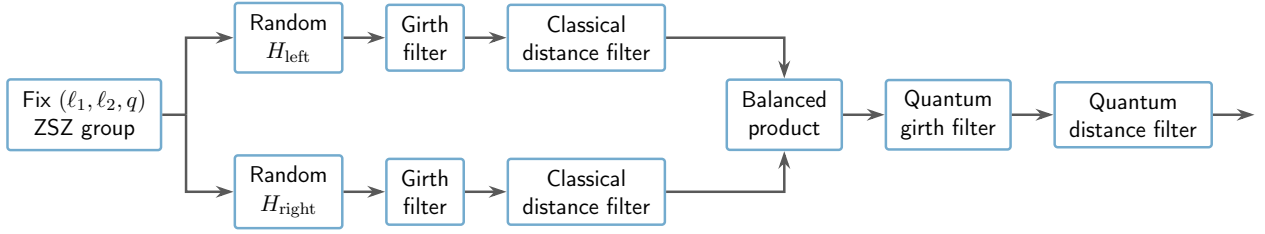
\begin{figure}[t]
    \centering
    \resizebox{\textwidth}{!}{%
    \begin{tikzpicture}[
      font=\sffamily,
      node distance=9mm and 7mm,
      box/.style={
        draw=blue!55,
        very thick,
        rounded corners=1.5pt,
        align=center,
        inner xsep=7pt,
        inner ysep=5pt,
        minimum height=10mm,
        fill=white
      },
      randbox/.style={box},
      girthbox/.style={box},
      distbox/.style={box},
      flowline/.style={
        very thick,
        draw=black!65
      },
      arrow/.style={
        -{Stealth[length=3mm,width=2mm]},
        very thick,
        draw=black!65
      },
    ]
    
    \def\rowsep{13mm}
    \def\colsep{7mm}
    
    % Main input
    \node[box] (fix) {Fix $(\ell_1,\ell_2,q)$ \\ ZSZ group};
    
    % Two horizontal branches, with aligned columns
    \coordinate (split) at ($(fix.east)+(4mm,0)$);
    \coordinate (branchstart) at ($(fix.east)+(12mm,0)$);
    
    \node[randbox, anchor=west] (left0) at ($(branchstart)+(0,\rowsep)$)
      {Random \\ $H_{\rm left}$};
    \node[randbox, anchor=west] (right0) at ($(branchstart)+(0,-\rowsep)$)
      {Random \\ $H_{\rm right}$};
    
    \node[girthbox, anchor=west] (left1) at ($(left0.east)+(\colsep,0)$)
      {Girth \\ filter};
    \node[girthbox, anchor=west] (right1) at (left1.west |- right0)
      {Girth \\ filter};
    
    \node[distbox, anchor=west] (left2) at ($(left1.east)+(\colsep,0)$)
      {Classical \\ distance filter};
    \node[distbox, anchor=west] (right2) at (left2.west |- right1)
      {Classical \\ distance filter};
    
    % Merge and final horizontal pipeline
    \node[box, anchor=west] (balanced)
      at ($(left2.east)!0.5!(right2.east)+(10mm,0)$)
      {Balanced \\ product};
    
    \node[box, right=of balanced] (qgirth) {Quantum \\ girth filter};
    \node[box, right=of qgirth] (qdist) {Quantum \\ distance filter};
    
    % Arrows
    \draw[flowline] (fix.east) -- (split);
    \draw[arrow] (split) |- (left0.west);
    \draw[arrow] (split) |- (right0.west);
    
    \draw[arrow] (left0.east) -- (left1.west);
    \draw[arrow] (left1.east) -- (left2.west);
    
    \draw[arrow] (right0.east) -- (right1.west);
    \draw[arrow] (right1.east) -- (right2.west);
    
    \draw[arrow] (left2.east) -- ++(4mm,0) -| (balanced.north);
    \draw[arrow] (right2.east) -- ++(4mm,0) -| (balanced.south);
    
    \draw[arrow] (balanced.east) -- (qgirth.west);
    \draw[arrow] (qgirth.east) -- (qdist.west);
    \draw[arrow] (qdist.east) -- ++(7mm,0);
    
    \end{tikzpicture}%
    }
    \caption{Filter pipeline for candidate ZSZ-LP code discovery.}
\label{fig:code-finding-flowchart}
\end{figure}

The first filters act on the classical seed codes $H_{\rm left}$ and $H_{\rm right}$. For every random left and right seed code, we compute the girth of its corresponding Tanner graph and only keep those whose Tanner girths exceed a minimum threshold, such as 6 or 8. From the survivors, we then estimate their minimum distances with \textsf{QDistEvol} and filter again by a minimum threshold. We find that sometimes, codes with a smaller girth will have much larger distances than those with the next larger girth; in these cases, we fallback to the smaller girth.
After the girth and distance filters, we then take pairwise balanced products between the left and right survivors to form a set of quantum ZSZ-LP codes. For each ZSZ-LP code in the set, we then filter by their $X$-Tanner and $Z$-Tanner girths. We also keep track of the number of minimum-length cycles at the girth for each code. We finally take the remaining survivors and estimate their $X$ and $Z$ distances with \textsf{QDistEvol}. Since this last process can be slow for large block lengths, we order the codes by their frequency of short cycles and progressively estimate the distances of codes with more and more short cycles until a code (or set of codes) exceeds a chosen minimum distance threshold. We then recompute the minimum distances of the final survivors with either an exact (SAT solver via \textsf{pySATDist}, which calls \textsf{PySAT} \cite{imms-sat18, itk-sat24}) or a large-scale estimator ($10^6$ iterations \textsf{QDistEvol}, amortized over multiple CPUs), and the winners then become our candidate ZSZ-LP codes.
For all candidate ZSZ-LP codes in Table \ref{tab:ZSZ-LP codes intro}, we find that an initial ensemble of a few thousand left and right seed codes is enough to produce candidate ZSZ-LP codes whose (exact or estimated) minimum distances are the same as those of their classical seed codes.

To benchmark the flexibility of the code discovery pipeline beyond the candidates in Table \ref{tab:ZSZ-LP codes intro}, we generate the two complementary ZSZ-LP ensembles shown in Figure \ref{fig:pareto frontiers}. Both broad searches begin by enumerating ZSZ groups with sizes between 12 and 200, corresponding to ZSZ-LP codes of lengths between 60 and 1000, guided by the theoretical no-go results in the Appendices to avoid certain low-distance or low-girth pitfalls. At the shortest lengths $n<100$, where the search is inexpensive, we exhaustively scan over all ZSZ group presentations. For each ZSZ group, we independently sample trinomial pairs for $H_{\rm left}$ and $H_{\rm right}$.

For the efficiency--length ensemble in the left panel, we bias the search towards large distance. We first sample candidate pools of left and right seed codes for every ZSZ group and estimate their classical distances with \textsf{QDistEvol}. We then rank them first by estimated classical distance and take pairwise balanced products from the best 16 left and 16 right seeds, from which we retain an ensemble of up to 256 ZSZ-LP codes per length. Importantly, we modify the code discovery pipeline to forgo the girth filters in order to maximize the potential distances. Quantum distances are then estimated with 20,000 iterations of \textsf{QDistEvol}. The near-frontier candidates are independently refined with additional iterations, and the plot displays only unique $\llbracket n,k,d\rrbracket$ points lying within 50\% of the observed frontier. Table \ref{tab:more ZSZ-LP codes} in Appendix \ref{app:more ZSZ-LP codes} contains several more ZSZ-LP codes along the observed Pareto front. 

For the cycle--length ensemble in the right panel, we reverse the optimization priority and retain the girth filters while dropping the distance filters. For each ZSZ action, we randomly sample 50,000 trinomial pairs per side, discard the classical seed codes with $4$-cycles, and retain the 1000 left and right pairs having the fewest $6$-cycles. We then screen their one million pairwise balanced products and keep only those for which both $H_X$ and $H_Z$ have Tanner girth at least 6; we note that we did not find any with girth 8. The best 1000 products are retained independently for each ZSZ group, ordered by the total number of $6$-cycles in $H_X$ and $H_Z$. In the plot, we display 30 codes per length to avoid too much clutter, sampled approximately uniformly among all associated ZSZ-LP codes. To give a sense of how much distance is sacrificed for fewer short cycles, we estimate their distances with 20,000 iterations of \textsf{QDistEvol}. From the scatter plot, we observe that the codes with the fewest short cycles have small distances; the codes with larger distances typically have more short cycles. Under BP decoding, the codes with fewer short cycles will likely have better ``waterfall'' performance close to threshold, but may behave poorly in the error-floor regime due to their small distances. Depending on the target logical error rate regime, physical error rate and BP decoder configuration, there could likely be a sweet spot between short-cycle density and minimum distance.

\begin{figure}[t]
    \centering
    \includegraphics[width=0.49\textwidth]{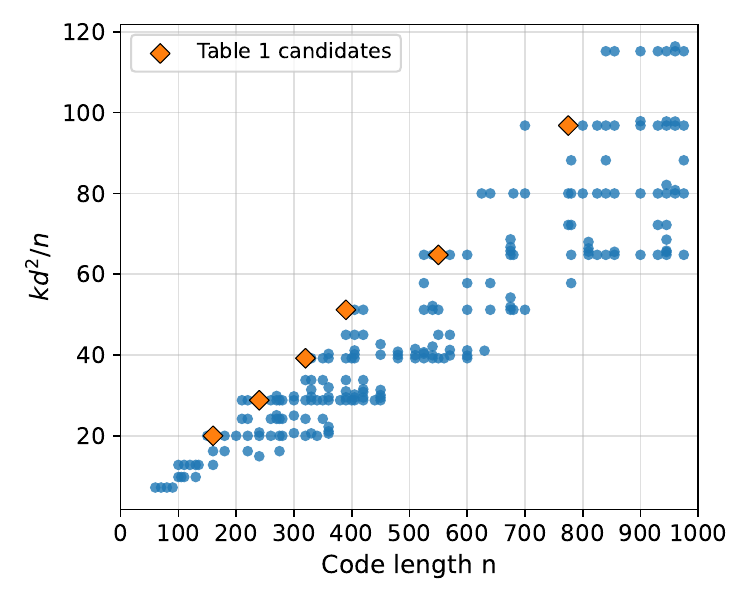} \hfill
    \includegraphics[width=0.5\textwidth]{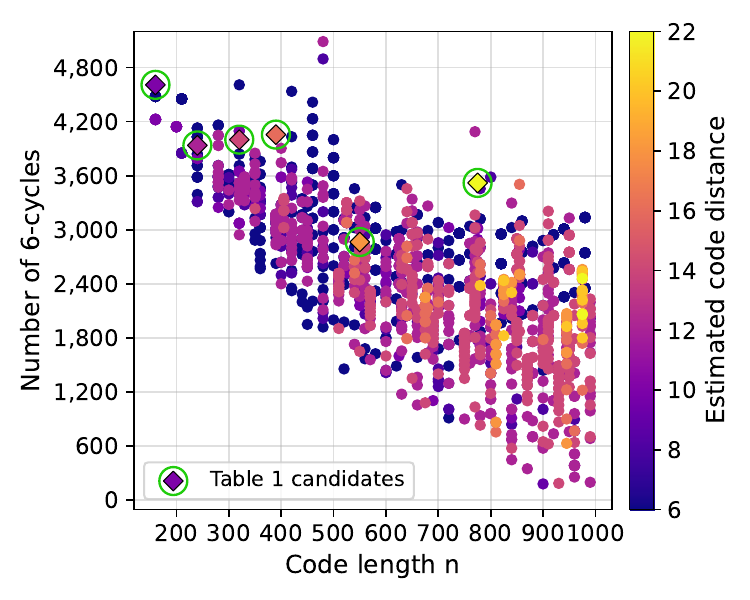}
    \caption{Pareto frontiers are plotted for two ZSZ-LP code ensembles, sampled using parts of the code discovery pipeline in Figure \ref{fig:code-finding-flowchart}. \textbf{Left:} The efficiency--length Pareto frontier. For this ensemble, we forgo the quantum girth filter and try to maximize the memory efficiency $kd^2/n$. To reduce clutter, we only plot the codes within 50\% of the observed Pareto front. All displayed code distances are estimated with at least 20,000 iterations of \textsf{QDistEvol}, with the ones near the frontier having at least 100,000 iterations. \textbf{Right:} The cycle--length Pareto frontier. For this ensemble, we forgo all distance filters and try to minimize the number of short cycles in the Tanner graphs; all displayed codes have quantum girth 6. We estimate their minimum distances with 20,000 iterations of \textsf{QDistEvol} and color each dot accordingly. For comparison, we also overlay the candidates in Table \ref{tab:ZSZ-LP codes intro}.}
    \label{fig:pareto frontiers}
\end{figure}

%%%%%%%%%%%%%%%%%%%%%%%%%%%%%%%%%%%%%%%%%%%%%%%%%%%%%

\section{Equivariant codewords and logical Pauli operators}
\label{sec:logical basis}

In this section, we describe methods to construct group-equivariant logical bases for the classical ZSZ-2BGA and quantum ZSZ-LP codes, with an emphasis on two complementary types: a \emph{minimum-weight} basis, whose operators all have weight $d$, and a \emph{symplectic} (canonical) basis built from systematic classical generator matrices. For some of the candidate codes in Table \ref{tab:ZSZ-LP codes intro}, the minimum-weight set forms a full-rank, but not symplectic, basis for the entire codespace. A non-symplectic basis means that the minimum-weight logical $\bar{X}$ and $\bar{Z}$ operators do not form anticommuting pairs; a minimum-weight logical $\bar{X}$ operator may anticommute with more than one minimum-weight logical $\bar{Z}$ operator. As such, they cannot directly map onto the $k$-qubit Pauli algebra. Alternatively, when the classical seed codes admit systematic generator matrices, we show that the lifted logical bases are automatically symplectic and full-rank, at the expense of larger operator weights; see the seventh column of Table \ref{tab:ZSZ-LP codes intro}.

\subsection{Classical left and right automorphisms}

Recall that the classical left and right seed codes enjoy free right and left actions by the ZSZ group respectively:
\begin{align}\label{eq:left and right auts}
    R[g]\,H_{\rm left} = H_{\rm left}\,R[g]^{\oplus2} \;,\qquad L[g]\,H_{\rm right} = H_{\rm right}\,L[g]^{\oplus2}
\end{align}
for all $g\in\ZSZ$. In other words, $R[g]$ and $L[g]$ are automorphisms of $H_{\rm left}$ and $H_{\rm right}$ respectively. As such, suppose we have a left codeword $\mathbf{c}_{\rm left}$ with $|\mathbf{c}|=d$ such that $H_{\rm left} \mathbf{c}_{\rm left} = 0$. Then $R[g]^{\oplus2}\mathbf{c}_{\rm left}$ is also a codeword of $H_{\rm left}$, since
\begin{align}
    H_{\rm left}\, R[g]^{\oplus2} \mathbf{c}_{\rm left} = R[g]\,H_{\rm left} \mathbf{c}_{\rm left} = 0 \, .
\end{align}
Hence, once we find a single left codeword $\mathbf{c}_{\rm left}$ with some desirable property such as minimum-weight, we can generate an entire set of size $|\ZSZ|=\ell$ of equivariant codewords (with possible duplicates) through its automorphism orbit:
\begin{align}
    \Aut(\mathbf{c}_{\rm left}) := \{ \mathbf{c} : \mathbf{c}=R[g]^{\oplus2} \mathbf{c}_{\rm left} \,,\, \forall g\in \ZSZ \}
\end{align}
An analogous statement holds for right codewords using ZSZ left actions. Since $k_{\rm left} = \ell$ when $H_{\rm left}$ has full rank, there is a possibility that this single automorphism orbit can span the entire left codespace, since $\rank\Aut(\mathbf{c}_{\rm left}) \leq \abs{\Aut(\mathbf{c}_{\rm left})} = \ell$. When this coincidence happens, we saw that the code is a \emph{one-generator} code \cite{Berman_1967, MacWilliams_1970, Sabin_1995, Ling_2001}. If the orbit is formed from a minimum-weight codeword, then we say it is a one-generator-minimal (OGM) basis. If the orbit is in reduced-row echelon form, then we say it is a one-generator-systematic (OGS) basis. In this work, we will construct both types of bases for the candidate codes; the last two columns of Table \ref{tab:classical ZSZ code params} indicate which candidates admit them.

To find OGM bases, we use \textsf{QDistEvol} to sample several different minimum-weight codewords for each seed code and compute their automorphism orbits. Remarkably, for every candidate ZSZ-LP code except ZSZ-LP-775, we find full-rank OGM bases for both its classical left/right codespaces; see the penultimate column of Table \ref{tab:classical ZSZ code params}. For ZSZ-LP-775, we only find OGM bases of rank 140 and 129 for the left and right codespaces respectively. Since every listed classical code has a full-rank parity-check matrix, the equivariant codewords actually span the entire codespace for all but the last code; i.e. they form a valid OGM generator matrix. We note that these OGM generator matrices are not systematic however. 

A one-generator systematic (OGS) generator matrix of the standard form $G=(\ident_k\;P)$ can be obtained through the following simple observation. When one of the group-algebra elements in the second block is invertible over the group algbera---say, $b$ for the left code or $d$ for the right code---its regular representation is also invertible, and the parity-check matrix factors as
\begin{align}
    H_{\rm left} &= L[b] \begin{pmatrix}\, L[b]^{-1}L[a] & \ident_\ell \,\end{pmatrix} = L[b]\,H^{\rm sys}_{\rm left}\, .
\end{align}
The row spaces of $H_{\rm left}$ and $H^{\rm sys}_{\rm left}$ are the same due to invertibility of $L[b]$, and so $H^{\rm sys}_{\rm left}$ is also a valid parity-check matrix of the same code. Since $H^{\rm sys}_{\rm left}$ is in systematic form, a systematic generator matrix readily follows:
\begin{align}
    G^{\rm sys}_{\rm left} = \begin{pmatrix}\,
        \ident_\ell & L[a]^\transp L[b^{-1}]^\transp
    \,\end{pmatrix} = \begin{pmatrix}\,
        \ident_\ell & L[\bar{a}\bar{b}^{-1}]
    \,\end{pmatrix} \, ,
\end{align}
where $\bar{a}$ denotes the antipode of $a$. An analogous calculation gives $G^{\rm sys}_{\rm right}$; see the discussion in Appendix \ref{app:systematic codewords} for more details. Unlike the OGM case, the rows of OGS generator matrices will generically not be minimum-weight. Nonetheless, systematic generator matrices are useful in two ways: \emph{(i)} they enable fast readout of the encoded message by simply examining the informational bits (the identity component of the generator matrix), and \emph{(ii)} they can lead to a symplectic logical Pauli basis in the quantum codes after the balanced product. Ultimately, an OGM versus OGS basis is a tradeoff between sparsity and addressability. For the first four candidate ZSZ-2BGA code pairs and corresponding ZSZ-LP codes, their invertible trinomials allow for OGS bases of maximum weight 18, 24, 30 and 46 (compared to their minimum distances of 10, 12, 14, and $\leq16$ respectively\footnote{Technically, our reported minimum distance for ZSZ-LP-390 is only an upper bound, but we are nonetheless quite confident of its tightness.}); see Table \ref{tab:OGS generators} for more details.

We emphasize that the classical left and right seed codes' ZSZ automorphisms \eqref{eq:left and right auts} do not generically carry over to the quantum ZSZ-LP code as a consequence of the ZSZ group being non-abelian. The reason is that the CSS check matrices \eqref{eq:H_X,H_Z balanced product} now contain both $H_{\rm left}$ and $H_{\rm right}$ and hence both left and right actions of the ZSZ group. As a consequence, only the center of the ZSZ group, i.e. the elements that commute with everything, can satisfy an analogous automorphism condition to \eqref{eq:left and right auts} for $H_X$ and $H_Z$.\footnote{Note that since the center commutes with everything in the group, it is abelian and therefore has no distinction between left and right.} Accordingly, we say that the center is the residual automorphism group left after the balanced product; see Appendix \ref{app:residual center aut} for formal details. For abelian groups like $\mathbb{Z}_\ell$ (in GB codes) and $\mathbb{Z}_{\ell_1} \times \mathbb{Z}_{\ell_2}$ (in BB codes), the entire group is its own center and so becomes the residual automorphism group. This sacrifice of automorphism group for circumventing low-weight codewords of \eqref{eq:H_left,H_right} is precisely the tradeoff mentioned in the introduction.

On the classical coding theory side, one-generator codes enable an additional practical payoff: \emph{fast encoding} \cite{Ling_2001}. The orbit therefore supplies a generator matrix whose rows all have weight exactly $d$, which is the sparsest generator matrix possible since every row of a generator matrix is itself a nonzero codeword. As a result, encoding a message amounts to $d$ shift-and-add passes over the message bits, costing $O(dn)$ binary operations in total, whereas encoding a generic (e.g. random) LDPC code with a dense generator matrix costs $O(n^2)$; preprocessing techniques can bring the generic cost close to linear \cite{RU_2001}. Furthermore, each shift is a simple translation on the $\ell_1\times\ell_2$ grid of group elements (Figure \ref{fig:ZSZ cayley graphs}), and so an encoder is architecturally similar to the quasi-cyclic LDPC encoders standardized in fifth-generation (5G) new-radio (NR) cellular communications \cite{Li_2006, Richardson_2018}: a $d$-stage pipeline of barrel shifters augmented with fixed wire permutations. Encoding remains fast even when the distance $d$ becomes large, since an approach based on the fast Fourier transform (FFT) costs $O\big(n(\log\ell_1+\ell_2)\big)$ operations regardless of $d$. One tradeoff is that these encoders are non-systematic, i.e. the message bits do not directly appear among the codeword bits, although the message can be recovered from a decoded codeword at little extra cost. The details of these encoding strategies can be found in Appendix \ref{app:fast encoding}.

\subsection{Logical Pauli bases from classical codewords}\label{sec:min-weight logical Paulis}

Recall the five-block structure of the CSS parity-check matrices \eqref{eq:5-block H_X,H_Z}. For a classical left or right codeword $\mathbf{c} \in \F^{2\ell}_2$, let
\begin{align}
    \mathbf{c} = \begin{pmatrix}\,
        \mathbf{c}_1 & \mathbf{c}_2 \,
    \end{pmatrix}
\end{align}
be its partition into the two sectors of the classical two-block parity-check matrices \eqref{eq:H_left,H_right}, where $\mathbf{c}_1, \mathbf{c}_2 \in \F^\ell_2$.
Then we can define the following $Z$-type and $X$-type Pauli operators that live in $\ker{H_X}$ and $\ker{H_Z}$ respectively:
\begin{subequations}\label{eq:codeword embedding}
\begin{align}
    \mathbf{z}_1 &= \begin{pmatrix}\,
        \mathbf{c}_{\rm left,1} & 0 & \mathbf{c}_{\rm left,2} & 0 & 0 \,
    \end{pmatrix} \\
    \mathbf{z}_2 &= \begin{pmatrix}\,
        0 & \mathbf{c}_{\rm left,1} & 0 & \mathbf{c}_{\rm left,2} & 0 \,
    \end{pmatrix} \\
    \mathbf{x}_1 &= \begin{pmatrix}\,
        \mathbf{c}_{\rm right,1} & \mathbf{c}_{\rm right,2} & 0 & 0 & 0 \,
    \end{pmatrix}  \label{eq:aut logical x1} \\
    \mathbf{x}_2 &= \begin{pmatrix}\,
        0 & 0 & \mathbf{c}_{\rm right,1} & \mathbf{c}_{\rm right,2} & 0 \,
    \end{pmatrix} \, ,  \label{eq:aut logical x2}
\end{align}
\end{subequations}
where $\mathbf{z}_1,\mathbf{z}_2 \in \ker{H_X}$ and $\mathbf{x}_1,\mathbf{x}_2 \in \ker{H_Z}$ directly follow from $\mathbf{c}_{\rm left} \in \ker{H_{\rm left}}$ and $\mathbf{c}_{\rm right} \in \ker{H_{\rm right}}$ inside \eqref{eq:H_X,H_Z balanced product}.

For the OGM codewords, having $\mathbf{z}_j \in \ker{H_X}$ does not immediately imply that it is a nontrivial logical $\bar{Z}$ operator since it could be a $Z$-stabilizer. For nontrivality, we must also ensure that $\mathbf{z}_j \notin \rs{H_Z}$, or in other words $\mathbf{z}_j \in \ker{H_X}\backslash \rs{H_Z}$; likewise we need $\mathbf{x}_j \in \ker{H_Z} \backslash \rs{H_X}$ for the $X$-type operators. Out of all equivariant logical candidates \eqref{eq:codeword embedding} for each ZSZ-LP code, inherited from their classical seed codes, we select a maximally linearly independent set of $\bar{Z}$ and $\bar{X}$ lying outside of $\rs{H_Z}$ and $\rs{H_X}$ respectively. For all ZSZ-LP codes except ZSZ-LP-775, we numerically find that the full-rank OGM bases lift to full-rank OGM logical $\bar{Z}$ and $\bar{X}$ Pauli bases respectively, which we will denote with $k\times n$ matrices $L_Z$ and $L_X$. Note that these OGM logical Pauli bases are not symplectic since they satisfy $L^{}_ZL^\transp_X = M \neq \ident$ for full-rank $M \in \F^{k\times k}_2$. To get a symplectic basis, we can choose to modify either the $\bar{X}$-logicals or the $\bar{Z}$-logicals. To modify the $\bar{Z}$ logicals, we simply multiply both sides by $M^{-1}$ on the left to get $\big( M^{-1} L^{}_Z \big) L^\transp_X = \ident$, where we then take $\tilde{L}_Z = M^{-1}L_Z$ to be our new $\bar{Z}$-logical basis, which is now symplectic with respect to $L_X$. So to obtain a symplectic logical Pauli basis from a minimum-weight one, we will generally need to sacrifice the minimum-weight property for at least one logical Pauli type.

Alternatively, a one-generator symplectic (OGS) basis can be obtained directly, at the expense of operator weight, if the classical seed codes admit systematic (OGS) generator matrices $G^{\rm sys}_{\rm left}=(\ident\;\;L[v^{\rm sys}_L])$ and $G^{\rm sys}_{\rm right}=(\ident\;\;R[v^{\rm sys}_R])$, whose rows are again equivariant codewords. With slight abuse of notation, we use the acryonym OGS to denote both one-generator-systematic and one-generator-symplectic, with the classical versus quantum setting differentiated by context. Among the candidate codes, this OGS property holds for ZSZ-LP-160 through ZSZ-LP-390. Embedding every row of $G^{\rm sys}_{\rm left}$ and $G^{\rm sys}_{\rm right}$ through the first-sector maps $\mathbf{z}_1$ and $\mathbf{x}_1$ of \eqref{eq:codeword embedding}, the resulting logical bases inherit the two-block structure:
\begin{align}\label{eq:systematic logical bases}
    L_Z = \big(\, \ident\;\; 0\;\; L[v^{\rm sys}_L]\;\; 0\;\; 0 \,\big) \;,\qquad
    L_X = \big(\, \ident\;\; R[v^{\rm sys}_R]\;\; 0\;\; 0\;\; 0 \,\big) \, .
\end{align}
These two bases overlap only on the first qubit subblock, where both restrict to the identity, and so $L^{}_ZL^\transp_X = \ident$ automatically: the systematic logical bases are symplectic as is, with conjugate pairs indexed by the group elements. Furthermore, unlike for the OGM bases, the full-rank property after the embedding \eqref{eq:codeword embedding} is guaranteed and requires no numerical verification. To see why, suppose some combination $\boldsymbol{\alpha} L^{}_Z$ were trivial, i.e. a $Z$-stabilizer in $\rs{H_Z}$. Stabilizers have vanishing inner products with every element of $\ker{H_Z}$, which contains the rows of $L_X$; however, $\big(\boldsymbol{\alpha} L^{}_Z\big)L^\transp_X = \boldsymbol{\alpha} \neq 0$, a contradiction. Hence the rows of $L_Z$ represent $\ell$ linearly independent, nontrivial logical classes, and likewise for $L_X$ upon exchanging $X\leftrightarrow Z$. Meanwhile, the invertibility of the trinomials $b$ and $d$ forces $H_{\rm left}$ and $H_{\rm right}$, and thus $H_X$ and $H_Z$, to have full rank, so the quantum code has exactly $k=\ell$ logical qubits and admits an OGS logical Pauli basis \eqref{eq:systematic logical bases}. The price to pay, echoing the sparsity-versus-utility tradeoff of the classical generators, are the logical weights $1+\wt(v^{\rm sys}_L)$ and $1+\wt(v^{\rm sys}_R)$, which can be extensive $O(\ell)$ and much larger than the minimum weight $d$; e.g. weights 30 and 28 versus $d=14$ for ZSZ-LP-320. Though we cannot rule out extremely special coincidences (we did not find any) where the OGM and OGS properties coincide.

\subsection{Surgery gadgets for equivariant logical Paulis}
\label{sec:surgery}

Code surgery is a general-purpose procedure to fault-tolerantly measure any subset of logical Pauli operators in an arbitrary quantum LDPC code \cite{Cohen_2022, Williamson_2026_gauging, cross2025improved, Swaroop_2026_adapters, he2025extractors, zheng2025high, cowtan2025fast, yuan2026parsimonious} and is a generalization of lattice surgery for topological codes \cite{Horsman_2012, landahl2014surgery}. We briefly review the algebraic formulation of code surgery before moving on to two specific constructions for the ZSZ-LP-390 candidate. At a high level, code surgery is a type of code deformation that converts some logical operators from the original code into low-weight stabilizers of the deformed code, sometimes referred to as gauging \cite{Williamson_2026_gauging} or weight reduction \cite{hastings2023weight, hsieh2025weight}. When the code is deformed, those logical operators are effectively measured and projected into the stabilizer subspace of the deformed code. Importantly, the LDPC property can be preserved during the deformation, which is sufficient for fault tolerance. We will recall the surgery diagram in the convention relevant for measuring logical $\bar{X}$ operators \cite{Yuan_2026_unified}. We write a CSS code as the chain complex
\begin{align}
    \mathcal{Q}:\quad
    S_X \xrightarrow{H_X^\transp} Q \xrightarrow{H_Z} S_Z \, ,        \label{eq:css_chain_complex_XQZ}
\end{align}
where $S_X$, $Q$, and $S_Z$ are respectively the $X$-syndrome, qubit, and $Z$-syndrome $\F_2$-vector spaces. A vector in $Q$ denotes the support of an $X$-type Pauli operator, and the $Z$-check matrix $H_Z$ computes its $Z$-syndrome. Logical $\bar{X}$ operators are elements of $\ker H_Z/\rs{H_X}$.

To perform $X$-surgery, we introduce another code called the ancilla code
\begin{align}
    \mathcal{Q}':\quad
    S'_X \xrightarrow{H_X^{\prime\transp}} Q' \xrightarrow{H'_Z} S'_Z \, ,        \label{eq:ancilla chain complex}
\end{align}
where we use primed labels for the ancilla's objects. The code deformation involves the attachment of the ancilla to the data code and is specified by two maps
\begin{align}
    \Gamma_1:S'_X\to Q \, ,
    \qquad
    \Gamma_0:Q'\to S_Z \, .
\end{align}
The map $\Gamma_1$ says which data qubits each new ancilla $X$-check touches, while $\Gamma_0^\transp$ says how each original data $Z$-check is extended onto the ancilla qubits. The merged code is a valid CSS code precisely when the following diagram commutes:
\begin{equation}\hspace{-2cm}
\begin{tikzcd}[column sep=large, row sep=large]
    \mathcal{Q}' : & S'_X & Q' & S'_Z \\
    \mathcal{Q} : & S_X & Q & S_Z
    \arrow["{H^{\prime\transp}_X}", from=1-2, to=1-3]
    \arrow["{H'_Z}", from=1-3, to=1-4]
    \arrow["{H_X^\transp}"', from=2-2, to=2-3]
    \arrow["{H_Z}"', from=2-3, to=2-4]
    \arrow["{\Gamma_1}"', from=1-2, to=2-3]
    \arrow["{\Gamma_0}"', from=1-3, to=2-4]
\end{tikzcd}
\label{eq:X surgery diagram}
\end{equation}
Equivalently, the commuting square (parallelogram) involving $S'_X,Q',Q,S_Z$ gives
\begin{align}
    H_Z\Gamma_1 = \Gamma_0 H^{\prime\transp}_X \, .        \label{eq:X surgery commuting square}
\end{align}
Ordering the merged qubit space as $Q\oplus Q'$ and the merged check spaces as $S_X\oplus S'_X$ and $S_Z\oplus S'_Z$, the parity-check matrices of the merged code have the block structure
\begin{align}\label{eq:X surgery merged checks}
    H_X^{\rm merge} = \begin{pmatrix}
        H_X & 0 \\
        \Gamma_1^\transp & H'_X
    \end{pmatrix} \, ,
    \qquad H_Z^{\rm merge} = \begin{pmatrix}
        H_Z & \Gamma_0 \\
        0 & H'_Z
    \end{pmatrix} \, .
\end{align}
The diagonal blocks retain the checks of the data and ancilla codes, while the off-diagonal blocks implement their attachment. Since the two diagonal codes are each CSS codes, the only new condition in $H_X^{\rm merge}{H_Z^{\rm merge}}^\transp=0$ is precisely the commuting-square condition \eqref{eq:X surgery commuting square}.
The ancilla part $H^{\prime\transp}_X$ determines the support of ancilla $X$-checks on ancilla qubits, and applying $\Gamma_0$ records exactly which deformed data $Z$ checks also touch those ancilla qubits. Thus \eqref{eq:X surgery commuting square} is the commutation condition for the merged CSS code.

The logical information extracted by surgery comes from products of ancilla $X$-checks whose support on the ancilla qubits cancels. Therefore the measured $X$-operator space on the data code is
\begin{align}
    \Gamma_1\!\left(\ker H^{\prime\transp}_X\right) \subseteq \ker H_Z \, .       \label{eq:X_surgery_measured_space}
\end{align}
If an element of \eqref{eq:X_surgery_measured_space} is nontrivial modulo $\rs{H_X}$, then the surgery measures the corresponding logical $\bar{X}$ operator.

The general procedure for measuring logical $\bar{X}$ operators via code surgery is as follows:
\begin{enumerate}
    \item Initialize all ancilla qubits in $\ket{0}$.
    \item Attach the ancilla by measuring all checks of the merged code. May require repeated rounds for fault tolerance.
    \item Infer the logical $\bar{X}$ measurements from products of the merged $X$-checks.
    \item Measure all ancilla qubits in the $Z$ basis to detach.
\end{enumerate}
The choice of ancilla initialization and measurement in the $Z$ basis is to ensure that the data code's checks are undisturbed (first block row of \eqref{eq:X surgery merged checks}). To ensure that the fault tolerance of the attachment step (step 2) is at least as good as the data code memory, we typically desire the merged code's minimum distances to be at least those of the data code. Note that for the $X$-type surgery illustrated by \eqref{eq:X surgery diagram}, the ancilla's $Z$-checks are not attached to the data code. As such, any relevant logical $\bar{Z}$ operator in the merged code must also be data $\bar{Z}$ operators when restricted to the data code's qubits $Q$. Hence, the merged $Z$-distance is automatically preserved. On the other hand, the ancilla's $X$-checks are attached to the data code's qubits through the attachment map $\Gamma_1$, and so the support of logical $\bar{X}$ operators in the data code can be deformed onto the ancilla code's qubits $Q'$ by appending stabilizers in $S'_X$. The ancilla code must therefore be constructed in a careful manner to ensure that these logical $\bar{X}$ deformations do not decrease the weight below $d_X$.

There are two main methods to construct the surgery ancilla in practice: \emph{(i)} a graph-based approach for measuring a small number of logical operators \cite{Williamson_2026_gauging}, and \emph{(ii)} a hypergraph-based approach for measuring a large number of logical operators \cite{Williamson_2026_gauging, zheng2025high}. For a generic LDPC code, different surgery ancillas may be required for measuring different logical operators. However, our construction of logical Pauli operators from automorphism orbits in Section \ref{sec:min-weight logical Paulis} enable a key advantage. Observe that for our equivariant logical $\bar{X}$ operators \eqref{eq:aut logical x1} and \eqref{eq:aut logical x2}, their $Z$-check neighborhoods, or local $Z$-Tanner graphs, are determined by $H_Z$ \eqref{eq:5-block H_Z}. When restricted to the supports of these logicals, the local $Z$-Tanner graphs resemble those of the classical seed code $H_{\rm right}$. Since the left-action automorphism of $H_{\rm right}$ \eqref{eq:left and right auts} acts as permutations on both sides of the parity-check matrix, this code automorphism is actually a Tanner-graph automorphism. As such, the local $Z$-neighborhoods of the equivariant $\bar{X}$ logicals are all isomorphic. Therefore, any surgery ancilla that we construct for one logical $\bar{X}$ is a valid surgery ancilla, in the sense that the diagram \eqref{eq:X surgery diagram} commutes, for all others in its classical automorphism orbit. However, this local neighborhood isomorphism for the $Z$-checks does not generically also hold for the $X$-checks, since the ZSZ right actions that preserve $X$-neighborhoods are not symmetries of $H_{\rm right}$, and so a distance-preserving surgery ancilla for one logical operator may no longer be distance-preserving for another logical operator. Note that this issue is not present when the underlying group is abelian, like in generalized bicycle codes. For a graph ancilla, this issue can also be circumvented by designing the graph to have sufficient edge expansion, i.e. Cheeger constant at least one, so $X$-check (vertex) deformations are guaranteed to not reduce weight \cite{webster2025explicit}.

We focus on ZSZ-LP-390 and its OGM logical Pauli basis and construct two different surgery gadgets to demonstrate the flexibility of code surgery for ZSZ-LP codes: one that measures a single minimum-weight logical $\bar{X}$ using a random graph ancilla \cite{Williamson_2026_gauging}, and another that measures eight logical $\bar{X}$s using a hypergraph product (HGP) ancilla \cite{zheng2025high}. See Table \ref{tab:surgery main} for a list of relevant properties for both types of surgery gadgets, and Figure \ref{fig:surgery gadgets} for a schematic illustration of the two ancilla constructions.

\begin{figure}[t]
    \centering
    \resizebox{\textwidth}{!}{%
    \begin{tikzpicture}[
      x=0.62cm,
      y=0.62cm,
      font=\footnotesize\sffamily,
      graph edge/.style={draw=ink!82,line width=0.75pt,line cap=round},
      graph vertex/.style={circle,fill=white,draw=ink,line width=0.75pt,
                           minimum size=3.8pt,inner sep=0pt},
      x support/.style={draw=xred,line width=2.25pt,line cap=round,line join=round},
      x check/.style={circle,fill=xredlight,draw=xred,line width=1.2pt,
                      minimum size=6pt,inner sep=0pt},
      z boundary/.style={draw=zblue,line width=1.95pt,line cap=round,line join=round},
      % HGP-ancilla cell styles (image-inspired lattice).
      hqubit v/.style={circle,fill=ink,draw=ink,line width=0.6pt,
                       minimum size=5.4pt,inner sep=0pt},
      hqubit c/.style={circle,fill=ink,draw=ink,line width=0.6pt,
                       minimum size=5.4pt,inner sep=0pt},
      hx cell/.style={rectangle,fill=xred,draw=xred!72!black,line width=0.7pt,
                      minimum size=6.6pt,inner sep=0pt},
      hz cell/.style={rectangle,fill=zblue,draw=zblue!72!black,line width=0.7pt,
                      minimum size=6.6pt,inner sep=0pt},
      z check/.style={rectangle,fill=zblue,draw=zblue!72!black,line width=0.7pt,
                      minimum size=5.4pt,inner sep=0pt},
      attach x/.style={draw=xred,dashed,line width=0.95pt,
                       -{Stealth[length=1.7mm,width=1.15mm]}},
      attach z/.style={draw=zblue,dashed,line width=0.95pt,
                       -{Stealth[length=1.7mm,width=1.15mm]}},
      map label/.style={fill=white,rounded corners=1pt,inner sep=1.5pt,
                        font=\footnotesize\bfseries\sffamily},
      pointer/.style={-{Stealth[length=1.2mm]},line width=0.6pt,shorten >=1pt},
    ]
    
    %=====================================================================
    % DATA CODE BLOCK (centre) : a flattened slab carrying three copies of
    % the logical X operator, stacked as three rows.
    %=====================================================================
    \path[fill=plane,draw=planeedge,line width=0.85pt,line join=round]
      (0.00,0.60) -- (10.00,0.60) -- (11.30,6.00) -- (1.30,6.00) -- cycle;
    \node[anchor=west,font=\small\bfseries\sffamily,text=ink,
          fill=plane,inner sep=1pt] at (1.55,5.55) {data code $\mathcal Q$};
    
    % --- Logical X_1 (top row) ---
    \coordinate (a1) at (1.50,4.30); \coordinate (a2) at (2.85,4.55);
    \coordinate (a3) at (4.30,4.25); \coordinate (a4) at (5.75,4.58);
    \coordinate (a5) at (7.25,4.26); \coordinate (a6) at (8.75,4.56);
    \coordinate (a7) at (10.15,4.35);
    \draw[x support] (a1)--(a2)--(a3)--(a4)--(a5)--(a6)--(a7);
    \foreach \q in {a1,a2,a3,a4,a5,a6,a7}\fill[black] (\q) circle (1.7pt);
    \node[font=\footnotesize\bfseries\sffamily,text=xred,anchor=south]
      at (4.30,4.42) {$\bar X_1$};
    
    % --- Logical X_2 (middle row) ---
    \coordinate (b1) at (1.10,3.00); \coordinate (b2) at (2.45,3.26);
    \coordinate (b3) at (3.90,2.96); \coordinate (b4) at (5.35,3.30);
    \coordinate (b5) at (6.85,2.98); \coordinate (b6) at (8.35,3.28);
    \coordinate (b7) at (9.85,3.05);
    \draw[x support] (b1)--(b2)--(b3)--(b4)--(b5)--(b6)--(b7);
    \foreach \q in {b1,b2,b3,b4,b5,b6,b7}\fill[black] (\q) circle (1.7pt);
    \node[font=\footnotesize\bfseries\sffamily,text=xred,anchor=south]
      at (3.90,3.12) {$\bar X_2$};
    
    % --- Logical X_3 (bottom row) ---
    \coordinate (c1) at (0.75,1.72); \coordinate (c2) at (2.10,1.98);
    \coordinate (c3) at (3.55,1.68); \coordinate (c4) at (5.00,2.00);
    \coordinate (c5) at (6.50,1.70); \coordinate (c6) at (8.00,1.98);
    \coordinate (c7) at (9.55,1.78);
    \draw[x support] (c1)--(c2)--(c3)--(c4)--(c5)--(c6)--(c7);
    \foreach \q in {c1,c2,c3,c4,c5,c6,c7}\fill[black] (\q) circle (1.7pt);
    \node[font=\footnotesize\bfseries\sffamily,text=xred,anchor=south]
      at (3.55,1.84) {$\bar X_3$};
    
    % qubit pointer (longer leader).
    \node[anchor=south,font=\footnotesize\sffamily,text=ink,
          fill=plane,inner sep=0.8pt] (dqnote) at (6.45,5.00) {qubit};
    \draw[-{Stealth[length=1.3mm]},draw=ink,line width=0.6pt]
      (dqnote.south) -- ($(dqnote.south)!0.88!(a4)$);
    
    % --- data Z-checks tapped by the HGP ancilla (right, held off the border) ---
    \coordinate (zAr) at (10.10,3.60); \draw[draw=ink,line width=0.7pt] (zAr)--(a7);
    \coordinate (zBr) at (9.80,2.30);  \draw[draw=ink,line width=0.7pt] (zBr)--(b7);
    \coordinate (zCr) at (9.50,1.10);  \draw[draw=ink,line width=0.7pt] (zCr)--(c7);
    \foreach \z in {zAr,zBr,zCr}\node[z check] at (\z) {};
    % --- data Z-check tapped by the graph ancilla (left, held off the border) ---
    \coordinate (zBl) at (1.15,2.40); \draw[draw=ink,line width=0.7pt] (zBl)--(b1);
    \node[z check] at (zBl) {};
    % data Z-check pointer (from below, clear of everything).
    \node[anchor=north,font=\footnotesize\sffamily,text=zblue!85!black,
          inner sep=0.8pt] (zdl) at (9.50,0.48) {$Z$-check};
    \draw[pointer,draw=zblue] (zdl.north) -- (9.50,0.98);
    
    %=====================================================================
    % GRAPH-SURGERY ANCILLA (left) : irregular sparse graph, attached to X_2.
    % Bean-shaped outline with a concave bottom so it reads as non-circular.
    %=====================================================================
    \begin{scope}[on background layer]
      \path[fill=blobfill,draw=blobedge,line width=0.8pt]
        plot[smooth cycle,tension=0.78] coordinates {
          (-5.40,3.15) (-4.95,4.45) (-4.15,5.28) (-2.75,5.55)
          (-1.10,5.00) (0.00,3.75) (0.00,2.30) (-1.00,1.60)
          (-2.30,2.00) (-3.40,1.35) (-4.55,1.60)
        };
    \end{scope}
    \coordinate (g1)  at (-4.60,3.25); \coordinate (g2)  at (-4.05,4.50);
    \coordinate (g3)  at (-2.85,4.90); \coordinate (g4)  at (-1.05,4.45);
    \coordinate (g5)  at (-0.42,2.60); \coordinate (g6)  at (-1.35,2.05);
    \coordinate (g7)  at (-3.15,1.85); \coordinate (g8)  at (-4.15,2.05);
    \coordinate (g9)  at (-3.30,3.40); \coordinate (g10) at (-2.55,2.72);
    
    % highlighted Z-check face (a 4-cycle) -- enlarged so its label fits cleanly.
    \path[fill=zbluelight!82,draw=none] (g3)--(g4)--(g5)--(g10)--cycle;
    
    % mesh edges : concave bottom (no edge spans the notch, no face diagonal).
    \foreach \e/\f in {g1/g2,g2/g3,g1/g8,g1/g9,g8/g7,g7/g9,g7/g10,
                       g6/g10,g5/g6,g9/g2,g9/g10,g3/g4,g4/g5,g5/g10,g10/g3}
      \draw[graph edge] (\e)--(\f);
    \draw[z boundary] (g3)--(g4)--(g5)--(g10)--cycle;
    \foreach \u in {g2,g4,g10}\draw[x support,line width=1.3pt] (g3)--(\u);
    
    \foreach \v in {g1,g2,g4,g5,g6,g7,g8,g9,g10}\node[graph vertex] at (\v) {};
    \node[x check] at (g3) {};
    \coordinate (qprime) at ($(g5)!0.50!(g10)$);
    \fill[zblue] (qprime) circle (1.55pt);
    % label the highlighted Z-check inside its (enlarged) face.
    \node[font=\scriptsize\sffamily,text=zblue!88!black,inner sep=0.5pt]
      at (-1.75,3.67) {$Z$-check};
    
    \node[anchor=south,font=\small\bfseries\sffamily,text=ink,
          fill=white,inner sep=1pt] at (-2.85,6.00)
      {graph ancilla $\mathcal{Q}'_{\mathrm{graph}}$};
    \node[anchor=east,font=\footnotesize\sffamily,text=xred!85!black,
          fill=white,inner sep=0.8pt] (gxn) at (-4.75,5.00) {$X$-check};
    \draw[pointer,draw=xred,shorten <=2pt] (gxn.east) -- ($(gxn.east)!0.90!(g3)$);
    \node[anchor=east,font=\footnotesize\sffamily,text=ink,
          fill=white,inner sep=0.8pt] (gen) at (-4.95,1.60) {qubit (edge)};
    \draw[pointer,draw=ink,shorten <=2pt,shorten >=4pt] (gen.east) -- (-3.65,1.95);
    
    %=====================================================================
    % HGP-SURGERY ANCILLA (right) : hypergraph-product lattice in a box.
    % Rows of X-checks / qubits map one-to-one to the three logicals.
    %=====================================================================
    % column x-positions : two "bit" columns, two "check" columns.
    \def\cA{13.80}\def\cB{14.70}\def\cC{15.95}\def\cD{16.85}
    % row y-positions : two "bit" rows (top), three "check/logical" rows.
    \def\rA{7.00}\def\rB{5.70}\def\rC{4.40}\def\rD{3.10}\def\rE{1.80}
    
    \begin{scope}[on background layer]
      \path[fill=blobfill,draw=blobedge,line width=0.8pt,rounded corners=2pt]
        (13.25,1.25) rectangle (17.45,7.55);
      % faint lattice guide lines.
      \foreach \yy in {\rA,\rB,\rC,\rD,\rE}
        \draw[draw=blobedge!35,line width=0.5pt] (\cA,\yy)--(\cD,\yy);
      \foreach \xx in {\cA,\cB,\cC,\cD}
        \draw[draw=blobedge!35,line width=0.5pt] (\xx,\rE)--(\xx,\rA);
    \end{scope}
    
    % top-left  : sector-1 qubits (green-ringed).
    \foreach \xx in {\cA,\cB}\foreach \yy in {\rA,\rB}
      \node[hqubit v] at (\xx,\yy) {};
    % top-right : Z-checks (blue).
    \foreach \xx in {\cC,\cD}\foreach \yy in {\rA,\rB}
      \node[hz cell] at (\xx,\yy) {};
    % bottom-left : X-checks (red)  -> three logical rows.
    \foreach \xx in {\cA,\cB}\foreach \yy in {\rC,\rD,\rE}
      \node[hx cell] at (\xx,\yy) {};
    % bottom-right : sector-2 qubits (black).
    \foreach \xx in {\cC,\cD}\foreach \yy in {\rC,\rD,\rE}
      \node[hqubit c] at (\xx,\yy) {};
    
    \node[anchor=south,font=\small\bfseries\sffamily,text=ink,
          fill=white,inner sep=1pt] at (15.35,7.75)
      {HGP ancilla $\mathcal{Q}'_{\mathrm{HGP}}$};
    % X-check / Z-check pointers into their quadrants.
    \node[font=\scriptsize\sffamily,text=xred,fill=white,inner sep=0.5pt]
      (hxc) at (14.10,5.12) {$X$-check};
    \draw[pointer,draw=xred] (13.92,4.92) -- (13.85,4.60);
    \node[font=\scriptsize\sffamily,text=zblue,fill=white,inner sep=0.5pt]
      (hzc) at (16.50,5.12) {$Z$-check};
    \draw[pointer,draw=zblue] (16.68,5.32) -- (16.80,5.52);
    
    %=====================================================================
    % ATTACHMENT MAPS
    %=====================================================================
    % --- graph ancilla -> logical X_2 (dashed maps kept clear of the solid edges) ---
    \draw[attach x] (g3) .. controls (-0.80,5.00) and (0.35,3.45) .. (b1);
    \node[map label,text=xred] at (-0.10,4.28) {$\Gamma_1$};
    \draw[attach z] (qprime) .. controls (-0.30,2.30) and (0.45,2.30) .. (zBl);
    \node[map label,text=zblue] at (0.15,1.95) {$\Gamma_0$};
    
    % --- HGP ancilla -> all three logicals ---
    % Gamma_1 : ancilla X-checks (red rows) -> data logical qubits.
    \draw[attach x] (\cA,\rC) -- (a7);
    \draw[attach x] (\cA,\rD) -- (b7);
    \draw[attach x] (\cA,\rE) -- (c7);
    \node[map label,text=xred] at (12.40,4.90) {$\Gamma_1$};
    % Gamma_0 : ancilla qubits (black rows) -> data Z-checks.
    \draw[attach z] (\cC,\rC) .. controls (14.80,3.70) and (12.60,3.62) .. (zAr);
    \draw[attach z] (\cC,\rD) .. controls (14.80,2.45) and (12.40,2.30) .. (zBr);
    \draw[attach z] (\cC,\rE) .. controls (14.60,1.30) and (12.20,1.12) .. (zCr);
    \node[map label,text=zblue] at (11.75,0.80) {$\Gamma_0$};
    \end{tikzpicture}
    }
    \caption{Schematic of the two $X$-surgery gadgets constructed for the ZSZ-LP-390 data code $\mathcal{Q}$ (centre), drawn carrying three copies $\bar{X}_1,\bar{X}_2,\bar{X}_3$ of the logical $\bar{X}$ operator. For a graph ancilla (left), the vertices, edges, and faces (cycles) are respectively the ancilla $X$-checks, qubits, and $Z$-checks; here $\mathcal{Q}'_{\mathrm{graph}}$ measures a single logical (attached to $\bar{X}_2$). For a hypergraph-product ancilla (right), the lattice of ancilla $X$-checks (red squares), $Z$-checks (blue squares), and qubits (black dots) is arranged so that each of its three rows of $X$-checks and qubits measures a different logical, allowing $\mathcal{Q}'_{\mathrm{HGP}}$ to measure all three in parallel. In both gadgets, the attachment maps $\Gamma_1$ (red, dashed) and $\Gamma_0$ (blue, dashed) send ancilla $X$-checks to data qubits and ancilla qubits to data $Z$-checks, realizing the off-diagonal blocks of the merged check matrices \eqref{eq:X surgery merged checks}.}
    \label{fig:surgery gadgets}
\end{figure}
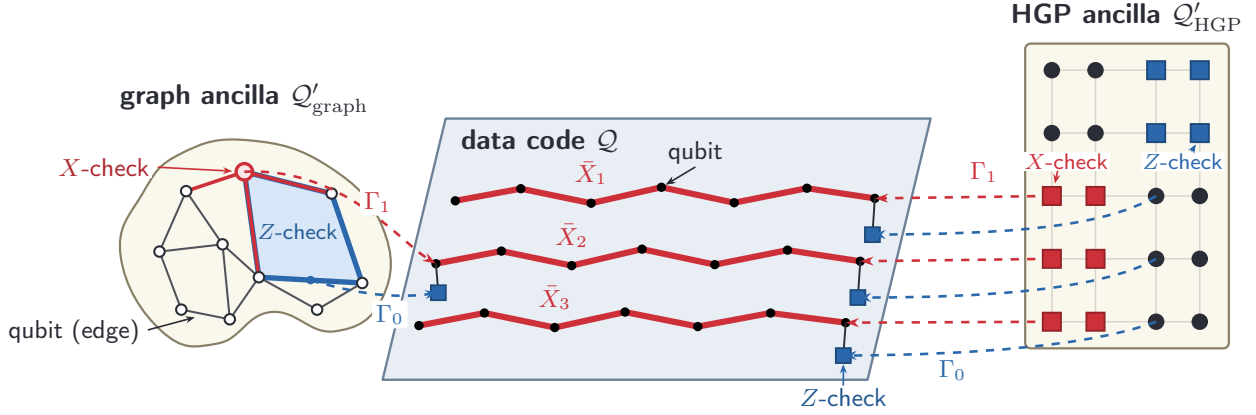

\subsubsection{Random graph ancilla for single logical measurement}

We choose the first logical $\bar{X}$ operator from the OGM basis of ZSZ-LP-390 and construct a graph-based surgery gadget for its measurement. We follow the randomized graph approach of \cite{Williamson_2026_gauging}, augmented by the edge-scoring subroutine of \cite{zheng2025high}. As a reminder, for a graph-based surgery ancilla, the vertices, edges and faces (cycles) correspond to $X$-checks, qubits and $Z$-checks respectively. The procedure starts by examining the local $Z$-Tanner graph of the logical $\bar{X}$ operator and assigns an ancilla vertex/$X$-check to every data qubit in the logical's support. Then, for each $Z$-check with nonzero even overlap with the logical support, it pairs up the overlap vertices and adds one edge per pair. Cycles are assigned as ancilla $Z$-checks. After this initial graph construction, we iterate the following routine:
\begin{enumerate}
    \item Estimate the merged code's $X$-distance using \textsf{QDistEvol}, and return a list of numerically found merged $\bar{X}$ logicals with weight less than a chosen target, which we set as $d_X$. Call these the ``bad'' logicals.
    \item Generate a list of candidate edges to add to the graph, and score each candidate edge by the number of bad logicals it eliminates through its added cycles/$Z$-checks. Elimination could mean either having a new $Z$-check anticommute with a bad logical or converting a bad logical to a stabilizer.
    \item Add the candidate edge with the highest score in addition to its extra corresponding checks.
\end{enumerate}
For our specific ZSZ-LP-390 logical $\bar{X}$, we find that only two additional edges are needed to ensure a distance-preserving merged code.

One can combine this single-logical surgery gadget to other graph-based surgery gadgets with ``bridge'' qubits \cite{cross2025improved, Swaroop_2026_adapters} to measure products of this logical $\bar{X}$ with other logical Paulis, either from the same code block or from different code blocks (which can also be a different code), to generate logical entanglement for teleportation or Pauli-based computation \cite{Litinski_2019}.

\begin{table}[t]
\centering\renewcommand{\arraystretch}{1.1}
\resizebox{\textwidth}{!}{
\begin{tabular}{cccccccccc}
\toprule
\textbf{Data code} & \textbf{$|\bar{X}_{\rm meas}|$} & \textbf{\makecell{Surgery\\ ancilla}} & \textbf{Merged code} & \textbf{$m'_X$} & \textbf{$m'_Z$} & $\Delta_C$ & $\Delta_Q$ & $\wt(\Gamma_0)$ & $\wt(\Gamma_1)$ \\ \midrule
$\llbracket 390,78,\leq16 \rrbracket$ & 1 & graph & $\llbracket 416,77,\leq16 \rrbracket$ & 16 & 11 & 10 & 12 & 1 & 1 \\
$\llbracket 390,78,\leq16 \rrbracket$ & 8 & HGP & $\llbracket 1622,70,\leq16 \rrbracket$ & 640 & 576 & 15 & 12 & 6 & 2 \\ \bottomrule
\end{tabular}
}
\caption{Some properties of two types of surgery gadgets for the ZSZ-LP-390 code are reported such as the number of measured $\bar{X}$ logicals $|\bar{X}_{\rm meas}|$, the ancilla construction method, the number of extra $X$-checks ($m'_X$) and $Z$-checks ($m'_Z$), the maximum check ($\Delta_C$) and qubit ($\Delta_Q$) degrees, and the maximum row and column weight of the attachment maps $\Gamma_0, \Gamma_1$. Note that for the HGP surgery gadget, we have ignored gauge logical qubits for both $k$ and $d$.}
\label{tab:surgery main}
\end{table}

\subsubsection{Hypergraph product ancilla for multiple logical measurements}

Recall that all equivariant logical $\bar{X}$ operators have the same weight and isomorphic $Z$-check neighborhoods. As such, they satisfy the ``generalized parallel pattern'' required for HGP code surgery \cite{zheng2025high}, where the surgery ancilla is a HGP code \cite{HGP}. To benchmark the potential of HGP surgery, we will choose eight logical $\bar{X}$ operators from the OGM basis and construct a HGP surgery ancilla for their parallel measurement. We first find a full-rank even partition of the logical $\bar{X}$ operators into the first and second sectors (labeled by subscript $j=1,2$ for $\mathbf{x}_j$ in \eqref{eq:codeword embedding}). We then select eight logical $\bar{X}$ operators whose maximum data-qubit overlap and shared $Z$-check overlap (congestion) are 2. For the HGP ancilla, we choose the convenient convention such that HGP($H_1,H_2$) has CSS check matrices $H'_X = (H^\transp_1\otimes\ident \;|\; \ident\otimes H^\transp_2)$ and $H'_Z = (\ident\otimes H_2 \;|\; H_1\otimes\ident)$. The measured-$X$ space \eqref{eq:X_surgery_measured_space} then involves
\begin{align}
    \ker H^{\prime\transp}_X = \ker\begin{pmatrix}
        H_1 \otimes \ident \\
        \ident \otimes H_2
    \end{pmatrix} = \ker H_1 \otimes \ker H_2 \, ,
\end{align}
which is the codespace of the tensor-product code $\mathcal{C}_1 \otimes \mathcal{C}_2$, where $\mathcal{C}_1$ and $\mathcal{C}_2$ are the individual codespaces of $H_1$ and $H_2$ respectively. We choose $H_1$ to be the local code derived from the local $Z$-check neighborhood of a single logical $\bar{X}$ (recall that they are all isomorphic), where qubits are mapped to bits and $Z$-checks to parity checks. For our chosen logical operator, this local code $H_1$ has $k_1=1$ logical bit; i.e. there are no other (unintended) logicals in the support of our measured logical. For $H_2$, since we are measuring 8 logical $\bar{X}$ operators, we want its code dimension to be at least $k_2\geq8$. Furthermore, since ZSZ-LP-390 has an $X$-distance of 16, we want $H_2$ to have a minimum distance of 16 \cite{zheng2025high}. We choose a quasi-cyclic LDPC matrix for $H_2$ with protomatrix
\begin{align}
    H^{\rm proto}_2 = \begin{pmatrix}
        x^2 & 0 & x^6 & x^7 & x^5 \\
        0 & 1 & 1 & x^7 & 1 \\
        1 & x^2 & x^3 & 0 & x^6 \\
        1 & 1 & x^3 & x & 0
    \end{pmatrix} \, ,
\end{align}
lift size 8, and code parameters $[40,8,16]$. Note that the HGP ancilla by itself may contain spurious/gauge logical qubits, since
\begin{align}
    k_{\rm HGP} &= \dim\ker H_1 \cdot \dim\ker H^\transp_2 + \dim\ker H^\transp_1 \cdot \dim\ker H_2
\end{align}
may not be zero due to linear dependencies among the rows of $H_1$ leading to a nonzero $\dim\ker H^\transp_1$. It is typical to just ignore these gauge logical qubits in the computation of $k$ and $d$ for the merged code. We find that trimming some linearly dependent rows from $H_1$ leads to an HGP ancilla with fewer physical (and gauge logical) qubits at the cost of a denser attachment map $\Gamma_0$. We trimmed 6 linearly dependent rows from $H_1$ at the cost of $\wt\Gamma_0=6$, which causes the maximum $Z$-check weight in the merged code to increase to $9+6=15$. Since the congestion of our eight logicals is 2, we also have $\wt\Gamma_1=2$. To see how the attachment maps work for an HGP ancilla, it is useful to arrange the ancilla's $X$-checks in a rectangular layout where columns correspond to codewords of $H_1$ and rows to codewords of $H_2$ inside $\ker H^{\prime\transp}_X$. In this layout, $\Gamma_1$ ($\Gamma_0$) attaches each measured logical's support (neighboring $Z$-checks) to the rows corresponding to the informational bits of $H_2$, i.e. its non-pivot columns in reduced-row echelon form.

We have used this HGP ancilla to measure eight logical $\bar{X}$ operators in a single ZSZ-LP-390 code block, but we note that the same HGP ancilla can straightforwardly be used to measure eight logical $\bar{X}$ operators in different code blocks (with potentially smaller congestion too) by having the attachment maps $\Gamma_0,\Gamma_1$ act on different code blocks. A larger HGP ancilla can also be used if we want to measure more than just eight logical operators, at the potential cost of increased congestion. In the scenario where we want to measure pairwise (or larger) products of these same logical operators, there is also a particularly elegant way to modify the HGP ancilla based on classical code augmentation \cite{Xu_2025_fast}. Recall that the attachment map $\Gamma_1$ is attached to the ancilla's $X$-checks that correspond to the informational bits of $H_2$. If, for example, we want to measure the pairwise product $\bar{X}_1\bar{X}_2$ rather than the individual operators, we simply add/augment an additional parity check (row) to $H_2$ with support on the first two informational bits (columns). In general, measuring any product-partition of the logical operators simply corresponds to augmenting additional checks to $H_2$.

Lastly, we comment that the above HGP ancilla incurs significantly more overhead than, say, simply using multiple copies of the graph-based single surgery gadget or using a random hypergraph approach \cite{zheng2025high}. Nonetheless, we present it here because we see it as the most didactic high-rate construction, and the quasi-cyclic nature of $H_2$ may allow for easier implementation in certain hardware such as reconfigurable atom arrays.

%%%%%%%%%%%%%%%%%%%%%%%%%%%%%%%%%%%%%%%%%%%%%%%%%%%%%

\section{$ZX$-dualities from automorphism-fold symmetries}
\label{sec:ZX dualities}

So far, we have considered independent seed codes $H_{\rm left}$ and $H_{\rm right}$ to construct ZSZ-LP codes, which we have found was a useful choice for achieving high distance as well as minimizing short cycles in the Tanner graphs. But one may wonder what advantages we can get if we pick specially structured $H_{\rm right}$ for a given $H_{\rm left}$. In this section, we briefly touch on a special case where a carefully chosen $H_{\rm right}$ leads to isomorphisms between the $X$- and $Z$-checks, which can enable certain transversal gates. The cost is that these ``fold-symmetric'' ZSZ-LP codes will typically have smaller distances and more short Tanner-graph cycles, echoing our general theme of a performance--utility tradeoff.

A $ZX$-duality is a symmetry of a CSS code that exchanges its $X$- and $Z$-stabilizer groups \cite{Breuckmann2024fold}. More concretely, let $\tau$ be a permutation of the physical qubits, which we also regard as its permutation matrix. In terms of CSS check matrices, the most general defining property is
\begin{align}\label{eq:ZX duality row spaces}
    \rs(H_X\tau)=\rs(H_Z)\; ,\qquad
    \rs(H_Z\tau)=\rs(H_X) \, .
\end{align}
Thus, a $ZX$-duality is essentially an ordinary code automorphism that also swaps the $X$ and $Z$ sectors. In this work we focus on the more special setting where $\tau$ maps the chosen $X$-checks onto the chosen $Z$-checks one by one, rather than merely mapping their row spaces. In other words, the qubit permutation can be accompanied by permutations of the check labels that transform $H_X$ directly into $H_Z$; i.e. a Tanner graph isomorphism. Some paradigmatic examples are \emph{(i)} a 2D color code that is ``self-dual'' with $H_X=H_Z$ \cite{Bombin_2006}, \emph{(ii)} a 2D square surface code with a diagonal ``fold'' symmetry \cite{Moussa_2016}, and \emph{(iii)} ``square'' hypergraph product (HGP) codes with a diagonal fold symmetry \cite{Quintavalle_2023_fold}. Our ZSZ-LP codes share the most similarity with the square HGP codes due to the product structure \eqref{eq:H_X,H_Z balanced product}; a fold symmetry essentially corresponds to swapping the sides of the tensor product $\otimes$. However, unlike for HGP codes where the product structure is over binary numbers (and hence qubits), our product structure is over group-algebra polynomials. A particularly insightful way to package the ZSZ-LP structure is to imagine an ordinary square HGP code formed from the product of two $1\times 2$ parity-check matrices, but now all local Hilbert spaces have a group-algebra-valued basis rather than a two-level one. From this perspective, the fold symmetry of square HGP codes \cite{Quintavalle_2023_fold} can extend to our balanced product codes provided that we can also map these local group-algebra elements (polynomials) onto one another.

To see this algebraic fold, recall the five-block matrices \eqref{eq:5-block H_X,H_Z}, which we write again as
\begin{align}\label{eq:5-block recap}
    H_X &= \begin{pmatrix}
        A & 0 & B & 0 & C^\transp \\
        0 & A & 0 & B & D^\transp
    \end{pmatrix} \; ,\qquad
    H_Z = \begin{pmatrix}
        C & D & 0 & 0 & A^\transp \\
        0 & 0 & C & D & B^\transp
    \end{pmatrix} \, ,
\end{align}
where $A=L[a]$, $B=L[b]$ are left-regular matrices and $C=R[c]$, $D=R[d]$ are right-regular matrices. To get a $H_X \leftrightarrow H_Z$ isomorphism, we can first attempt to mirror their block structures. Suppose we exchange the second and third block columns, denoted by $\Pi_{23}$, which puts the zero blocks in the correct locations:
\begin{align}\label{eq:5-block fold}
    H_Z \Pi_{23} = \begin{pmatrix}
        C & 0 & D & 0 & A^\transp \\
        0 & C & 0 & D & B^\transp
    \end{pmatrix} \; ,\qquad
    H_X \Pi_{23} &= \begin{pmatrix}
        A & B & 0 & 0 & C^\transp \\
        0 & 0 & A & B & D^\transp
    \end{pmatrix} \, .
\end{align}
Now, comparing \eqref{eq:5-block fold} with \eqref{eq:5-block recap}, we observe that what remains is to provide a means to map $A\leftrightarrow C$ and $B\leftrightarrow D$ as well as their transposes. Recall that $A,B$ and $C,D$ are left-regular and right-regular matrix representations over a group algebra $\F_2[G]$.
The remaining obstacle is thus exchanging left and right multiplication. Inversion does exactly this exchange since it reverses the order of a group product. More generally, we may further compose inversion with any group automorphism $\phi:G\rightarrow G$ \cite{Eberhardt_2025_fold} and define
\begin{align}\label{eq:theta automorphism fold}
    \theta:G\rightarrow G\; ,\qquad
    \theta(g)=\phi(g^{-1}) \, .
\end{align}
Due to the embedded inversion, the map $\theta$ reverses the order of products, $\theta(gh)=\theta(h)\theta(g)$, and hence converts left multiplication into right multiplication. In other words, it is the antiautomorphism version of $\phi$. Let $P_\theta$ be the permutation matrix induced by $\theta$ on the group-indexed basis of the regular representation, and let $p \in \F_2[G]$ be a group-algebra element. Then
\begin{align}\label{eq:left right theta conjugation}
    P^{}_\theta L[p]P_\theta^{-1}=R[\phi(\bar p)] \, ,
\end{align}
where $\bar{p}$ denotes the antipode map that inverts all group elements. We consequently choose the right seed polynomials to be
\begin{align}\label{eq:automorphism fold c d}
    c=\phi(\bar a)\; ,\qquad d=\phi(\bar b) \, ,
\end{align}
which results in $A\leftrightarrow C$ and $B\leftrightarrow D$. It remains to show that their transposes also map onto one another.

The fifth block column, containing the transposed matrices, correctly exchanges $A^\transp \leftrightarrow C^\transp$ and $B^\transp \leftrightarrow D^\transp$ provided
\begin{align}\label{eq:phi square seed condition}
    \phi^2(a)=a\; ,\qquad \phi^2(b)=b \, .
\end{align}
An involutive automorphism $\phi^2=1$ is the simplest choice, for which \eqref{eq:phi square seed condition} holds automatically and $\theta^2=1$. The overall $ZX$-duality, or Tanner graph isomorphism, then takes the compact form
\begin{align}\label{eq:ZSZ LP automorphism ZX duality}
    H_Z = P^{\oplus 2}_\theta H^{}_X P^{\oplus 5}_\theta \Pi^{}_{23} \, .
\end{align}
We see that the physical action $P^{\oplus 5}_\theta \Pi^{}_{23}$ comprises of group (anti)automorphism $P^{\oplus 5}_\theta$ within each data subblock and an exchange of data subblocks $\Pi^{}_{23}$. As such, we call the physical transformation an automorphism-fold, where $P^{\oplus 5}_\theta$ represents the (group) ``automorphism'' part and $\Pi^{}_{23}$ the ``fold'' part; see Figure \ref{fig:ZSZ-LP automorphism-fold} for an illustration where the data and check qubits are arranged according to the balanced product structure \eqref{eq:H_X,H_Z balanced product}. Because $\theta$ and $\Pi_{23}$ are both involutions and mutually commute, the resulting physical permutation $\tau=P^{\oplus 5}_\theta\Pi^{}_{23}$ also satisfies $\tau^2=\ident$ and hence resembles an algebraic analogue of a fold.

\begin{figure}[t]
    \centering
    \begin{tikzpicture}[
      font=\sffamily\small,
      datablock/.style={
        draw=black!55, very thick, rounded corners=1pt, fill=black!5,
        minimum width=1.8cm, minimum height=1.2cm
      },
      xcheck/.style={
        draw=red!70, very thick, rounded corners=1pt, fill=red!8,
        minimum width=1.8cm, minimum height=1.2cm
      },
      zcheck/.style={
        draw=blue!70, very thick, rounded corners=1pt, fill=blue!8,
        minimum width=1.8cm, minimum height=1.2cm
      },
      foldaxis/.style={
        draw=black!65, line width=1.25pt, dashed,
        dash pattern=on 5pt off 3pt
      },
      swap/.style={
        {Stealth[length=2.4mm,width=2mm]}-{Stealth[length=2.4mm,width=2mm]},
        draw=green, line width=1.2pt,
        preaction={draw=white, line width=3.2pt}
      },
      automorphism/.style={
        -{Stealth[length=1.9mm,width=1.6mm]},
        draw=green, line width=0.9pt
      },
      xautomorphism/.style={
        -{Stealth[length=1.9mm,width=1.6mm]},
        draw=green, line width=0.9pt
      },
      zautomorphism/.style={
        -{Stealth[length=1.9mm,width=1.6mm]},
        draw=green, line width=0.9pt
      },
    ]

    % A compact 3-by-3 array makes the check-side fold explicit.
    \node[datablock] (q1) at (0,1.73) {};
    \node[datablock] (q2) at (2.6,1.73) {};
    \node[datablock] (q3) at (0,0) {};
    \node[datablock] (q4) at (2.6,0) {};
    \node[datablock] (q5) at (5.2,-1.73) {};
    \node[xcheck] (s1) at (0,-1.73) {};
    \node[xcheck] (s2) at (2.6,-1.73) {};
    \node[zcheck] (t1) at (5.2,1.73) {};
    \node[zcheck] (t2) at (5.2,0) {};

    % This slope follows the corner-to-corner diagonals of Q_1, Q_4, and Q_5.
    \draw[foldaxis]
      ($(q1.north west)+(-0.18,0.12)$) --
      ($(q5.south east)+(0.18,-0.12)$);

    % The block exchange runs along the opposite diagonal and crosses the fold.
    \draw[swap]
      ($(q2.south west)+(-0.07,-0.05)$)
      to[bend left=18]
      node[pos=0.18, below right=3pt, fill=white, text=green,
           inner sep=0.8pt] {$\Pi_{23}$}
      ($(q3.north east)+(0.07,0.05)$);

    % The two check-block pairs are exchanged across the same fold.
    \draw[swap]
      ($(s1.north east)+(0.07,0.05)$)
      to[bend right=18]
      ($(t1.south west)+(-0.07,-0.05)$);
    \draw[swap]
      ($(s2.north east)+(0.07,0.05)$)
      to[bend right=18]
      ($(t2.south west)+(-0.07,-0.05)$);

    % The same group (anti)automorphism acts within every data subblock.
    \foreach \q/\lab in {q1/Q_1,q2/Q_2,q3/Q_3,q4/Q_4,q5/Q_5}{
      \coordinate (aut\q) at ($(\q.center)+(0.40,0.22)$);
      \draw[automorphism]
        ($(aut\q)+(125:0.32)$)
        arc[start angle=125,end angle=440,radius=0.32];
      \node[font=\sffamily\scriptsize, fill=black!5,
            text=green, inner sep=0.8pt]
        at (aut\q) {$P_\theta$};
      \node[font=\sffamily\bfseries, anchor=south west, fill=black!5,
            inner sep=1.2pt]
        at ($(\q.south west)+(0.13,0.11)$) {$\lab$};
    }

    \foreach \s/\lab in {s1/C^X_1,s2/C^X_2}{
      \coordinate (aut\s) at ($(\s.center)+(0.40,0.22)$);
      \draw[xautomorphism]
        ($(aut\s)+(125:0.32)$)
        arc[start angle=125,end angle=440,radius=0.32];
      \node[font=\sffamily\scriptsize, fill=red!8,
            text=green, inner sep=0.8pt]
        at (aut\s) {$P_\theta$};
      \node[font=\sffamily\bfseries, anchor=south west, fill=red!8,
            inner sep=1.2pt]
        at ($(\s.south west)+(0.13,0.11)$) {$\lab$};
    }

    \foreach \t/\lab in {t1/C^Z_1,t2/C^Z_2}{
      \coordinate (aut\t) at ($(\t.center)+(0.40,0.22)$);
      \draw[zautomorphism]
        ($(aut\t)+(125:0.32)$)
        arc[start angle=125,end angle=440,radius=0.32];
      \node[font=\sffamily\scriptsize, fill=blue!8,
            text=green, inner sep=0.8pt]
        at (aut\t) {$P_\theta$};
      \node[font=\sffamily\bfseries, anchor=south west, fill=blue!8,
            inner sep=1.2pt]
        at ($(\t.south west)+(0.13,0.11)$) {$\lab$};
    }
    \end{tikzpicture}
    \caption{Schematic of the automorphism-fold (green) on the data and check subblocks. The data subblocks $Q_i$ are gray, the $X$-check subblocks $C^X_1,C^X_2$ are red, and the $Z$-check subblocks $C^Z_1,C^Z_2$ are blue. The dashed diagonal line represents the fold axis. The green double-sided arrows denote the interblock exchanges $Q_2\leftrightarrow Q_3$, $C^X_1\leftrightarrow C^Z_1$, and $C^X_2\leftrightarrow C^Z_2$ across the diagonal, mimicking a fold. The circular arrow inside every subblock denotes the intrablock rearrangement according to the group (anti)automorphism $P_\theta$.}
    \label{fig:ZSZ-LP automorphism-fold}
\end{figure}
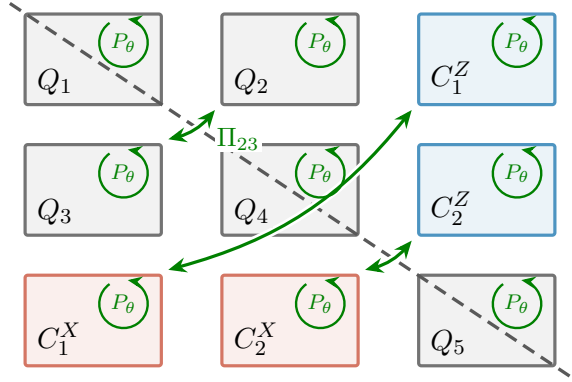

The simplest group automorphism satisfying \eqref{eq:phi square seed condition} is the trivial automorphism $\phi=1$, which would be closest direct analogue to the geometric fold in HGP codes. It unfortunately turns out that the trivial automorphism introduces ``too much symmetry'' into the code, resulting in weight-3 logical operators and thus $d\leq3$; see Appendix \ref{app:fold-symmetric-low-distance} for formal details. We thus move forward to nontrivial, involutive automorphisms for our ZSZ groups. Since ZSZ is just a two-generator group ($x$ and $y$) involving one semidirect product, finding these automorphisms is quite easy. Indeed, $\phi$ is completely determined by the two images $g_x=\phi(x)$ and $g_y=\phi(y)$. With multiplication given by the push-through relations \eqref{eq:push-through relations}, we can enumerate the $|\ZSZ|^2=\ell^2=\ell^2_1\ell^2_2$ possible image pairs and retain those satisfying the group-presentation relations $g_x^{\ell_1}=g_y^{\ell_2}=1$ and $g_yg_xg_y^{-1}=g_x^q$ \eqref{eq:ZSZ presentation}. Finally, because $x$ and $y$ generate the group, involutivity only needs to be checked on them individually: $\phi^2(x)=x$ and $\phi^2(y)=y$. This direct enumeration is relatively inexpensive for the group sizes we consider.

We present three examples of fold-symmetric ZSZ-LP codes in Table \ref{tab:ZSZ-LP ZX dualities}. We reuse the same ZSZ groups as those in Table \ref{tab:ZSZ-LP polynomials} but alter the trinomials. To find these code instances, we first fix the two trinomials $a,b$ of $H_{\rm left}$ and then exhaustively search through all involutive ZSZ group automorphisms for $\phi$. For the first two examples, we could find fold-symmetric ZSZ-LP codes with the same $H_{\rm left}$ and minimum distances as those of the candidates in Table \ref{tab:ZSZ-LP codes intro}. For the third example, the $H_{\rm left}$ of ZSZ-LP-390 produced fold-symmetric instances with poor distances, so we used a different $H_{\rm left}$ with $d_{\rm left}=16$ instead.

\begin{table}[t]
\centering\renewcommand{\arraystretch}{1.3}
\begin{tabular}{ccccc}
\toprule
$\llbracket n,k,d\rrbracket$ & $\ell_1,\ell_2,q$ & $a$ & $b$ & \textbf{Group aut.} $\phi$ \\ \midrule
$\llbracket 160,32,10\rrbracket$ & $16,2,9$ & $1+x^{13}+x^{14}y$ & $1+x^3y+x^{13}y$ & $\phi(x)=xy \;,\; \phi(y)=y$ \\
$\llbracket 240,48,12\rrbracket$ & $12,4,7$ & $1+x^8y+x^7y^2$ & $1+x^2+x^5y$ & $\phi(x)=x^{11} \;,\; \phi(y)=y^3$ \\
$\llbracket 390,78,\leq15\rrbracket$ & $26,3,3$ & $1+x^{21}y+x^5y^2$ & $1+y+xy$ & $\phi(x)=x^{-1} \;,\; \phi(y)=y$ \\ \bottomrule
\end{tabular}
\caption{Some fold-symmetric ZSZ-LP codes, specified by their left trinomials $a,b$ and group automorphism $\phi$; the right trinomials are determined by $c=\phi(\bar{a})$ and $d=\phi(\bar{b})$. For the first two codes, we borrow the left trinomials from those of the candidates in Table \ref{tab:ZSZ-LP polynomials}, while the $n=390$ entry uses a different left trinomial. Exact and upper bounds on code distances are computed using the \textsf{pySATDist} and \textsf{QDistEvol} ($10^6$ iterations) methods respectively \cite{webster2026distance}.}
\label{tab:ZSZ-LP ZX dualities}
\end{table}

\subsection{Hadamard-type gates}

The first gate arising from a $ZX$-duality is a Hadamard-type fold-transversal gate \cite{Breuckmann2024fold}. Apply a physical Hadamard to every data qubit and then exchange every two-element orbit of the fold:
\begin{align}\label{eq:physical fold Hadamard}
    \mathcal{H}_\tau
    =\left(\prod_{i<\tau(i)}\mathrm{SWAP}_{i,\tau(i)}\right)
      \left(\bigotimes_{i=1}^n H_i\right) \, .
\end{align}
The transversal Hadamard exchanges Pauli $X$ and $Z$, while the depth-1 SWAP layer moves the support of each operator through the fold. Since the physical permutation $\tau$ maps every $X$-check to a $Z$-check and vice versa, the gate \eqref{eq:physical fold Hadamard} preserves the stabilizer group and is therefore an encoded Clifford gate. In practice, the qubit permutation can be implemented by relabeling or physical motion.

Assume that the code admits the OGS bases of \eqref{eq:systematic logical bases}. Recall that both $L_X$ and $L_Z$ restrict to the identity on $Q_1$, so their conjugate logical pairs are naturally indexed by group elements $g\in G$. The fold acts on this first subblock by $g\mapsto\theta(g)$, and the remaining pieces of the OGS operators follow modulo stabilizers. Hence
\begin{align}\label{eq:logical fold Hadamard action}
    \mathcal{H}_\tau:\qquad
    \bar{X}_g\longmapsto\bar{Z}_{\theta(g)}\; ,\qquad
    \bar{Z}_g\longmapsto\bar{X}_{\theta(g)} \, .
\end{align}
Equivalently, this is a logical Hadamard on every encoded qubit followed by the logical permutation $g\mapsto\theta(g)$. A fixed point of $\theta$ receives a logical Hadamard, whereas a two-element orbit $\{g,\theta(g)\}$ receives Hadamards on both logical qubits together with a logical SWAP. For the three examples in Table \ref{tab:ZSZ-LP ZX dualities}, the logical folds have respectively $(4,14)$, $(36,6)$, and $(26,26)$ fixed-point/two-cycle counts. Thus their logical gates comprise $k=32$, 48, and 78 parallel Hadamards together with 14, 6, and 26 logical SWAPs, respectively. From a computational standpoint, a fold-symmetric ZSZ-LP code could hence serve as a low-overhead code block for constant-rate Hadamard gates.

\subsection{Phase-type gates}

A $ZX$-duality can also give rise to a diagonal, phase-type fold-transversal gate. Let $F=\{i:\tau(i)=i\}$ be the physical qubits fixed by the fold and partition it into $F_+$ and $F_-$. The physical circuit has the form \cite{Breuckmann2024fold}
\begin{align}\label{eq:physical fold phase}
    \mathcal{S}_\tau
    =\left(\bigotimes_{i\in F_+}S_i\right)
     \left(\bigotimes_{i\in F_-}S_i^\dagger\right)
     \left(\prod_{i<\tau(i)}\mathrm{CZ}_{i,\tau(i)}\right) \, .
\end{align}
Thus, every two-element orbit is coupled by a CZ gate, while every fixed qubit receives either $S$ or $S^\dagger$. Unlike the Hadamard-type gate, a $ZX$-duality alone is not quite sufficient: the phases accumulated by each $X$-check must cancel. A useful sufficient criterion is that the total number of fixed qubits is even, every $X$-check overlaps $F$ an even number of times, and every $X$-check contains an even number of complete two-element orbits. The sets $F_+$ and $F_-$ can then be chosen so that the $S$ and $S^\dagger$ phases globally cancel, akin to a doubly even CSS code. All three examples in Table \ref{tab:ZSZ-LP ZX dualities} satisfy these orbit-parity conditions.

Ignoring phases, \eqref{eq:physical fold phase} leaves $Z$ operators unchanged and acts on an $X$-type support $s$ as
\begin{align}
    X(s)\longmapsto X(s)Z(\tau(s)) \, .
\end{align}
For an $X$-check, the second factor is precisely its paired $Z$-check, so this again preserves the stabilizer group. Assuming the OGS basis, the identity on $Q_1$ immediately identifies the corresponding logical action:
\begin{align}\label{eq:logical fold phase action}
    \mathcal{S}_\tau:\qquad
    \bar{Z}_g\longmapsto\bar{Z}_g\; ,\qquad
    \bar{X}_g\longmapsto\bar{X}_g\bar{Z}_{\theta(g)} \, .
\end{align}
Each fixed point of $\theta$ therefore receives a logical phase gate, while every two-cycle $\{g,\theta(g)\}$ receives a logical CZ. Up to a logical $Z$ Pauli frame determined by the choices of $S$ versus $S^\dagger$ on physical fixed points, we can summarize the logical circuit as
\begin{align}\label{eq:logical fold CZ S gate}
    \overline{\mathcal{S}}_\tau
    =\prod_{g=\theta(g)}\bar{S}_g
     \prod_{\{g,h\}:\,h=\theta(g),\,g\neq h}\overline{\mathrm{CZ}}_{g,h} \, ,
\end{align}
where the second product contains one representative of every two-element orbit. For the three examples in Table \ref{tab:ZSZ-LP ZX dualities}, the logical circuits therefore implement respectively 4, 36, and 26 logical phase gates together with 14, 6, and 26 logical CZ gates. At the physical level, their folds have respectively 12, 108, and 78 fixed qubits. Hence \eqref{eq:physical fold phase} contains 74, 66, and 156 disjoint physical CZ gates, respectively, together with phase gates on the fixed qubits. From a computational standpoint, the fold-symmetric ZSZ-LP codes with large fixed-point group automorphisms could serve as a constant-rate resource for $\ket{Y}=S\ket{+}$ states, which can feed into FTQC primitives such as twist-free surgery \cite{Chamberland_2022_twistfree, Cowtan_2026_parallel}.

%%%%%%%%%%%%%%%%%%%%%%%%%%%%%%%%%%%%%%%%%%%%%%%%%%%%%

\section{Syndrome extraction for reconfigurable atom arrays}
\label{sec:syndrome extraction}

As is customary for qLDPC codes, we propose bare-ancilla syndrome extraction to nondestructively measure our check operators and obtain our error syndrome. For each check operator (row of $H_X$ or $H_Z$), we introduce one ancilla, often called a check qubit, that will be used to measure the check operator's parity. The parity is measured by entangling the check qubit with its corresponding data qubits, given by the support of the check operator, and then projectively measuring the check qubit in its computational basis \cite{gottesman1997thesis, Dennis_2002}. The syndrome-data entanglement is typically performed with two-qubit gates, whose mutual ordering or schedule can lead to many different syndrome extraction circuits. 

We now describe how we build syndrome extraction gate schedules that are reasonably natural for reconfigurable atom arrays. In these arrays, there are typically static traps generated by spatial light modulators (SLMs), and dynamical tweezers generated by acousto-optic deflectors (AODs) to move atoms between SLM traps. Recall from \eqref{eq:5-block H_X,H_Z} that the quantum code has five data-qubit subblocks. In this section, we will label these five subblocks by $Q_1, Q_2, Q_3, Q_4, Q_5$ for ease of discussion. Each data subblock contains $\ell=\ell_1\ell_2$ qubits, which we arrange as $\ell_2$ rows and $\ell_1$ columns according to the group element $x^ay^b$: the $x$ exponent is the column coordinate, while the $y$ exponent is the row coordinate, which mimicks the ZSZ Cayley graph structure in Figure \ref{fig:ZSZ cayley graphs}. Thus the physical data layout comprises of five rectangular data-atom subblocks of identical shapes. Similarly, we denote the two $X$-check syndrome subblocks by $C^X_1,C^X_2$ and the two $Z$-check syndrome subblocks by $C^Z_1,C^Z_2$. The resulting physical layout is illustrated in Figure \ref{fig:atom-layout}.

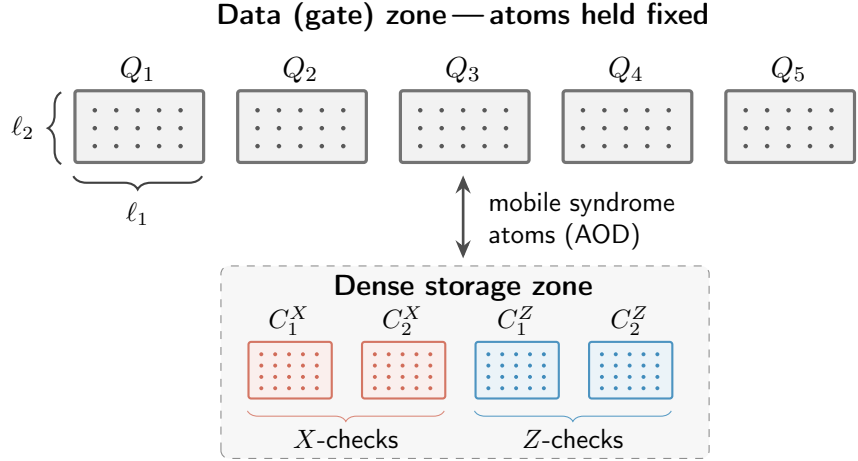
\begin{figure}[t]
    \centering
    \begin{tikzpicture}[
      font=\sffamily\small,
      datablock/.style={draw=black!55, very thick, rounded corners=1pt, fill=black!5},
      sblock/.style={draw=red!70, thick, rounded corners=1pt, fill=red!8},
      tblock/.style={draw=blue!70, thick, rounded corners=1pt, fill=blue!8},
      dataatom/.style={circle, fill=black!65, inner sep=0pt, minimum size=1.7pt},
      satom/.style={circle, fill=red!75, inner sep=0pt, minimum size=1.5pt},
      tatom/.style={circle, fill=blue!75, inner sep=0pt, minimum size=1.5pt},
      zonebox/.style={draw=black!45, dashed, rounded corners=3pt},
      transport/.style={{Stealth[length=2.4mm,width=2mm]}-{Stealth[length=2.4mm,width=2mm]},
                        line width=1pt, draw=black!70},
    ]

    % ---------- DATA (GATE) ZONE: five subblocks A..E ----------
    \node[font=\sffamily\bfseries] at (5.15,1.95) {Data (gate) zone\,---\,atoms held fixed};
    \foreach \lab [count=\i from 0] in {Q_1,Q_2,Q_3,Q_4,Q_5}{
      \coordinate (o\lab) at ({\i*2.15},0);
      \node[datablock, anchor=south west, minimum width=1.7cm, minimum height=0.95cm]
            (data\lab) at (o\lab) {};
      \node[font=\sffamily\bfseries] at ($(o\lab)+(0.85,1.25)$) {$\lab$};
      \foreach \cx in {1,...,5}{
        \foreach \cy in {1,...,3}{
          \node[dataatom] at ($(o\lab)+(\cx*0.2833,\cy*0.2375)$) {};
        }
      }
    }

    % dimension braces on subblock Q_1
    \draw[draw=black!75, thick, decorate, decoration={brace, amplitude=5pt}]
         ($(oQ_1)+(-0.16,0)$) -- ($(oQ_1)+(-0.16,0.95)$) node[midway, left=6pt] {$\ell_2$};
    \draw[draw=black!75, thick, decorate, decoration={brace, amplitude=5pt, mirror}]
         ($(oQ_1)+(0,-0.16)$) -- ($(oQ_1)+(1.7,-0.16)$) node[midway, below=6pt] {$\ell_1$};

    % ---------- transport arrow between zones ----------
    \draw[transport] (5.15,-0.15) -- (5.15,-1.25);
    \node[align=left, anchor=west] at (5.35,-0.75)
         {mobile syndrome\\atoms (AOD)};

    % ---------- DENSE STORAGE ZONE: C^X_1,C^X_2,C^Z_1,C^Z_2 ----------
    \def\sy{-3.12}
    \draw[zonebox, fill=black!3] (1.95,-3.90) rectangle (8.40,-1.35);
    \node[font=\sffamily\bfseries] at (5.15,-1.64) {Dense storage zone};
    \foreach \nm/\xx/\bs/\as/\lb in {%
        S1/2.30/sblock/satom/{$C^X_1$}, S2/3.80/sblock/satom/{$C^X_2$},
        T1/5.30/tblock/tatom/{$C^Z_1$}, T2/6.80/tblock/tatom/{$C^Z_2$}}{
      \node[\bs, anchor=south west, minimum width=1.1cm, minimum height=0.75cm]
            (\nm) at (\xx,\sy) {};
      \node at ($(\xx,\sy)+(0.55,1.05)$) {\lb};
      \foreach \cx in {1,...,5}{
        \foreach \cy in {1,...,4}{
          \node[\as] at ($(\xx,\sy)+(\cx*0.1833,\cy*0.15)$) {};
        }
      }
    }
    % Group braces for the X- and Z-check subblocks.
    \draw[decorate, decoration={brace, amplitude=4pt, mirror}, draw=red!60]
         (2.30,-3.27) -- (4.90,-3.27) node[midway, below=3pt] {$X$-checks};
    \draw[decorate, decoration={brace, amplitude=4pt, mirror}, draw=blue!65]
         (5.30,-3.27) -- (7.90,-3.27) node[midway, below=3pt] {$Z$-checks};

    \end{tikzpicture}
    \caption{Physical layout used for syndrome extraction on a reconfigurable atom array. The five data-qubit subblocks $Q_1, Q_2, Q_3, Q_4, Q_5$ (one per block-column of $H_X,H_Z$ in \eqref{eq:5-block H_X,H_Z}) are held stationary in the gate zone, each laid out as an $\ell_2\times\ell_1$ rectangle of atoms: $\ell_1$ columns indexed by the $x$-exponent and $\ell_2$ rows by the $y$-exponent. The $X$-check qubits form the syndrome subblocks $C^X_1,C^X_2$ and the $Z$-check qubits form $C^Z_1,C^Z_2$ (the two block-rows of $H_X$ and $H_Z$ respectively), held compactly in a dense storage zone \cite{Zhou_2025_arch}. During syndrome extraction the syndrome atoms are rearranged and transported by AODs to the data subblocks for entangling gates.}
\label{fig:atom-layout}
\end{figure}

In this block picture, choosing a CNOT ordering for syndrome extraction is essentially choosing an ordering for the monomials in the nonzero blocks of $H_X$ and $H_Z$ \eqref{eq:5-block H_X,H_Z}. For example, if a nonzero block contains the trinomial $L[g_1]+L[g_2]+L[g_3]$, then one can first move the relevant syndrome subblock into the relative frame for $g_1$, apply all CNOT gates corresponding to $L[g_1]$ transversally, then do the same for $g_2$ and $g_3$. The same idea applies to right-regular blocks $R[g]$. Note that a ``single monomial'' here represents $\ell$ physical CNOT gates applied transversally, one for each group element in the subblock.

The transition cost between two successive monomials is determined by the push-through relations \eqref{eq:push-through relations}. If the syndrome block is currently aligned for a left monomial $L[g]$, and we next want $L[h]$, then the relevant transition monomial is $hg^{-1}$. If it is currently aligned for a right monomial $R[g]$ and we next want $R[h]$, then the relevant transition monomial is $g^{-1}h$. As an example, using $(\ell_1,\ell_2,q) = (26,3,3)$ for ZSZ-LP-390, we can have something like
\begin{align}
    1 \xlongrightarrow{\;\;xy\;\;} L[xy] \xlongrightarrow{\,y^2(xy)^{-1}\,=\,x^{23}y\,} L[y^2] \, .
\end{align}
When switching between a left and right interaction, the scheduler moves the currently active reference-frame component back to the data frame and then aligns the other component to the next monomial. After writing the transition monomial in canonical form $x^\alpha y^\beta$, the push-through relations tell us what physical rearrangement primitives are needed. Left multiplication by $y^\beta$ produces an affine stretch (the $T$ matrix in \eqref{eq:L[xy] tensor decomp}), which can be decomposed into consecutive horizontal ``riffle shuffles'': the subblock is first stretched horizontally by a factor of $q$, then the right most $q-1$ partitions are interleaved among the leftmost partition akin to shuffling a deck of cards \cite{ZSZ_codes}. This last interleaving step can either be performed sequentially with one crossed AOD, or in parallel with $q-1$ crossed AODs, one for each partition. For right multiplication by $x^\alpha$, the horizontal shear leads to row-dependent cyclic shifts through the push-through relation \eqref{eq:push-through x right}. The above two non-rigid movements are the main places where the non-abelian twist $q$ shows up in the hardware compilation problem. See \cite[Algorithm 1]{ZSZ_codes} and \cite[Algorithm 2]{ZSZ_codes} for further details on these subroutines.

For the routing cost model, we apply a simple uniform counting rule where a cyclic shift and a one-step riffle shuffle are treated as comparable primitive AOD rearrangements. A left move by $y^\beta$ is decomposed into sequential $q$-riffle shuffles, followed by the final vertical cyclic shift. A right move by $x^\alpha$ is decomposed into sequential horizontal cyclic shifts, grouped by rows that require the same displacement. The cost model thus keeps track of two important hardware costs: the number of primitive rearrangements and how many independent AOD channels are needed at once. The latter can matter when $q$ is large: a parallelized riffle shuffle may have small depth, but it uses extra simultaneous crossed-AOD channels during the shuffle. From the cost model, one may notice that for the most AOD-friendly ZSZ-LP codes, we would like to have both $q$ and $\ell_2$ be small. The candidate ZSZ-LP codes with $n\geq320$ all somewhat possess this desirable criteria.

\begin{table}[t]
\centering
\begin{tabular}{ccccccc}
\toprule
\textbf{Code name} & \textbf{\makecell{AOD \\ channels}} & \textbf{\makecell{Riffle \\ shuffles}} & \textbf{\makecell{Cyclic \\ shifts}} & \textbf{\makecell{Move \\ layers}} & \textbf{\makecell{Gate \\ layers}} & \textbf{\makecell{Total time for \\ rearrangements}} \\ \midrule
ZSZ-LP-320 & 4 & 20 & 136 & 18 & 24 & 35 ms \\
ZSZ-LP-390 & 4 & 20 & 114 & 17 & 24 & 33 ms \\
ZSZ-LP-550 & 4 & 36 & 184 & 18 & 25 & 58 ms \\
ZSZ-LP-775 & 2 & 40 & 182 & 18 & 24 & 59 ms \\
BB-144 (gross) & 2 & 0 & 17 & 6 & 7 & 6 ms \\\bottomrule
\end{tabular}
\caption{The routing costs are listed for one round of sequential $X+Z$ syndrome extraction for the candidate ZSZ-LP codes with $q=2,3$ using a simple greedy scheduler, with stationary data qubits and mobile check qubits. For comparison, we also report the costs for the BB-144 gross code with an interleaved depth-7 syndrome extraction circuit \cite{BB_codes}. To compute the movement times, we use an interatomic (storage zone) spacing of $5\,\mu$m and an AOD acceleration of $5500$\,m/$s^2$ \cite{Zhou_2025_arch}. Note that the total rearrangement time does not include any gate-layer movements, and so the full time for one round of syndrome extraction may be higher.}
\label{tab:greedy AOD cost}
\end{table}

We use a simple greedy scheduler for the numbers reported in Table \ref{tab:greedy AOD cost}. For each syndrome subblock, the scheduler keeps track of the current left and right reference frames. At each step, it looks at the remaining monomial jobs involving that syndrome subblock, computes the transition cost from the current frame to each candidate job, and chooses the cheapest next job. Ties are broken deterministically so that the output schedule is reproducible. After routes are built for the individual syndrome subblocks, parallelizable move layers and parallelizable gate layers are greedily merged. This lets two syndrome subblocks move in the same layer when they do not require the same physical resource, and it lets two monomial gate layers run together when they do not reuse the same syndrome or data subblock. The pseudocode is provided in Appendix \ref{app:greedy-aod-scheduler}.

The routing costs of the greedy scheduler are listed in Table \ref{tab:greedy AOD cost}. We report the peak number of simultaneous AOD channels, the number of riffle shuffles and cyclic shifts, the number of move and gate layers, and a rough physical rearrangement time using a constant-acceleration model \cite{Zhou_2025_arch}. In these schedules, data atoms are kept fixed and only syndrome atoms are rearranged. We also do not include the extra motion needed to bring syndrome atoms from the storage zone to the gate zone, so the reported rearrangement times should be understood as the time spent on the syndrome rearrangement itself rather than the full duration of a syndrome extraction round. For the ZSZ-LP codes, we expect the rearrangements to dominate the overall time budget, since gate moves for one syndrome subblock can always be interleaved during the move layer for the other subblock.

%%%%%%%%%%%%%%%%%%%%%%%%%%%%%%%%%%%%%%%%%%%%%%%%%%%%%%%%%

\section{Numerical simulations}
\label{sec:numerics}

\subsection{Quantum circuit-level memory}

For the memory simulations of all ZSZ-LP codes in Table \ref{tab:ZSZ-LP codes intro}, we use the reconfigurable-atom-array-inspired sequential $X+Z$ syndrome extraction schedule described in the previous section. As is customary in the literature, we perform $d$ repeated rounds of syndrome extraction for each memory experiment to combat measurement errors. To probe the error-correcting performance against both $X$ and $Z$ errors, we run two types of experiments. For an $X$-memory ($Z$-memory) experiment, where we are correcting $Z$ errors, we begin with all data qubits in $\ket{+}$ ($\ket{0}$), perform $d$ rounds of $X+Z$ syndrome extraction, and finally measure all data qubits in the $X$-basis ($Z$-basis). We report the per-round block logical error rates (LER), i.e. any logical failure in the entire code block, in the left panel of Figure \ref{fig:numerics}. Both the right panel in Figure \ref{fig:zsz-lp vs other kd2n} and the left panel in Figure \ref{fig:numerics} are averaged over $X$-memory and $Z$-memory experiments. Recall that for a raw LER $\bar{p}$ over $T$ rounds, the per-round LER is computed as $\bar{p}_{\rm round} = 1-(1-\bar{p})^{1/T}$. For the LER per-qubit-round, simply multiply $T$ by $k$. To extrapolate the LER curves of ZSZ-LP-320 through ZSZ-LP-550, we fit our subthreshold ($p\leq0.4\%$) LER data points to the phenomenological fitting function $\bar{p} = p^{d_{\rm circ}/2} {\rm e}^{c_0+c_1p+c_2p^2}$ with fit parameters $c_0,c_1,c_2$ and circuit-distance estimate $d_{\rm circ}$ \cite{BB_codes}. We choose $d_{\rm circ} = d-2$ based on circuit-distance estimates for the small ZSZ-LP codes; see the discussion in Appendix \ref{app:more numerics} and Table \ref{tab:relay-gamma-intervals} for more details. The ${\rm e}^{c_0+c_1p+c_2p^2}$ part is to account for the initial ``waterfall'' behavior, where the near-threshold slope of the log-log LER curves is observed to be steeper than $d_{\rm circ}/2$.

On the right panel in Figure \ref{fig:numerics}, we report the $X$-surgery performance of the graph-surgery gadget in Table \ref{tab:surgery main}. To benchmark the logical measurement fidelity, we simulate the noisy surgery circuit described in Section \ref{sec:surgery}. In particular, we initialize all data qubits in $\ket{+}$ and all surgery ancilla qubits in $\ket{0}$, perform $d$ rounds of syndrome extraction on the merged code, then measure all surgery ancilla qubits in the $Z$ basis, perform one round of syndrome extraction on the data code, and finally measure all data qubits in the $X$ basis. Since all logical $\bar{X}$ operators have +1 eigenvalues from the transversal $\ket{+}$ initialization, surgery succeeds if we correctly infer a +1 outcome from the relevant combination of ancilla $X$-checks. For the merged code, we perform syndrome extraction in two steps: we first perform the CNOTs for both data and ancilla blocks in parallel, and then we perform the CNOTs corresponding to the attachment maps. For the ancilla block and attachment maps, we schedule the CNOTs using a random edge-coloration of the Tanner subgraphs \cite{Tremblay_2022}.

\begin{figure}[t]
    \centering
    \includegraphics[width=0.5\textwidth]{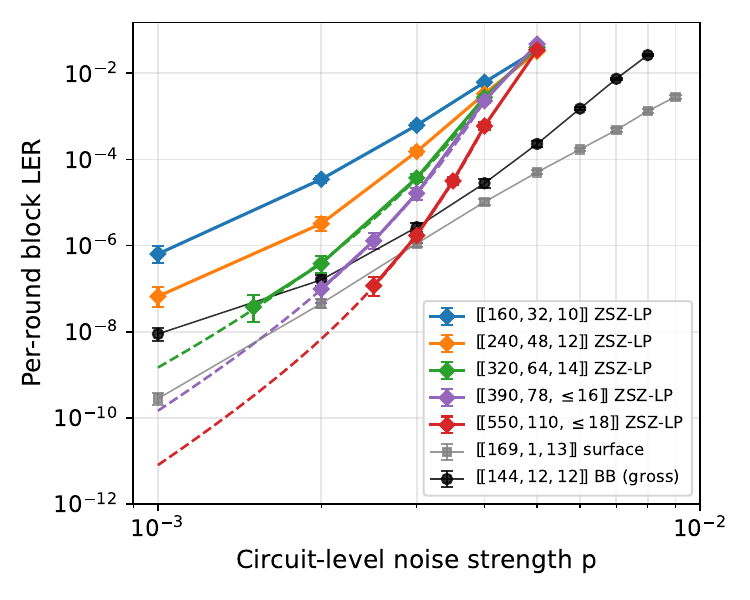}
    \includegraphics[width=0.49\textwidth]{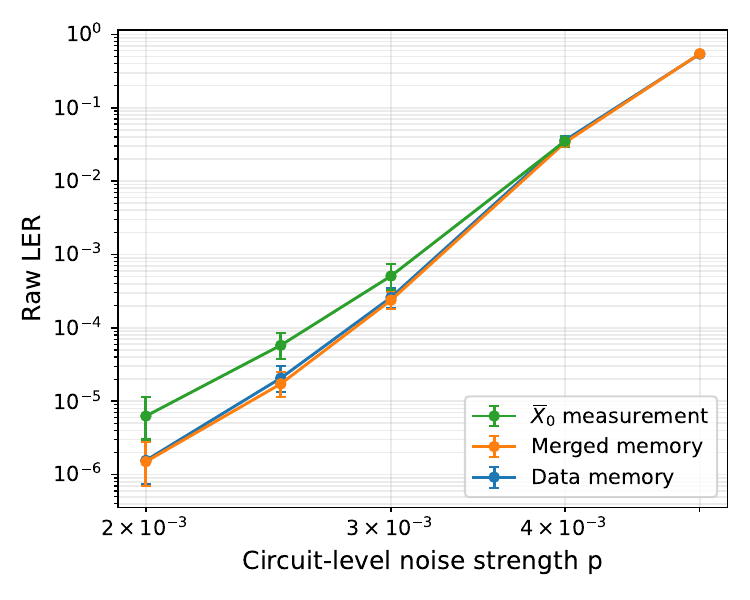}
    \caption{\textbf{Left:} The block logical error rate (LER) per syndrome-extraction round, averaged over $X$-memory and $Z$-memory simulations, for the candidate ZSZ-LP codes in Table \ref{tab:ZSZ-LP codes intro} as well as the BB-144 gross code and the $d=13$ surface code. \textbf{Right:} The raw (unnormalized) LER is plotted for the graph-surgery gadget of ZSZ-LP-390 in Table \ref{tab:surgery main}. ``$\bar{X}_0$ measurement'' refers to the logical $\bar{X}$ measurement of the first $\bar{X}$-logical in the OGM basis. The LERs of the $d$-round memory circuits for the $\llbracket 390,78,\leq16\rrbracket$ data code and $\llbracket 416,77,\leq16\rrbracket$ merged codes are also plotted for comparison. Error bars denote $95\%$ Clopper-Pearson confidence intervals.}
    \label{fig:numerics}
\end{figure}

Our noise model is parameterized by a strength $p$ and consists of the following:
\begin{itemize}[itemsep=0pt, topsep=5pt]
    \item Two-qubit gates experience depolarizing noise with probability $p$.
    \item Initializations in $\ket{1}/\ket{-}$ occur instead of $\ket{0}/\ket{+}$ with probability $p$.
    \item Measurement results incorrectly flip with probability $p$.
\end{itemize}
We have not included single-qubit or idling noise, and so our numerical results should only be interpreted as an optimistic best-case scenario. We ignore single-qubit noise since single-qubit gates have much higher fidelity than two-qubit gates in both trapped-ion and reconfigurable atom arrays, which are our platforms of interest. Idling noise, on the other hand, depends sensitively on the particular choice of hardware configuration and compilation, which can be very complex and change over time. Since we want to be agnostic/future-proof to different hardware compilation schemes, we forgo idling noise in this work.

We perform syndrome decoding with the Relay-BP message-passing algorithm \cite{Relay_BP}, accelerated with a graphics processing unit (GPU). The GPU Relay-BP decoder \cite{cudaq_qec} is configured with the following shared parameters:
\begin{itemize}[itemsep=-3pt, topsep=5pt]
    \item $\textsf{gamma0} = 0.125$
    \item $\textsf{pre\_iter} = 80$
    \item $\textsf{num\_sets} = 1000$
    \item $\textsf{max\_iterations} = 60$
    \item $\textsf{stop\_nconv} = 5$
\end{itemize}
The gamma distribution \textsf{gamma\_dist} is individually tuned for each code's memory circuit through a grid search: see Table \ref{tab:relay-gamma-intervals}. The parallelization in the decoding algorithm happens within each BP iteration: all messages from qubit to check nodes and vice versa are performed in parallel, a schedule known as ``flooding''. The relay legs are still performed sequentially. We use single-precision floating point for the messages with $\textsf{proc\_float}=\text{`fp32'}$. We generate the syndrome data with \textsf{Stim} \cite{Stim} and decode the syndromes with the GPU-accelerated Relay-BP decoder from the \textsf{cudaq-qec} library \cite{cudaq_qec}. For the surgery circuit, we use the same relay configuration as that for the data memory code. Due to the inherent open detector boundaries of the surgery circuit, we found that we needed to double the number of relay legs ($\textsf{num\_sets} \rightarrow 2000$) in order to suppress decoder nonconvergences.

For our candidate ZSZ-LP codes, under our noise model, we observe a pseudothreshold around $p\approx0.5\%$, below which the logical error rates drop steeply with increasing code sizes. For comparison, we also plot the logical error rates of the BB-144 ``gross'' code decoded with the same Relay-BP parameters above but with $\mathsf{gamma\_dist} = [-0.24, 0.66]$ \cite{Relay_BP}, as well as the $d=13$ rotated surface code decoded with \textsf{PyMatching} \cite{Higgott2025sparseblossom}. Under this noise model, we see that the ZSZ-LP-390 code begins to outperform the gross code around $p\approx0.2\%$, with the ZSZ-LP-550 code crossing around $p\approx0.3\%$. For ZSZ-LP-390's $X$-surgery experiment, we observe that the single-logical graph-surgery gadget has comparable performance as that of the data code memory but with a slightly shallower LER slope. We note that this gap could potentially be narrowed by fine-tuning the relay parameters specifically for the surgery circuit rather than simply borrowing those from the memory circuit.

\begin{table}[t]
\centering
\begin{tabular}{cccccc}
\toprule
\textbf{Code name} &
  \textbf{\makecell{DEM size\\ $(N_{\rm det},N_{\rm loc})$}} &
  \textbf{\makecell{Mean\\ latency}} &
  \textbf{\makecell{p99\\ latency}} &
  \textbf{\makecell{p99.9\\ latency}} &
  \textbf{\makecell{Max \\ latency}} \\ \midrule
ZSZ-LP-320 & $(1920,28992)$ & 1.2 ms & 2.0 ms & 3.2 ms & 201 ms \\
ZSZ-LP-390 & $(2652,40326)$ & 1.5 ms & 2.6 ms & 3.9 ms & 260 ms \\
ZSZ-LP-550 & $(4180,63910)$ & 2.2 ms & 3.3 ms & 4.6 ms & 393 ms \\ \bottomrule
\end{tabular}
\caption{The mean, 99th percentile (p99), 99.9th percentile (p99.9) and worst-case (max) decoding latencies are listed for three candidate ZSZ-LP codes for $d$ rounds of syndrome extraction at $p=0.1\%$ using GPU Relay-BP decoding on an NVIDIA L40S GPU (18,176 CUDA cores, 48\,GB GDDR6 VRAM). Reported times for mean, p99 and p99.9 are $XZ$-averaged over 100,000 $X$-memory and $Z$-memory shots; max latency is averaged over 100 shots. We also list the dimensions of the detector error model (DEM) matrices used for decoding; $N_{\rm det}$ and $N_{\rm loc}$ are the numbers of detectors and error locations respectively.}
\label{tab:decoding latency}
\end{table}

In practice, we not only want an accurate decoder, but also one that can process syndromes faster than they are being generated to avoid the exponential backlog problem \cite{Terhal_2015_review}. To this end, we benchmark the time, or latency, it takes for the GPU Relay-BP decoder to decode a single spacetime block of $d+1$ syndromes at $p=0.1\%$ (averaged over $X$-memory and $Z$-memory) for three candidate ZSZ-LP codes, with the results reported in Table \ref{tab:decoding latency}. We observe that most shots reach the required number of convergences ($\mathsf{stop\_nconv}=5$) in a short number of relay legs, leading to a typical decoding latency of only a few milliseconds. To characterize the tail of the latency distribution, we also report the 99th and 99.9th percentile times, which are also a few milliseconds for the codes that we benchmarked. Even so, there will be rare events such as nonconvergences that cause all relay legs to be run. To benchmark this worst-case scenario, we set $\textsf{stopping\_criterion} = \text{`All'}$, which forces the Relay-BP decoder to run all relay legs regardless of convergence; this latency upper bound is reported as ``max latency'' in the final column of Table \ref{tab:decoding latency}. The max latency can be reduced by decreasing the number of relay legs but at the expense of accuracy. Recalling the estimated syndrome extraction times in Table \ref{tab:greedy AOD cost}, we conclude that this GPU Relay-BP configuration should be sufficient for sliding-window decoding with a commit region of 1, where $\sim d$ rounds of syndromes need to be decoded every syndrome extraction round \cite{Dennis_2002}, for reconfigurable atom arrays under our greedy AOD scheduler. There is also ample room to encompass faster syndrome extraction schedules or additional relay legs like for the surgery circuit. For the extremely rare nonconvergent events, which may cause more than $d$ rounds of syndromes to accumulate with the current latencies, we can use parallel-window decoding \cite{Skoric_2023_parallel}. In the future, we expect these decoding latencies to improve with both software and hardware upgrades on the GPU side. For comparison, the average CPU (AMD EPYC 7R13 Processor) decoding latencies for ZSZ-LP-320, ZSZ-LP-390 and ZSZ-LP-550 were $27\,$ms, $39\,$ms and $73\,$ms respectively, representing speedups of roughly 23x, 26x, and 33x of the L40S GPU over the EPYC 7R13 CPU. A field-programmable gate array (FPGA) or an application-specific integrated circuit (ASIC) may improve latencies even further by compiling a DEM and message-passing algorithm directly into silicon \cite{Valls_2021, Das_2022, Bascones_2025, maurer2025real, Maan_2026}; a simple DEM-extrapolation of the latencies in \cite{maurer2025real} for the gross code gives estimated latencies of $2.6\,\mu$s, $3.7\,\mu$s and $5.9\,\mu$s, assuming full parallelization. But a downside of hardware-compiled decoders is that they are inherently more rigid than general-purpose chips like GPUs, and it is unclear whether they can be adapted to the surgery setting where the number of distinct DEMs can become extensive \cite{Litinski_2019}.

\subsection{Classical channel coding}

Since our classical ZSZ-2BGA codes may be of independent interest for classical error correction, we benchmark them in the paradigmatic  setting of classical channel coding. Here, unlike the previous scenario for quantum memory, we assume noise only acts on the physical bits in between rounds of \emph{noiseless} encoding and decoding, a protocol that describes transmission between, say, a mobile device and a wireless router or cell tower. A message of $k$ bits of information is encoded into a codeword of $n>k$ bits, which is then transmitted over a noisy communication channel and finally decoded by the receiver. For the channel, we use binary phase-shift keying (BPSK) over the additive white Gaussian noise (AWGN) channel, the canonical noise model for classical communications \cite{ModernCodingTheory}. Each codeword bit $c_i\in\F_2$ is modulated to a unit-energy symbol $x_i = (-1)^{c_i}$, and the receiver observes $y_i = x_i + z_i$ with independent Gaussian noise $z_i \sim \mathcal{N}(0,\sigma^2)$ of variance $\sigma^2 = N_0/2$. This model closely approximates realistic communication links: receiver noise aggregates many independent microscopic processes (e.g. thermal noise) and is therefore Gaussian to an excellent approximation by statistical mechanics, i.e. the central limit theorem. As such, it serves as the reference benchmark for wireless standards such as fifth-generation new radio (5G NR) \cite{Richardson_2018}. The noise strength is conventionally quantified by the signal-to-noise ratio $E_b/N_0$ of the energy spent per information bit to the noise power spectral density; for unit-energy symbols, $E_b/N_0 = 1/(2R\sigma^2)$, quoted in decibels. Crucially, unlike in the quantum setting where our syndrome measurements are digitized, each received symbol is a continuous variable and carries \emph{soft} information about its bit through the log-likelihood ratio (LLR) $\log[P(c_i=0\,|\,y_i)/P(c_i=1\,|\,y_i)] = 2y_i/\sigma^2$, which a message-passing decoder can leverage directly. Since all considered codes are linear and the channel and decoder are symmetric, we may transmit the all-zero codeword without loss of generality; a block error is declared whenever the decoded word differs from the transmitted one.

Shannon's channel coding theorem states that (arbitrarily) reliable communication at rate $R$ is possible if and only if $R$ is below the channel capacity. For BPSK inputs, the capacity of the AWGN channel per channel use is
\begin{align}
    C_{\rm BPSK}(\sigma) \;=\; 1 - \mathbb{E}_L\!\left[\log_2\!\left(1+{\rm e}^{-L}\right)\right],
    \qquad L \sim \mathcal{N}\!\left(2/\sigma^2,\,4/\sigma^2\right),
    \label{eq:bpsk awgn capacity}
\end{align}
where $L$ is distributed as the LLR of a received sample given a transmitted $0$, and $\mathcal{N}(\mu,\sigma^2)$ denotes the normal distribution with mean $\mu$ and variance $\sigma^2$. Setting $C_{\rm BPSK}=R=1/2$ and solving numerically yields the Shannon limit $(E_b/N_0)_\star \approx 0.19$ dB for rate-1/2 codes, indicated by the vertical dashed line in the right panel of Figure \ref{fig:classical zsz vs random dist}. The channel capacity thus sets an ultimate limit on the threshold of any error-correcting code family under the noise channel, independent of any particular decoder. Note that no analogue exists for the depolarizing channel in the quantum setting, and finding one is still an active area of research.

To perform soft-decision decoding over the AWGN channel, we use the GPU-accelerated BP decoder \textsf{LDPCBPDecoder} from the \textsf{Sionna} library \cite{sionna_github, sionna_paper}, configured with the following parameters:
\begin{itemize}[itemsep=-3pt, topsep=5pt]
    \item $\textsf{num\_iter} = 100$
    \item $\textsf{cn\_update} = \text{`boxplus-phi'}$
    \item $\textsf{llr\_max} = 20$
    \item $\textsf{precision} = \text{`single'}$
\end{itemize}
As with the GPU Relay-BP decoder, a ``flooding'' schedule is executed where all messages are updated in parallel within each iteration. The `boxplus-phi' option implements a sum-product check-node update, input LLRs are clipped to $\pm\,\textsf{llr\_max}$, and messages are stored in single-precision floating point (fp32).

We compare the ZSZ-2BGA codes against two rate-1/2 baselines with the same $[310,155]$ parameters as the largest ZSZ-2BGA code: \emph{(i)} the best $(3,6)$-regular PEG instance from the 1000-code ensemble in the left panel of Figure \ref{fig:classical zsz vs random dist}, with \textsf{QDistEvol}-estimated minimum distance 20 \cite{PEG_LDPC}, and \emph{(ii)} a degree-irregular 5G-NR LDPC code constructed with \textsf{Sionna} according to the 3GPP specification \cite{Richardson_2018}: base graph 2 (BG2) with lift size $\ell=26$, 105 shortened (filler) bits, and $2\ell=52$ punctured high-degree ``state'' variable nodes that are never transmitted. The 5G code's irregular degree distribution trades distance and girth for a near-capacity threshold and superior performance in the waterfall regime. It is decoded on its full pruned Tanner graph with zero LLRs injected at the punctured positions; since the punctured bits consume no channel energy, all three codes transmit 310 symbols per 155 information bits at equal $E_b/N_0$. Eliminating the punctured variables yields a parity-check matrix of the transmitted $[310,155]$ code with \textsf{QDistEvol}-estimated minimum distance only 11, a deliberate tradeoff in the 5G design that favors BP-decoding pseudothresholds (waterfall performance) and rate flexibility over girth and minimum distance (error-floor performance) \cite{Richardson_2018}.

The right panel of Figure \ref{fig:classical zsz vs random dist} plots the resulting block error rates. The ZSZ-2BGA-310 code performs nearly identically to the best PEG code throughout the waterfall region, while the 5G-NR code enjoys roughly a 0.5 dB advantage there owing to its density-evolution-optimized irregular degree profile and punctured state nodes. The tables turn in the low-error regime: beyond $E_b/N_0\approx4$\,dB the 5G curve visibly flattens toward an error floor governed by its low minimum distance and girth of 6, with 65\% of its recorded failures at $E_b/N_0\geq3.75$\,dB being undetected logical errors, whereas the ZSZ-2BGA and PEG codes exhibit no undetected errors at all. Consequently, the ZSZ-2BGA-310 curve crosses the 5G-NR curve around $4.5$\,dB and attains the lowest measured block error rate of ${\approx}\,1.6\times10^{-7}$ at $4.75$\,dB with no visible error floor, consistent with its larger minimum distance (22 versus 11) and girth (8 versus 6). The smaller ZSZ-2BGA codes display the expected finite-length scaling, with waterfalls that steepen with increasing block length $n$, past an observed pseudothreshold around $E_b/N_0 \approx 1.25\,$dB, which is within $\approx1.1\,$dB of the Shannon limit. We conclude that the ZSZ-2BGA family is competitive with the strongest random $(3,6)$-regular constructions (i.e. from PEG) under soft-decision decoding at these block lengths, and that its high minimum distance and girth make it preferable in applications demanding very low residual error rates.

%%%%%%%%%%%%%%%%%%%%%%%%%%%%%%%%%%%%%%%%%%%%%%%%%%%%%%%

\section{Outlook}\label{sec:outlook}

This work serves as an initial investigation into ZSZ-LP codes, and there are plenty of additional directions on both the theoretical and practical fronts. We briefly describe some of them below.

\textbf{More bounds on minimum distance.} We have only provided theoretical upper bounds on minimum distances in Appendices \ref{app:classical 2BGA codes} and \ref{app:abelian-low-weight-logicals}. It would be interesting if any theoretical lower bounds, beyond girth-based ones, exist for both classical ZSZ-2BGA and quantum ZSZ-LP codes. In addition, our upper bounds are coarse-grained to the level of the entire ZSZ group; tighter bounds can likely be derived by taking into account the particular choices of monomials, which may provide a smarter method for code discovery.

\textbf{Minimum-weight and symplectic logical bases.} The one-generator minimum-weight (OGM) and systematic/symplectic (OGS) logical bases that we have provided in Tables \ref{tab:OGM generators} and \ref{tab:OGS generators} are distinct from one another. It would be of further practical relevance if a ZSZ-2BGA code pair and corresponding ZSZ-LP code instance is found to admit a logical basis that is both OGM and OGS, in analogy to the small quasi-cyclic codes in \cite[Table 3]{Xu_2025_fast} and the GB codes in \ref{webster2025explicit}. Such an OGM+OGS basis could further lower the practical spacetime footprint in FTQC compilation.

\textbf{Improved syndrome extraction.} For syndrome extraction, we focused on schedules suited for reconfigurable atom arrays. Perhaps the most immediate practical improvement along these lines would be to lower the rearrangement or routing complexity when implementing syndrome extraction. Faster syndrome extraction directly translates to faster logical clock speeds in the FTQC. One avenue is to take the current codes and improve upon the simple greedy scheduler by finding even lower-cost routes and/or interleaving $X$-syndrome and $Z$-syndrome extraction like for the BB codes. Moving the data qubits in addition to the check qubits may also reduce the total routing time. Another avenue is to co-design the ZSZ polynomials so that the rearrangement costs become inherently cheap, in analogy to what has been done with Kasai's APM codes \cite{kasai2026breaking, zhao2026ultra}. Recent progress on morphing circuits for group-algebra codes may also improve both routing overheads and pseudothresholds \cite{Shaw_2025_morphing_1, shaw2026morphing2, chao2026morphing}. Modifying syndrome-extraction circuits to handle atom loss is also pertinent for neutral-atom implementations \cite{liu2026loss}. In addition to these software-side optimizations, hardware upgrades such as three-dimensional AODs \cite{Picard_2026_3D, guo2025aod3d} could significantly lower the rearrangement costs by parallelizing certain subroutines such as row-dependent cyclic shifts, which are currently executed sequentially. Overall, we believe there is substantial room for future improvements on neutral-atom implementation. We also did not touch on trapped-ion implementations such as in a quantum charged-coupled device (QCCD) \cite{Kielpinski_2002, Pino_2021}, which have different movement constraints, and leave this exploration to future work. Hardware constraints aside, it would be interesting if distance-preserving (and equivariant) syndrome extraction schedules exist, and furthermore if they can have minimal circuit depth 12.\footnote{Since both $H_X$ and $H_Z$ \eqref{eq:5-block H_X,H_Z} are $(6,9)-LDPC$, the maximum qubit degree is 12, which lower-bounds the circuit depth.}

\textbf{Single-shot decoding.} Another way to boost the logical clock speed is to lower the number of syndrome rounds per decoding cycle. $\Theta(d)$ rounds are usually required so that timelike (measurement) and spacelike error mechanisms have similar protection. Previous work on two-block ZSZ codes \cite{ZSZ_codes} numerically demonstrated that ZSZ group-algebra codes can possess the ``confinement'' \cite{Quintavalle_2021_conf} needed for single-shot decoding, where a reliable correction can be inferred with only a single round of noisy syndrome information \cite{Spielman_1995_linear, Bombin_2015_single}. In practice, we may use a constant number of rounds greater than one depending on target performance. It would also be interesting if the surgery gadgets can be made single-shot \cite{cowtan2025fast, chang2026constant}; see the later discussion on locally testable codes for one possible direction. We leave the exploration of single-shot gadgets to future work.

\textbf{Improved decoding.} Despite the already good performance of Relay-BP for our candidate ZSZ-LP codes, we observe that nearly all decoding failures are due to BP nonconvergence rather than actual logical errors. For example, at $p=0.2\%$, only 1 of the 25 observed decoding failures for ZSZ-LP-390 was an undetected logical error. A nice property about a nonconverged failure is that it is flagged, and so we know about it in real time. There are several ways to further lower the observed LER. One way is to simply run more relay legs and hope for more convergences at the cost of increased decoder latency, like what we did for the surgery circuit. Another way is to postprocess the nonconvergences with more powerful but expensive techniques such as ordered-statistics (OSD) or even stronger decoders such as Tesseract \cite{tesseract_decoder} or Frontier \cite{frontier_decoder}. Another possibility would be to extend machine-learning (ML) decoders for BB codes \cite{gu2026cascade} to ZSZ-LP codes, though the lack of obvious translation symmetry may cause difficulties. One could also try to preprocess the syndromes using a predecoder \cite{ai_predecoder_surface} with the goal of reducing the frequency of nonconvergences. Recent progress on importance sampling for QEC \cite{beverland2025fail} may also help us better fine-tune our relay parameters (e.g. the gamma interval) for low error rates, where traditional Monte-Carlo becomes computationally expensive. On the GPU side, at the time of our simulations, fp32 was the lowest supported precision on the \textsf{cudaq-qec} Relay-BP decoder. Recent results on FPGAs demonstrate that lower-bit precision can increase throughput while maintaining nearly identical performance for parallelized Relay-BP \cite{maurer2025real}, and it would be interesting to see if these findings carry over to GPUs as well. Similar to the outlook on improved syndrome extraction, we also anticipate significant progress on this front.

\textbf{FTQC compilation.} Could ZSZ-LP codes eventually reduce the FTQC spacetime footprint when compiling large-scale algorithms, compared to current schemes? As mentioned in the introduction, a smaller qLDPC block allows for more intrinsic interblock parallelism, which could potentially lower the time overhead when compiling algorithms. However, due to the slower logical clock speeds and lack of code automorphisms compared to other codes such as BB codes, the end-to-end practical advantage is unclear. It would be interesting to compare the costs of a full extractor system \cite{he2025extractors} for ZSZ-LP codes. The concurrent work \cite{mitten_codes} may already provide answers to some of these questions.

\textbf{Locally testable classical ZSZ-LP codes.} In the realm of classical coding theory, it is known that the lifted/balanced product can also produce classical locally testable codes with constant rate, constant relative distance and constant query complexity (${\rm c}^3$-LTC) \cite{PK_good_qLDPC, Dinur_2022_LTC}. A ${\rm c}^3$-LTC is essentially an asymptotically good LDPC code with an additional ``local testing'' property that allows a receiver to reliably test whether a received bitstring is a codeword by only querying a constant number of bits \cite{Spielman_1995_linear}. Candidate LTCs can be obtained by simply replacing $H_{\rm right}$ with its transpose in the balanced product \eqref{eq:H_X,H_Z balanced product} and then relabeling the $Z$-check, qubit and $X$-check nodes in the quantum CSS Tanner graph to the bit, check and ``metacheck'' nodes of a classical LTC's Tanner graph; i.e. $H^\transp_Z \rightarrow H^{}_{\rm LTC}$. In this manner, we obtain candidate LTCs with block length $n=4\ell$ and code dimension $k\geq\ell$; i.e. design rate 1/4. As an example, using the polynomials of ZSZ-LP-320 results in a $(6,6)$-regular $[256,64,\leq54]$ code, whose local testability requires further analysis.\footnote{Probing the local testability, or ``soundness'' profile, of an LTC would involve estimating its so-called locally minimal distance \cite{PK_good_qLDPC, Dinur_2023_good}.} Motivated by the strong finite-length properties of ZSZ-LP qLDPC codes, we anticipate that the same construction could possibly lead to small-length LTCs with good performance. In FTQC, ${\rm c}^3$-LTCs have recently been shown theoretically to enable single-shot code surgery on generic qLDPC codes \cite{zhang2026accelerating}. The discovery of practical finite-length instances could hence bring this theoretical proposal to the practical realm, and we leave such studies to future work.

\section*{Acknowledgments}

We thank Rui Chao, Adam Holmes, Vadym Kliuchnikov, Muyuan Li, Kyungjoo Noh and Jan Olle for valuable discussions on practical qLDPC design, and we thank Igor Baratta, Christopher Chamberland, Ben Howe, Antonio deMarti iOlius, Justin Lietz and Melody Ren for valuable discussions on qLDPC decoding. We thank the CUDA-Q software team for continued development on GPU-accelerated decoders. We thank Qian Xu and Hsin-Yuan Huang for bringing our attention to the concurrent works and coordinating the joint arXiv release. Finally, we thank Krysta Svore for feedback on the manuscript and creating a research environment that made this work possible.

\textbf{AI disclosure:} we acknowledge LLM assistance in conducting numerical simulations, brainstorming, proof verifications and proofreading; the author assumes full responsibility for all content.

\addcontentsline{toc}{section}{References}
\begin{small}
\bibliographystyle{alpha}
\bibliography{thebib}
\end{small}

\newpage
\appendix

\section{Additional ZSZ-LP codes}\label{app:more ZSZ-LP codes}

Table \ref{tab:more ZSZ-LP codes} lists additional ZSZ-LP codes along the efficiency--length Pareto frontier in the left panel of Figure \ref{fig:pareto frontiers}. Note that most of these code instances have rather large twist factors ($q$) and so may not be as amenable for reconfigurable atom arrays compared to the $q=2,3$ candidates in Table \ref{tab:ZSZ-LP codes intro}. Nonetheless, we list them here since they could potentially be useful for other hardware architectures such as ion traps.

\begin{table}[t]
\centering\renewcommand{\arraystretch}{1.3}
\resizebox{\textwidth}{!}{
\begin{tabular}{ccccccc}
\toprule
$\llbracket n,k,d \rrbracket$ & $\ell_1,\ell_2,q$ & $a$ & $b$ & $c$ & $d$ & \textbf{OGS} \\ \midrule
$\llbracket 60,12,6 \rrbracket$ & $3,4,2$ & $1+y^3+x^2y^3$ & $1+x+xy^2$ & $1+y+xy^3$ & $1+y^2+y^3$ & \cmark \\
$\llbracket 100,20,8 \rrbracket$ & $5,4,4$ & $1+x^2y+y^3$ & $1+x^2+xy^3$ & $1+y^3+x^4y^3$ & $1+x^4y^2+xy^3$ & \cmark \\
$\llbracket 150,30,10 \rrbracket$ & $15,2,11$ & $1+x^{11}+x^{12}y$ & $1+y+x^6y$ & $1+x^3+x^{11}$ & $1+x^6+x^{14}y$ & \cmark \\
$\llbracket 210,42,12 \rrbracket$ & $21,2,8$ & $1+xy+x^3y$ & $1+x^{12}+x^{17}$ & $1+x^4y+x^{15}y$ & $1+x^9y+x^{18}y$ & \cmark \\
$\llbracket 525,105,\leq18 \rrbracket$ & $35,3,11$ & $1+x^{22}+xy^2$ & $1+x^4+x^{16}y^2$ & $1+x^{34}+x^{12}y$ & $1+x^{28}y+x^{31}y$ & \xmark \\
$\llbracket 625,125,\leq20 \rrbracket$ & $25,5,6$ & $1+x^6+xy$ & $1+x^{24}y+x^{21}y^3$ & $1+x^3y^2+x^{11}y^3$ & $1+x^4+x^6y^2$ & \cmark \\
$\llbracket 700,140,\leq22 \rrbracket$ & $35,4,8$ & $1+x^{29}y+x^2y^2$ & $1+x^7y+x^9y$ & $1+x^{19}+x^{16}y$ & $1+x^{32}y^2+x^{23}y^3$ & \xmark \\
$\llbracket 840,168,\leq24 \rrbracket$ & $28,6,11$ & $1+y+x^5y^2$ & $1+x^2y^3+x^{19}y^3$ & $1+x^{24}y+x^2y^2$ & $1+x^{26}y^2+x^{27}y^3$ & \xmark \\ \bottomrule
\end{tabular}
}
\caption{Additional ZSZ-LP codes on the numerically observed $kd^2/n$ frontier. The first column lists their code parameters, followed by the ZSZ group parameters $(\ell_1,\ell_2,q)$ and the trinomials $a,b,c,d$. Exact and upper bounds on code distances are computed using the \textsf{pySATDist} and \textsf{QDistEvol} ($10^6$ iterations) methods respectively. The final column indicates whether the code admits a one-generator-symplectic (OGS) logical Pauli basis.}
\label{tab:more ZSZ-LP codes}
\end{table}

\section{Binary group algebras and protomatrices}

In this appendix, we introduce some concepts and tools from abstract algebra, which hopefully provides the adequate context for the group-algebraic manipulations in the later appendices. To readers familiar with linear-algebraic manipulations for linear codes, these techniques are essentially generalizations of those over finite fields to group algebras. Working over group algebras is a convenient way to prove many properties of the codes described in the main text, especially in a way that respects the codes' group symmetries.

Let $G$ be a finite group of order $|G|$. Its regular representation describes how $G$ acts on itself by translation. To define it, let $V=\F_2^{|G|}$ be the vector space with basis vectors $\{\ket{h}:h\in G\}$ indexed by the group elements. Each $g\in G$ defines the permutation operators
\begin{align}\label{eq:regular representation}
    L[g]\ket{h}=\ket{gh}\, ,\qquad R[g]\ket{h}=\ket{hg}
\end{align}
called the left-regular and right-regular representations respectively.\footnote{Sometimes, the right-regular representation is defined with respect to group inverses so that the internal ordering $R[a]R[b]=R[ab]$ is preserved, but it nonetheless introduces unnecessary visual clutter with inverses, and so we avoid it in this work.} Note that left and right are the same when $G$ is abelian but can be distinct when $G$ is non-abelian. Since matrix multiplication always proceeds in right-to-left order, we have $L[a]L[b]=L[ab]$ and $R[a]R[b]=R[ba]$. Thus the elements of $G$ are realized as $|G|\times|G|$ permutation matrices, allowing group multiplication to be studied using linear algebra.

Let
\begin{align}
        \Ring := \F_2[G]
\end{align}
denote the binary group ring of $G$, which for $\F_2$ coefficients is also a group algebra.  Every element of $\Ring$ has a unique expansion
\begin{align}
        p=\sum_{g\in G} p_g g \, ,\qquad p_g\in \F_2 \, .
\end{align}
We denote the coefficient Hamming weight of $p$ by
\begin{align}
        \wt(p) := \abs{\{g\in G:p_g=1\}} \, ,
\end{align}
i.e. the number of nonzero group elements in its expansion. We denote
\begin{align}
        \bar{p} := \sum_{g\in G} p_g g^{-1}
\end{align}
for the antipode of $p$. Note that the antipode map is an involution but is not the same as inversion over the group ring, since $p\bar{p}\neq1$ in general. We say that $p \in \Ring$ is left-invertible (right-invertible) if there exists $q \in \Ring$ such that $qp=1$ ($pq=1$). For finite rings, a left inverse automatically implies a right inverse, and furthermore both inverses are equal\footnote{Let $qp=pq'=1$. From associativity, we have $q=q(pq')=(qp)q'=q'$.} and $p$ is called a unit; we denote its unique inverse by $p^{-1}$. We say that $\Ring$ is commutative if $pq=qp$ for all $p,q\in\Ring$, which holds if and only if $G$ is abelian.

A left ideal of $\Ring$ is an additive subgroup $I\subseteq \Ring$ such that $rI\subseteq I$ for every $r\in\Ring$, and a right ideal is defined similarly. An ideal will mean a two-sided ideal unless otherwise specified. When the ambient ring is commutative, these notions coincide. Given elements $p_1,\ldots,p_t$ in a commutative ring $\mathsf{S}$, the ideal generated by them is
\begin{align}
        \braket{p_1,\ldots,p_t}
        := \left\{\sum_{i=1}^t s_i p_i : s_i\in\mathsf{S}\right\}.
\end{align}
An ideal generated by one element is called principal; thus
\begin{align}
        \braket{p}=\mathsf{S}p=\{sp:s\in\mathsf{S}\}
\end{align}
is the principal ideal generated by $p$. An ideal $I\subseteq\mathsf{S}$ is proper if $I\neq\mathsf{S}$, and it is maximal if it is proper and there is no proper ideal strictly between $I$ and $\mathsf{S}$. A commutative ring $\mathsf{S}$ is \emph{local} if it has a unique maximal ideal, denoted by $\mathfrak{m}$. Equivalently, the elements outside $\mathfrak{m}$ are precisely the units of $\mathsf{S}$. If $I\subseteq\mathsf{S}$ is an ideal and $M$ an $\mathsf{S}$-module, then $IM$ denotes the submodule of $M$ consisting of finite sums of elements $im$ with $i\in I$ and $m\in M$. We say that $M$ is finitely generated if there exist $m_1,\ldots,m_t\in M$ such that every element of $M$ can be written as $\sum_i s_i m_i$ with $s_i\in\mathsf{S}$.
Finally, when $G$ is abelian so $\Ring=\F_2[G]$ is a finite commutative ring, it is Artinian and hence decomposes as a finite product of Artinian local rings.

Like for group elements, let $L[a], R[a]\in\F^{|G|\times|G|}_2$ denote the left-regular and right-regular representations of ring element $a\in\Ring$ respectively. A valuable property of the regular representation is that there is an isomorphism between length-$|G|$ binary vectors and elements of $\Ring=\F_2[G]$, owing to \ref{eq:regular representation}, and so many operations on our codes' binary matrices can be reduced to operations on the ring. For example, linear combinations of vectors in $\F^{|G|}_2$ corresponds to linear combinations of ring elements.

For a matrix $M$ over $\Ring$, define
\begin{align}
        M^\dagger := \overline{M}^\transp,
\end{align}
that is, transpose together with ring involution. A protomatrix is a small matrix over $\Ring$, such as $A\in \Ring^{r\times c}$.. ``Lifting'' $A$ to binary involves replacing each group-ring entry in $A$ with its regular representation, giving a binary matrix of size $r|G|\times c|G|$.  If the entries of $A$ have bounded coefficient weight, then the expanded matrix is sparse since the regular representation consists of permutation matrices with unit row and column weights.
For our purposes, we will be interested in $1\times2$ protomatrices, i.e. $r=1$ and $c=2$, for the classical parity-check matrices in Appendix \ref{app:classical 2BGA codes} and $2\times5$ protomatrices for the quantum CSS parity-check matrices in Appendix \ref{app:ZSZ-LP codes}.

%%%%%%%%%%%%%%%%%%%%%%%%%%%%%%%%%%%%%%%%%%%%%%%%%%%%%%%%%%%%%

\section{Classical $G$-lifted LDPC codes}
\label{app:classical 2BGA codes}

A classical $G$-lifted LDPC code is nothing more than taking a protomatrix and lifting it into a binary parity-check matrix through $\Ring=\F_2[G]$. For protomatrix $H_{\rm proto} \in \Ring^{r\times c}$ and lift size $|G|$, we obtain a (left) parity-check matrix $H = L[H_{\rm proto}] \in \F^{r|G|\times c|G|}_2$. The code rate $k/n$ is lower-bounded by the \emph{design rate}:
\begin{align}
    \frac{k}{n} \geq \frac{n-\text{\#rows}(H)}{n} = \frac{c|G|-r|G|}{c|G|} = 1-\frac{r}{c} \, ,
\end{align}
which is nonzero when $r<c$. If $H$ has full (row) rank, then the rate is equal to the design rate. For our scenario of interest ($r=1$, $c=2$), our classical lifted LDPC codes will have design rate 1/2. Since $H_{\rm proto}$ consists of two block matrices, we will call these codes classical two-block group-algebra (2BGA) codes. We will also focus on the scenario where the group-algebra elements have coefficient weight 3, i.e. trinomials.

For a $1\times2$ protomatrix, the parity-check matrix takes the two-block form
\begin{align}\label{eq:2BGA H}
    H = \begin{pmatrix}
        \,A & B\,
    \end{pmatrix}
\end{align}
where $A,B\in\F^{|G|\times|G|}_2$ are regular representations of $a,b\in\Ring$ respectively, with $\wt(a)=\wt(b)=\wt(A)=\wt(B)=3$. Here $\wt(A)$ denotes the maximum row and column weight of binary matrix $A$. With slight abuse of notation, we use $H$ for both its ring and binary presentations, with distinction given by context. For abelian $G$, there is only one kind of regular representation. For non-abelian $G$, we need to distinguish between $L[a]$ and $R[a]$ for $A$. As proposed in the main text, we define $H_{\rm left} := L[H]$ and $H_{\rm right} := R[H]$.

\subsection{Abelian lifts and low distance}

\begin{prop}[Abelian $G$ implies constant code distance]
\label{prop:abelian distance bound}
    Suppose $\mathcal{C}$ is a linear code whose parity-check matrix is given by \eqref{eq:2BGA H} with abelian lift group $G$ and nonzero group-algebra elements $a,b \in \Ring$. Then $d(\mathcal{C}) \leq \wt(a)+\wt(b)$.
\end{prop}

\begin{proof}
    Let $S=(\bar{b}\;\;\bar{a})$ be an $\Ring$-valued matrix of weight $\wt(S) = \wt(\bar{a}) + \wt(\bar{b}) = \wt(a)+\wt(b)$. Then $HS^\dagger=ab+ba=0$ since $G$ is abelian, which implies $HS^\transp=0$ for their regular representations. Thus, the rows of $S$ are contained in $\mathcal{C}$, and we conclude that $d(\mathcal{C}) \leq \wt(a)+\wt(b)$.
\end{proof}

In the context of abelian group algebras, $S=(\bar{b}\;\;\bar{a})$ is known as a Koszul syzygy of $H$. When $a$ and $b$ are trinomials, then we have $d\le6$. Another way to see this low-weight syzygy is to observe that when $G$ is abelian, \eqref{eq:2BGA H} is simply the $X$-check matrix $H_X$ of a quantum 2BGA code \cite{2BGA_codes}. The Koszul syzygy is then precisely the $Z$-check matrix $H_Z$. This argument also extends to the non-abelian case if we use opposite-type regular representations such as $A=L[a]$ and $B=R[b]$ since $[A,B]=0$ from associativity of left and right group multiplication. The only way around is to have non-abelian $G$ as well as $A$ and $B$ be both left-regular or both right-regular representations such that $AB\neq BA$.

\subsection{Distance upper bounds for ZSZ groups}
\label{app:ZSZ distance bounds}

As argued in the main text and shown formally in Proposition \ref{prop:abelian distance bound}, abelian groups lead to poor distance of the resulting classical 2BGA codes, which motivates the use of non-abelian groups such as ZSZ. However, some non-abelian groups are ``almost abelian'', quantified by the size of their commutator subgroup. The smaller the commutator subgroup is, the closer the group is to being abelian. Intuitively, we would expect the minimum distances of classical 2BGA codes built from ``almost abelian'' groups to also be small. We formalize this intuition by using commutator subgroups. We will focus on left 2BGA codes with parity-check matrix $H_{\rm left} = (L[a] \;\;L[b])$, but we note that the proof techniques easily generalize to $H_{\rm right}$ by interchanging the side of multiplication.

\begin{thm}[Commutator-subgroup distance bound]
\label{thm:commutator distance bound}
    Let $G$ be a finite group, let $\Ring=\F_2[G]$, and let $a,b\in\Ring$ be nonzero. Consider the classical left 2BGA code with parity-check matrix $H_{\rm left} = (L[a] \;\;L[b])$.
    Let $[G,G]$ be the commutator subgroup. Then
    \begin{align}\label{eq:commutator subgroup distance bound}
        d \leq \big(\!\wt(a)+\wt(b)\big)\,\big|[G,G]\big| .
    \end{align}
\end{thm}

\begin{proof}
    We identify a length-$\abs{G}$ binary column vector with an element of the group algebra $\Ring$. Under this identification, a pair $(u,v)\in\Ring^2$ is a codeword of $H_{\rm left} = (L[a]\;\;L[b])$ precisely when
    \begin{align}
        au+bv=0 .
    \end{align}
    Let $\mathcal{D} = [G,G]$ denote the commutator subgroup. Define the ``Casimir'' element
    \begin{align}
        \Omega_{\mathcal{D}}:=\sum_{\delta\in\mathcal{D}}\delta \in \Ring .
    \end{align}
    Since $G/\mathcal{D}$ is abelian by definition, for every $g,h\in G$ there exists some $\delta(g,h)\in\mathcal{D}$ such that
    \begin{align}
        gh=hg\,\delta(g,h).
    \end{align}
    Also, since $\Omega_\mathcal{D}$ is the sum of all elements in $\mathcal{D}$, left multiplication by any element of $\mathcal{D}$ simply permutes the terms in $\Omega_{\mathcal{D}}$; i.e. $\delta\Omega_{\mathcal{D}}=\Omega_{\mathcal{D}}$. Therefore
    \begin{align}
        gh\Omega_{\mathcal{D}}
        =hg\,\delta(g,h)\Omega_{\mathcal{D}}
        =hg\Omega_{\mathcal{D}} .
    \end{align}
    By bilinearity, this implies
    \begin{align}
        ab\Omega_{\mathcal{D}}=ba\Omega_{\mathcal{D}},
    \end{align}
    or equivalently, since we work over the binary field $\F_2$,
    \begin{align}\label{eq:commutator averaged syzygy}
        (ab+ba)\Omega_{\mathcal{D}}=0 \, .
    \end{align}
    It follows that
    \begin{align}
        a\big(b\Omega_{\mathcal{D}}\big)+b\big(a\Omega_{\mathcal{D}}\big)
        =(ab+ba)\Omega_{\mathcal{D}}
        =0 \, ,
    \end{align}
    and so the vector
    \begin{align}
        v_{\mathcal{D}} := \big(b\Omega_{\mathcal{D}}, \;a\Omega_{\mathcal{D}}\big)
    \end{align}
    is a codeword in $\ker H_{\rm left}$. Its weight is bounded by
    \begin{align}
        \wt(v_{\mathcal{D}})
        &\leq \wt\!\big(b\Omega_{\mathcal{D}}\big) + \wt\!\big(a\Omega_{\mathcal{D}}\big) \notag \\
        &\leq \wt(b)\abs{\mathcal{D}} + \wt(a)\abs{\mathcal{D}} \notag \\
        &=\big(\!\wt(a)+\wt(b)\big)\abs{\mathcal{D}} \, ,
    \end{align}
    where in the second line we used the fact that the support of $a\Omega_{\mathcal{D}}$ is contained in the union of at most $\wt(a)$ left cosets of $\mathcal{D}$, each of size $\abs{\mathcal{D}}$; cancellations over $\F_2$ can only reduce the weight.

    There is an edge/degenerate case where $v_\mathcal{D}=0$, corresponding to $a\Omega_{\mathcal{D}} = b\Omega_{\mathcal{D}} = 0$. Since $a$ and $b$ are not both zero by construction, the vector $u_\mathcal{D} = (\Omega_{\mathcal{D}},0)$ is then a nonzero codeword because
    \begin{align}
        H u_\mathcal{D} = a\Omega_{\mathcal{D}}+b\cdot0=0 \, ,
    \end{align}
    and it too has weight $\abs{u_\mathcal{D}} = |\mathcal{D}| \leq (\wt(a)+\wt(b))\abs{\mathcal{D}}$.
\end{proof}

\begin{cor}[ZSZ commutator-subgroup distance bound]
\label{cor:ZSZ commutator distance bound}
    For $G = \ZSZ_{\ell_1,\ell_2,q} = \mathbb{Z}_{\ell_1} \rtimes_q \mathbb{Z}_{\ell_2}$,
    \begin{align}
        [G,G]=\langle x^{q-1}\rangle\;,\qquad \big|[G,G]\big| = \frac{\ell_1}{\gcd(\ell_1,q-1)} ,
    \end{align}
    and hence
    \begin{align}\label{eq:ZSZ commutator distance bound}
        d \leq \big(\!\wt(a)+\wt(b)\big)\frac{\ell_1}{\gcd(\ell_1,q-1)} .
    \end{align}
\end{cor}

\begin{proof}
    We specialize Theorem \ref{thm:commutator distance bound} to $G=\ZSZ_{\ell_1,\ell_2,q}$. The relation $yxy^{-1}=x^q$ implies
    \begin{align}
        yxy^{-1}x^{-1}=x^{q-1} \, ,
    \end{align}
    so $\langle x^{q-1}\rangle\subseteq [G,G]$. Conversely, $\langle x^{q-1}\rangle$ is normal in $G$ since $yx^{q-1}y^{-1} = (x^{q-1})^q \in \langle x^{q-1}\rangle$, and so after quotienting $G$ by $\langle x^{q-1}\rangle$, the relation $yxy^{-1}=x^q$ becomes $yxy^{-1}=x$; hence the quotient is abelian and so $[G,G]\subseteq \langle x^{q-1}\rangle$. Therefore $[G,G]=\langle x^{q-1}\rangle$. Since $x$ has order $\ell_1$, the subgroup generated by $x^{q-1}$ has size $\ell_1/\gcd(\ell_1,q-1)$, giving \eqref{eq:ZSZ commutator distance bound}.
\end{proof}

For trinomials, $\wt(a)=\wt(b)=3$, so Corollary \ref{cor:ZSZ commutator distance bound} gives
\begin{align}\label{eq:trinomial ZSZ dist bound}
    d \leq 6\abs{[G,G]} = 6\frac{\ell_1}{\gcd(\ell_1,q-1)} \leq 6\ell_1 .
\end{align}
Thus, when $\ell_1\ll\ell_2$, the distance of the classical two-block code is upper-bounded by a scale set by the small normal $\mathbb{Z}_{\ell_1}$ direction, even though the block length is $2\ell_1\ell_2$. The abelian case $q=1$ has trivial $[G,G]=\{1\}$ and recovers the syzygy bound $d \leq \wt(a)+\wt(b)$ of Proposition \ref{prop:abelian distance bound}. The same proof also gives the identical bound for the right-regular two-block code $H_{\rm right} = (R[a]\;\;R[b])$, with the averaged codeword written as $(\Omega_{[G,G]}b,\Omega_{[G,G]}a)$ under the convention that $R[\cdot]$ acts by right multiplication. Applying \eqref{eq:trinomial ZSZ dist bound} to the classical ZSZ-2BGA codes in Table \ref{tab:classical ZSZ code params} gives distance upper-bounds of 12, 12, 48, 78, 66, 186 respectively. So we see that \eqref{eq:trinomial ZSZ dist bound} can be quite tight for small-length ($n<100$) ZSZ-2BGA codes but is likely loose for the larger block lengths.

\subsection{Tanner graph cycles and girth bounds}
\label{app:classical tanner girth}

\begin{defn}[Tanner graph girth]
    Let $T = (V \cup C, E)$ be the bipartite Tanner graph of a parity-check matrix $H \in \mathbb{F}_2^{|C| \times |V|}$, whose two parts are the variable (bit/qubit) nodes $V$ and the check nodes $C$, with $(v,c) \in E$ if and only if $H_{cv} = 1$. Let $\partial_{ve} \in \mathbb{F}_2^{(|V|+|C|) \times |E|}$ be the vertex--edge incidence matrix of $T$, so that $\ker \partial_{ve}$ is the space of closed cycles on $T$. The girth of $T$ is then defined as
    \begin{equation}
        \operatorname{girth}(T) \;:=\; \min_{z \,\in\, \ker \partial_{ve} \setminus \{0\}} |z| ,
    \end{equation}
    with the convention $\operatorname{girth}(T) = \infty$ when $T$ is acyclic. Because $T$ is bipartite, every closed cycle alternates between $V$ and $C$ and therefore has even length; hence $\operatorname{girth}(T)$ is even and at least 4.
\end{defn}

We bound the girth of the left 2BGA code $H_{\rm left}=(L[a]\;\;L[b])$ from \eqref{eq:H_left,H_right}, whose Tanner graph is $(3,6)$-regular; the same statements hold verbatim for $H_{\rm right}=(R[c]\;\;R[d])$ with left and right multiplication interchanged. Throughout we write the trinomials as sums of distinct group elements $a=a_1+a_2+a_3$ and $b=b_1+b_2+b_3$, normalized so that $a_1=b_1=1$.

Before invoking any group structure, we note a generic graph-theoretic ceiling. Because the Tanner graph is $(3,6)$-regular and has only $\ell$ check nodes, the Moore bound on graph girth, based on counting the number of vertices within a fixed graph distance of any node until a cycle is closed, forces the girth to grow no faster than logarithmically in the block length: $\operatorname{girth}(H_{\rm left}) \leq 4\log_{10}\ell + O(1) = O(\log n)$. This bound is oblivious to both the group and the choice of trinomials. The remainder of this appendix will then quantify how the girth can be further limited by the arithmetic of the ZSZ group, which is the operative constraint at the finite block lengths of interest.

The starting point is the observation that every cycle of the Tanner graph corresponds to a relation among monomial transitions. A closed walk alternating between check and bit nodes advances from one check to the next by descending to a shared bit node and re-ascending; if the bit node lies in the block $A=L[a]$ and is incident to the monomials $a_i\neq a_j$, then the two check labels $c,c'$ it connects satisfy $c'=a_ja_i^{-1}c$, and analogously $c'=b_jb_i^{-1}c$ for the block $B=L[b]$. Thus each check-to-check move left-multiplies the check label by an element of the \emph{transition set}
\begin{align}\label{eq:classical transition set}
    \mathcal{T} := \big\{a_ia_j^{-1} : i\neq j\big\}\cup\big\{b_ib_j^{-1} : i\neq j\big\}\subseteq G\setminus\{1\} \, ,
\end{align}
and a cycle of length $2m$ is precisely a relation $s_m\cdots s_1=1$ with each $s_t\in\mathcal{T}$. The Tanner girth is therefore twice the length of the shortest such nontrivial relation. A single low-order transition already suffices to close a short cycle, which gives the following upper bound.

\begin{prop}[Element-order girth bound for left 2BGA codes]
\label{prop:classical-element-order-girth}
    Let $G$ be a finite group and let $a,b\in\F_2[G]$ be trinomials with distinct monomial terms. Then for every transition $s\in\mathcal{T}$ in \eqref{eq:classical transition set}, the Tanner graph of $H_{\rm left}=(L[a]\;\;L[b])$ contains a simple cycle of length $2\operatorname{ord}_G(s)$. Consequently,
    \begin{align}
        \operatorname{girth}(H_{\rm left})\leq 2\min_{s\in\mathcal{T}}\operatorname{ord}_G(s)\leq 2\max_{g\in G}\operatorname{ord}_G(g) \, .
    \end{align}
\end{prop}

\begin{proof}
    Fix a transition $s=a_ja_i^{-1}\in\mathcal{T}$ arising from the block $A=L[a]$; the argument for a transition of $B=L[b]$ is identical. Let $o=\operatorname{ord}_G(s)$, and note $o\geq2$ since $a_i\neq a_j$ implies $s\neq1$. We use the convention that $L[g]$ connects the bit vertex labeled $h$ to the check vertex labeled $gh$. Pick any check vertex $c_0\in G$ and define the check and bit vertices
    \begin{align}
        c_k = s^k c_0\, ,\qquad v_k = a_i^{-1}s^k c_0 \qquad (k=0,1,\ldots,o-1) \, ,
    \end{align}
    where each $v_k$ lies in the first bit block. Then $v_k$ is incident to
    \begin{align}
        a_i v_k = s^k c_0 = c_k\, ,\qquad
        a_j v_k = a_ja_i^{-1}s^k c_0 = s^{k+1}c_0 = c_{k+1}\, ,
    \end{align}
    with the index $k+1$ taken modulo $o$. Hence
    \begin{align}
        c_0 \,\text{--}\, v_0 \,\text{--}\, c_1 \,\text{--}\, v_1 \,\text{--}\, \cdots \,\text{--}\, c_{o-1} \,\text{--}\, v_{o-1} \,\text{--}\, c_0
    \end{align}
    is a closed walk of length $2o$. The check vertices $c_k=s^kc_0$ are pairwise distinct because they form the orbit of $c_0$ under the order-$o$ element $s$, and the bit vertices $v_k=a_i^{-1}s^kc_0$ are pairwise distinct for the same reason; all bit vertices moreover lie in the same block. The walk is therefore a simple Tanner cycle of length $2o$, so $\operatorname{girth}(H_{\rm left})\leq 2o$. Minimizing over $s\in\mathcal{T}$ yields the first inequality, and the second follows since $\operatorname{ord}_G(s)\leq\max_{g}\operatorname{ord}_G(g)$.
\end{proof}

Proposition \ref{prop:classical-element-order-girth} shows that the girth is throttled by the smallest order attained by any transition monomial. For random trinomials the transitions are essentially uniform group elements, and two effects then compete. Generic, accidental coincidences among the six transitions, e.g. an $A$-transition happening to equal a $B$-transition, occur at an expected rate that is independent of $\ell$ and matches that of a random $(3,6)$-regular ensemble; their density therefore vanishes as $\ell\to\infty$. We also have the group-induced short cycles of Proposition \ref{prop:classical-element-order-girth}: whenever the group contains many low-order elements, a random transition is likely to be one of them, and by the free right action \eqref{eq:H_left symmetry} each such transition contaminates the graph with an entire orbit of $\Theta(\ell)$ short cycles. The girth of a random ZSZ-2BGA code is thus governed by the density of low-order elements, since a transition of order $j$ closes a cycle of length $2j$---an involution ($j=2$) a $4$-cycle, an order-$3$ element a $6$-cycle, and so on.

The low-order content of a ZSZ group is dictated by its parameters $\ell_1,\ell_2,q$ through the push-through relations \eqref{eq:push-through relations}. Writing a group element in canonical form $x^py^r$ and iterating \eqref{eq:push-through relations} gives $(x^py^r)^{k}=x^{p\Sigma_k}y^{kr}$ with $\Sigma_k=\sum_{t=0}^{k-1}q^{tr}$. Setting $\alpha=\ell_2/\gcd(\ell_2,r)$ for the order of $y^r$ and $\sigma_r=\sum_{t=0}^{\alpha-1}q^{tr}\bmod\ell_1$, one obtains the closed form
\begin{align}\label{eq:zsz element order}
    \operatorname{ord}_G(x^py^r) = \alpha\cdot\frac{\ell_1}{\gcd(\ell_1,\,p\,\sigma_r)} \, .
\end{align}
In particular, the number of involutions is
\begin{align}\label{eq:zsz involution count}
    \nu_2(G) = \mathbb{I}[\ell_1\text{ even}] + \mathbb{I}[\ell_2\text{ even}]\cdot\gcd\!\big(\ell_1,\,1+q^{\ell_2/2}\big) \, ,
\end{align}
where $\mathbb{I}[\cdot]$ is the indicator function that evaluates to 1 if its argument is true and 0 otherwise, and the analogous count $\nu_3(G)$ of order-$3$ elements is governed by the cyclotomic-like sums $\gcd(\ell_1,\,1+q^r+q^{2r})$. The involution count \eqref{eq:zsz involution count} exposes a sharp degeneracy: whenever $\ell_2$ is even and $q^{\ell_2/2}\equiv-1\pmod{\ell_1}$, like in the dihedral case $\ell_2=2,\,q=\ell_1-1$, the gcd jumps to $\ell_1$, saturating the group with $\Theta(\ell_1)$ involutions and forcing girth $4$ for almost every trinomial pair.

\begin{cor}[Dihedral girth--distance dichotomy]
\label{cor:dihedral dichotomy}
    Let $G = \mathbb{Z}_{\ell_1}\rtimes_q\mathbb{Z}_2$ with $q=\ell_1-1$ be the dihedral group of order $\ell=2\ell_1$, and let $a,b\in\F_2[G]$ have at least two distinct monomial terms. Then
    \begin{align}
        \operatorname{girth}(H_{\rm left})=4
        \qquad\text{or}\qquad
        d \leq \wt(a)+\wt(b) \, .
    \end{align}
    In particular, a dihedral 2BGA code can never simultaneously achieve non-minimal girth and distance exceeding the abelian bound.
\end{cor}

\begin{proof}
    Without loss of generality, we work with respect to the normalization $a_1=b_1=1$. Decompose the dihedral group into rotations and reflections: $G=\langle x\rangle\cup\langle x\rangle y$. By the push-through relations \eqref{eq:push-through relations} with $q=-1$, every reflection is an involution: $(x^py)^2=x^{p-p}y^2=1$. First suppose some trinomial contains a reflection, say $a_j\in\langle x\rangle y$. The transition $a_ja_1^{-1}=a_j\in\mathcal{T}$ then has order 2, so Proposition \ref{prop:classical-element-order-girth} produces a Tanner $4$-cycle, and hence $\operatorname{girth}(H_{\rm left})=4$. Otherwise, every monomial is a rotation, i.e. $a,b\in\F_2[\langle x\rangle]$. Left multiplication by $\langle x\rangle$ preserves the two cosets $\langle x\rangle$ and $\langle x\rangle y$, so the Tanner graph of $H_{\rm left}$ splits into two disconnected components, each a copy of a 2BGA code over the cyclic group $\langle x\rangle\cong\mathbb{Z}_{\ell_1}$. Proposition \ref{prop:abelian distance bound} applied to either copy then gives $d\leq\wt(a)+\wt(b)$.
\end{proof}

Corollary \ref{cor:dihedral dichotomy} thus informs us that we should probably discount dihedral groups from consideration. Conversely, taking $\ell_1$ and $\ell_2$ both odd eliminates all involutions outright. Table \ref{tab:classical low-order census} lists the order-$2$ and order-$3$ densities $\nu_2/\ell$ and $\nu_3/\ell$ for the six classical ZSZ-2BGA seed codes, alongside their left Tanner girths. The two codes that attain girth $8$, ZSZ-2BGA-220 and ZSZ-2BGA-310, are precisely those for which both densities are negligible, whereas ZSZ-2BGA-156 is pinned at girth $6$ by its large population of order-$3$ elements ($\nu_3/\ell=1/3$) despite carrying only a single involution.

\begin{table}[t]
\centering\renewcommand{\arraystretch}{1.1}
\begin{tabular}{ccccccc}
\toprule
\textbf{Code name} & $(\ell_1,\ell_2,q)$ & $\nu_2$ & $\nu_3$ & $\nu_2/\ell$ & $\nu_3/\ell$ & \textbf{Girth} \\ \midrule
ZSZ-2BGA-64 (left)  & $(16,2,9)$ & 3 & 0  & 0.094 & 0.000 & 6 \\
ZSZ-2BGA-96 (left)  & $(12,4,7)$ & 3 & 2  & 0.062 & 0.042 & 6 \\
ZSZ-2BGA-128 (left) & $(16,4,3)$ & 3 & 0  & 0.047 & 0.000 & 6 \\
ZSZ-2BGA-156 (left) & $(26,3,3)$ & 1 & 26 & 0.013 & 0.333 & 6 \\
ZSZ-2BGA-220 (left) & $(22,5,3)$ & 1 & 0  & 0.009 & 0.000 & 8 \\
ZSZ-2BGA-310 (left) & $(31,5,2)$ & 0 & 0  & 0.000 & 0.000 & 8 \\ \bottomrule
\end{tabular}
\caption{Low-order element census for the left ZSZ-2BGA codes of Table \ref{tab:classical ZSZ code params}. For each group $\ZSZ_{\ell_1,\ell_2,q}$ of order $\ell=\ell_1\ell_2$, we list the number of order-$2$ elements $\nu_2$ and order-$3$ elements $\nu_3$ obtained from \eqref{eq:zsz element order}, their densities $\nu_2/\ell$ and $\nu_3/\ell$, and their Tanner girths.}
\label{tab:classical low-order census}
\end{table}

These observations suggest a simple heuristic for maximizing the attainable girth: choose the ZSZ parameters $(\ell_1,\ell_2,q)$ so that low-order elements are scarce. Choosing the group order $\ell=\ell_1\ell_2$ coprime to $6$ eliminates every element of order $2$ or $3$ via \eqref{eq:zsz element order}, so that no monomial transition can force a $4$- or $6$-cycle; the only short cycles that remain are accidental cross-block coincidences, whose vanishing density makes girth $6$ and $8$ progressively easier to attain as $\ell$ grows. When divisibility by $2$ or $3$ is unavoidable, one should at least steer clear of the twist degeneracies $q^{\ell_2/2}=-1$ and $1+q^r+q^{2r}\equiv0\pmod{\ell_1}$, which saturate $G$ with $\Theta(\ell_1)$ low-order elements and plague the code with 4-cycles and 6-cycles. This girth criterion is complementary to the distance bound of Corollary \ref{cor:ZSZ commutator distance bound}: whereas a large minimum distance calls for a large commutator subgroup $[G,G]=\langle x^{q-1}\rangle$, a large girth calls for a small population of low-order elements, and the two are largely independent knobs on the ZSZ group algebra. The $\text{ZSZ}_{31,5,2}$ group of ZSZ-2BGA-310, whose order $155$ is coprime to $6$, is particularly favorable on both counts.

\subsection{Equivariant codewords}
\label{app:codeword orbits}

In this appendix, we explicitly report orbit generators for both one-generator-minimal (OGM) and one-generator-systematic (OGS) codeword bases for each classical ZSZ-2BGA code in Table \ref{tab:classical ZSZ code params}. The OGM and OGS codewords then translate to weight-$d$ and symplectic logical Pauli operators of the quantum ZSZ-LP codes respectively through the embedding \eqref{eq:codeword embedding}. We also discuss conditions for when an OGS basis exists, which when compared to the minimum-weight basis, provides a tradeoff between sparsity and utility. For an element $x^ay^b$ in a ZSZ group with parameters $(\ell_1,\ell_2,q)$, we assign the integer $x^ay^b \mapsto a\ell_2+b \in [0,\ell)$, which we use to construct the regular representations, i.e. index the columns of the subblocks in \eqref{eq:H_left,H_right} and \eqref{eq:5-block H_X,H_Z}.

We use the convention
\begin{align}
    G_{\rm left}=(L[u_L]\;\;L[v_L])\, ,\qquad
    G_{\rm right}=(R[u_R]\;\;R[v_R])\, ,
\end{align}
where the rows of $G_{\rm left}$ and $G_{\rm right}$ are the automorphism orbits of the corresponding left and right seed codewords. Explicitly, under the identification of vector coordinates with group elements, the row of $G_{\rm left}$ at group element $g$ is the ring pair $(\bar{u}_L\,g\;\;\bar{v}_L\,g)$, and that of $G_{\rm right}$ is $(g\,\bar{u}_R\;\;g\,\bar{v}_R)$, since the rows of the lifted blocks $L[p]$ and $R[p]$ correspond to the ring elements $\bar{p}\,g$ and $g\,\bar{p}$ respectively. These generator polynomials represent valid codewords when
\begin{align}
    G^{}_{\rm left} H_{\rm left}^{\transp}=0\, ,\qquad
    G^{}_{\rm right} H_{\rm right}^{\transp}=0\, .
\end{align}
Equivalently, in the ZSZ group algebra,
\begin{align}\label{eq:generator validity}
    u_L\bar{a}+v_L\bar{b}=0\, ,\qquad
    \bar{c}\,u_R+\bar{d}\,v_R=0\, ,
\end{align}
where the second identity takes into account that our convention of right-regular representation is anti-homomorphic under row multiplication: $R[g]R[h] = R[hg]$.

\subsubsection{Minimum-weight generators}

Table \ref{tab:OGM generators} lists the polynomial pairs used for the minimum-weight orbit generators. These polynomials were found via numerical search for weight-$d$ codewords using \textsf{QDistEvol}. For the first five codes, the orbit rank equals $k_{\rm left} = k_{\rm right} = \ell_1\ell_2 = \ell$. For ZSZ-LP-775, the displayed single orbit has rank $140$ for the left code and $129$ for the right code, so additional orbit generators would be needed to span the full classical codespace.

\begin{table}[t]
\centering
\resizebox{\textwidth}{!}{%
\begin{tabular}{ccc}
\toprule
\textbf{Code name} & \textbf{$u_L$} & \textbf{$v_L$} \\ \midrule
ZSZ-LP-160 & $x^8+x^{11}+(x^4+x^8+x^{13})y$ & $1+x^{10}+x^{15}+(x^3+x^{11})y$ \\
ZSZ-LP-240 & $x^4+x^6y+(x^8+x^9)y^2+x^4y^3$ & $x^3+x^4+x^7+(1+x^3+x^8+x^{10})y$ \\
ZSZ-LP-320 & $x^7+(1+x^7+x^{13})y+(x^4+x^7)y^2+(x^4+x^9+x^{14})y^3$ & $y+x^9y^2+(x^3+x^{10}+x^{13})y^3$ \\
ZSZ-LP-390 & $x+x^{14}+(x^8+x^{16}+x^{24})y+(x^2+x^3)y^2$ & $x^5+x^7+x^{25}+(x^4+x^8+x^{12})y+(x^{12}+x^{17}+x^{20})y^2$ \\
ZSZ-LP-550 & $x^{13}+x^{18}y+(x^8+x^{10})y^2+(1+x^3+x^7+x^{11}+x^{15})y^4$ & $x^{13}+x^{20}y+(x^5+x^8+x^{11}+x^{14}+x^{17})y^2+(x^4+x^{16})y^4$ \\
ZSZ-LP-775 & $x^{11}+x^{25}+x^{25}y+(x^{11}+x^{15})y^2+(1+x^2+x^3)y^3+(x^2+x^{17}+x^{27})y^4$ & $x^5+x^{11}+x^{13}+x^{13}y+(x^5+x^{20}+x^{25})y^3+(x+x^{12}+x^{22}+x^{27})y^4$ \\ \bottomrule
\end{tabular} } \\
\vspace{1em}
\resizebox{\textwidth}{!}{%
\begin{tabular}{ccc}
\toprule
\textbf{Code name} & \textbf{$u_R$} & \textbf{$v_R$} \\ \midrule
ZSZ-LP-160 & $x^3+x^8+x^{13}+(x^3+x^5+x^7+x^{10})y$ & $x^5+x^8+x^{10}y$ \\
ZSZ-LP-240 & $x^8+x^{10}+x^{11}+x^6y+x^4y^2$ & $x+x^6+x^8+x^{10}+y+(x+x^6)y^3$ \\
ZSZ-LP-320 & $x^8+x^9+x^6y+(x^2+x^7+x^{11}+x^{14})y^2+(x^{11}+x^{12})y^3$ & $x^8+(x^7+x^{12}+x^{14})y^2+x^{12}y^3$ \\
ZSZ-LP-390 & $x^8+x^{11}+x^{24}+(x^2+x^{23})y+(x^{14}+x^{21})y^2$ & $1+x^3+x^6+(x^4+x^9+x^{13}+x^{18})y+(x^{16}+x^{23})y^2$ \\
ZSZ-LP-550 & $x^5+x^{12}+(x^8+x^{13})y+(x^{12}+x^{17}+x^{18})y^2+(1+x+x^7+x^8)y^3$ & $x^{21}+(x^{14}+x^{15})y+(x^{11}+x^{19})y^2+(x^8+x^{17})y^3$ \\
ZSZ-LP-775 & $x^{22}+(x^2+x^3+x^{19}+x^{25}+x^{26})y+(x^7+x^{12}+x^{22})y^2+x^{27}y^3+(x^8+x^{24})y^4$ & $x^{12}+x^{16}+(x^2+x^{15}+x^{25})y+x^{27}y^2+(x^{13}+x^{14}+x^{20})y^3+x^{29}y^4$ \\ \bottomrule
\end{tabular} }
\caption{Minimum-weight OGM generator polynomials for candidate ZSZ-2BGA codes. The top table gives $G_{\rm left}=(L[u_L]\;\;L[v_L])$ for $\Aut(\mathbf{c_{\rm left}})$, while the bottom table gives $G_{\rm right}=(R[u_R]\;\;R[v_R])$ for $\Aut(\mathbf{c_{\rm right}})$.}
\label{tab:OGM generators}
\end{table}

\subsubsection{Systematic generators}
\label{app:systematic codewords}

We now discuss when the orbit generators can be brought into \emph{one-generator systematic} (OGS) form, where one block of the generator matrix is the identity:
\begin{align}\label{eq:OGS generators}
    G^{\rm sys}_{\rm left}=\big(\,\ident\;\;L[v^{\rm sys}_L]\,\big)\, ,\qquad
    G^{\rm sys}_{\rm right}=\big(\,\ident\;\;R[v^{\rm sys}_R]\,\big)\, .
\end{align}
Such a form is convenient for classical encoding: the identity block constitutes an information set, so the message appears verbatim in the first block of every codeword, while the parity block is generated by a single polynomial. Setting $u_L=1$ in the validity condition $u_L\bar{a}+v_L\bar{b}=0$ \eqref{eq:generator validity} forces $v^{\rm sys}_L\bar{b}=\bar{a}$, which admits the unique solution
\begin{align}
    v^{\rm sys}_L=\bar{a}\,\bar{b}^{-1}=\overline{b^{-1}a}
\end{align}
when $b$ is invertible in $\Ring$. Likewise, $v^{\rm sys}_R=\overline{c\,d^{-1}}$ exists when $d$ is invertible, and the mirrored forms with the identity in the second block require invertible $a$ or $c$ respectively. The corresponding seed codewords are $(1\;\;b^{-1}a)$ and $(1\;\;c\,d^{-1})$. Since the identity block guarantees rank $\ell$, an invertible $b$ (or $a$) certifies that the classical left code is one-generator, with no numerical search required. Conversely, when $b$ is not invertible, it is a zero divisor of the finite-dimensional algebra $\Ring$, and every nonzero $v\in\ker{L[b]}$ furnishes a single-block codeword $(0\;\;v)$; the minimum weight of this kernel then upper-bounds the code distance, a liability that a distance-filtered code search must implicitly weed out.

\begin{table}[t]
\centering\renewcommand{\arraystretch}{1.25}
\resizebox{\textwidth}{!}{%
\begin{tabular}{ccc}
\toprule
\textbf{Code name} & \textbf{$v^{\rm sys}_L$} & \textbf{$\wt(v^{\rm sys}_L)$} \\ \midrule
ZSZ-LP-160 & $1+x+x^4+x^6+x^7+x^9+x^{10}+x^{12}+x^{13}+x^{15}+(x^2+x^6+x^7+x^8+x^9+x^{10}+x^{14})y$ & 17 \\[0.2em] \Xhline{0.10pt} \noalign{\vskip 0.2em}
ZSZ-LP-240 & \makecell{$1+x+x^6+x^7+x^8+x^{10}+x^{11}+(1+x^2+x^6+x^9+x^{10})y$ \\ $+(x+x^3+x^5+x^6+x^7+x^8)y^2+(1+x+x^4+x^9+x^{10})y^3$} & 23 \\[0.2em] \Xhline{0.10pt} \noalign{\vskip 0.2em}
ZSZ-LP-320 & \makecell{$1+x^2+x^4+x^{11}+x^{14}+x^{15}+(x+x^3+x^4+x^7+x^9+x^{14})y$ \\ $+(1+x+x^2+x^3+x^5+x^6+x^7+x^8+x^9+x^{10}+x^{12})y^2+(x+x^6+x^8+x^9+x^{10}+x^{13})y^3$} & 29 \\[0.2em] \Xhline{0.10pt} \noalign{\vskip 0.2em}
ZSZ-LP-390 & \makecell{$x+x^2+x^3+x^4+x^9+x^{12}+x^{13}+x^{14}+x^{15}+x^{17}+x^{19}+x^{20}+x^{21}+x^{22}+x^{24}$ \\ $+(1+x^2+x^3+x^5+x^8+x^9+x^{11}+x^{12}+x^{14}+x^{18}+x^{20}+x^{23}+x^{24})y$ \\ $+(x^3+x^5+x^6+x^9+x^{10}+x^{12}+x^{13}+x^{14}+x^{15}+x^{16}+x^{17}+x^{19}+x^{20}+x^{22}+x^{23}+x^{24}+x^{25})y^2$} & 45 \\ \bottomrule
\end{tabular} } \\
\vspace{1em}
\resizebox{\textwidth}{!}{%
\begin{tabular}{ccc}
\toprule
\textbf{Code name} & \textbf{$v^{\rm sys}_R$} & \textbf{$\wt(v^{\rm sys}_R)$} \\ \midrule
ZSZ-LP-160 & $1+x+x^3+x^4+x^{10}+x^{12}+x^{13}+x^{15}+(x^3+x^4+x^6+x^7+x^9+x^{11}+x^{13}+x^{14}+x^{15})y$ & 17 \\[0.2em] \Xhline{0.10pt} \noalign{\vskip 0.2em}
ZSZ-LP-240 & \makecell{$1+x+x^6+x^7+x^8+x^{10}+x^{11}+(x^4+x^5+x^6+x^8+x^9+x^{10}+x^{11})y$ \\ $+(x^2+x^3+x^4+x^6+x^7+x^8+x^9)y^3$} & 21 \\[0.2em] \Xhline{0.10pt} \noalign{\vskip 0.2em}
ZSZ-LP-320 & \makecell{$1+x^2+x^4+x^8+x^{12}+x^{14}+(x^2+x^5+x^6+x^7+x^8+x^{10}+x^{12}+x^{13}+x^{14})y$ \\ $+(x+x^9+x^{11})y^2+(1+x^3+x^4+x^5+x^8+x^{10}+x^{11}+x^{12}+x^{14})y^3$} & 27 \\[0.2em] \Xhline{0.10pt} \noalign{\vskip 0.2em}
ZSZ-LP-390 & \makecell{$x+x^3+x^4+x^6+x^8+x^{11}+x^{12}+x^{13}+x^{14}+x^{18}+x^{20}+x^{21}+x^{23}$ \\ $+(x+x^3+x^5+x^8+x^9+x^{10}+x^{11}+x^{12}+x^{16}+x^{22}+x^{23}+x^{25})y$ \\ $+(x+x^6+x^7+x^9+x^{10}+x^{11}+x^{12}+x^{13}+x^{14}+x^{15}+x^{16}+x^{17}+x^{18}+x^{19}+x^{23}+x^{25})y^2$} & 41 \\ \bottomrule
\end{tabular} }
\caption{Systematic generator polynomials for the candidate codes admitting the OGS form \eqref{eq:OGS generators} along with their coefficient weights. The top table lists $v^{\rm sys}_L$ for $G^{\rm sys}_{\rm left}=(\ident\;\;L[v^{\rm sys}_L])$, while the bottom table lists $v^{\rm sys}_R$ for $G^{\rm sys}_{\rm right}=(\ident\;\;R[v^{\rm sys}_R])$. The corresponding codeword rows have weight $1+\wt(v^{\rm sys})$.}
\label{tab:OGS generators}
\end{table}

Whether the required units exist depends strongly on the underlying group. If the ZSZ group is a 2-group, i.e. $\ell=\ell_1\ell_2$ is a power of two, then $\Ring$ is a local ring whose unique maximal ideal is the augmentation ideal, and an element of $\Ring$ is invertible if and only if it has odd coefficient weight \cite{Passman_1977}. Every trinomial is then automatically invertible, and the OGS form always exists. Away from 2-groups, the unit density appears to decrease with the odd part of $\ell$: sampling 1000 random (normalized) trinomials for the groups underlying ZSZ-LP-240, 390, 550 and 775, we find that roughly $74\%$, $78\%$, $32\%$ and $4\%$ are invertible respectively. In the last case, $\ell=155$ is odd, so $\Ring$ is semisimple and splits into several matrix components that must each be invertible, which a random sparse element rarely satisfies. Among the candidate codes, all four trinomials of ZSZ-LP-160 and ZSZ-LP-320 are automatically invertible (2-groups), and ZSZ-LP-240 and ZSZ-LP-390 have invertible $b$ and $d$, so the first four candidates all admit OGS forms for both their left and right codes. In contrast, none of the $a,b,c,d$ trinomials is invertible for ZSZ-LP-550 and ZSZ-LP-775, even though $H_{\rm left}$ and $H_{\rm right}$ retain full rank; these two codes admit no OGS form in either orientation, and any systematic encoder must use an information set that is not aligned with the two blocks, thereby losing the group-algebra structure. We also emphasize that the OGS generator is not minimum-weight: its rows have weight $1+\wt(b^{-1}a)$, with the weight of $b^{-1}a$ being close to $\ell/2$ for our examples. The OGS and minimum-weight orbit generators are therefore different bases of the same code that trade systematic encoding against sparsity; see Appendix \ref{app:fast encoding} for the corresponding encoding strategies. As a side benefit, since every OGS row is a codeword, $d\leq1+\wt(b^{-1}a)$ provides a cheap distance upper bound that can be used to pre-filter random code searches. The systematic generator polynomials $v^{\rm sys}_L$ and $v^{\rm sys}_R$ for the four candidate codes admitting them are listed in Table \ref{tab:OGS generators} along with their weights, and one can straightforwardly verify that $G^{\rm sys}_{\rm left}H^\transp_{\rm left}=G^{\rm sys}_{\rm right}H^\transp_{\rm right}=0$ for every entry. Comparing against Table \ref{tab:OGM generators}, the systematic generators are considerably denser than their minimum-weight counterparts; e.g. for ZSZ-LP-320, the OGS generator rows have weights 30 and 28, versus the minimum weight $d=14$.

\subsection{One-generator codes and fast encoding}
\label{app:fast encoding}

We now elaborate on the fast-encoding claims made in the main text, using the equivariant generator matrices constructed in the previous subsection. We focus on $H_{\rm left}$ for simplicity. Recall that the rows of $G_{\rm left}=(L[u_L]\;\;L[v_L])$ form the orbit $\{(\bar{u}_L\,g\;\;\bar{v}_L\,g) : g\in G\}$ of the seed codeword $(\bar{u}_L\;\;\bar{v}_L)$ under the free right ZSZ action \eqref{eq:H_left symmetry}, which for a minimum-weight seed spans the entire codespace for every candidate code in Table \ref{tab:classical ZSZ code params} except that of ZSZ-LP-775. Taking row combinations, the codespace is precisely the set of pairs $(\bar{u}_L\,m\;\;\bar{v}_L\,m)$ over all $m\in\Ring$, i.e. the cyclic right $\Ring$-submodule of $\Ring^2$ generated by the seed codeword. This structure is the direct generalization of a classical cyclic code, whose codewords $m(x)g(x)$ form the principal ideal generated by a single generator polynomial $g(x)$ in $\F_2[x]/(x^n-1)$, as well as of one-generator quasi-cyclic codes \cite{Ling_2001}; hence the name \emph{one-generator} in the previous subsection. Codes realized as ideals and modules in group algebras are known as group-algebra codes \cite{Berman_1967, MacWilliams_1970} and have been studied for metacyclic groups such as ZSZ \cite{Sabin_1995}.

Encoding with an orbit generator is fast because its rows are translates of a single codeword. Every row of $G_{\rm left}$ has weight exactly $\wt(u_L)+\wt(v_L)$, which equals $d$ for the minimum-weight generators of Table \ref{tab:OGM generators}. Since any row of any generator matrix is itself a nonzero codeword and hence has weight at least $d$, the minimum-weight orbit basis is the sparsest generator matrix possible. Encoding a message $\mathbf{m}\in\F^\ell_2$, identified with $m=\sum_g m_g\,g \in\Ring$, amounts to two sparse group-algebra products
\begin{align}
    \mathbf{m}^\transp G_{\rm left} = (\,\bar{u}_L m\;\;\bar{v}_L m\,) \, ,
\end{align}
where the antipode $\bar{u}_L$ inverts every group element of $u_L$ and preserves its weight. Expanding $\bar{u}_Lm = \sum_{h\in\operatorname{supp}(\bar{u}_L)} hm$, each summand $hm$ is simply a permutation, i.e. a left translation, of the coefficients of $m$. Encoding therefore costs $d$ permute-and-accumulate passes over the message vector, or $O(dn)$ binary operations in total, in direct analogy with encoding a cyclic code through multiplication by a sparse generator polynomial. In contrast, a generic (e.g. random) LDPC code has a dense generator matrix and costs $O(n^2)$ binary operations to encode, although preprocessing techniques can bring this close to linear \cite{RU_2001}.

The permutations involved are moreover hardware-friendly. By the push-through relations \eqref{eq:push-through relations}, left translation by a monomial $x^\alpha y^\beta$ acts on the $\ell_1\times\ell_2$ grid of group elements (Figure \ref{fig:ZSZ cayley graphs}) as $x^iy^j \mapsto x^{\alpha+q^\beta i}\,y^{\beta+j}$: a relabeling of the $\ell_2$ grid rows, composed with one of the $\ell_2$ fixed affine maps $i\mapsto q^\beta i$ and a cyclic shift by $\alpha$ within each row. If each grid row is stored as a machine word, each pass reduces to a fixed bit permutation followed by word rotations. An encoder is therefore architecturally similar to the quasi-cyclic LDPC encoders standardized in 5G NR cellular communications \cite{Li_2006, Richardson_2018}, namely a $d$-stage pipeline of barrel shifters augmented with fixed wire permutations.

For the minimum-weight generator matrix, the orbit encoder is non-systematic: the message bits do not directly appear among the codeword bits. Nevertheless, message recovery from a decoded codeword is cheap. Since group algebras of finite groups over fields are self-injective, a full-rank orbit guarantees the existence of B\'ezout-type elements $p,s\in\Ring$ satisfying $p\bar{u}_L+s\bar{v}_L=1$, and so the message can be recovered through two further group-algebra multiplications:
\begin{align}
    p\big(\bar{u}_Lm\big)+s\big(\bar{v}_Lm\big) = \big(p\bar{u}_L+s\bar{v}_L\big)m = m \, .
\end{align}

Finally, encoding remains quasi-linear even for dense generator polynomials, e.g. should the construction be scaled to lengths where the distance $d$ grows large. Decompose the message into its $y$-components $m=\sum_{j} m_j(x)\,y^j$, and likewise a fixed multiplier $w=\sum_\beta w_\beta(x)\,y^\beta$. The push-through relations \eqref{eq:push-through relations} give
\begin{align}
    wm = \sum_{\beta,j} w_\beta(x)\,m_j\big(x^{q^\beta}\big)\, y^{\beta+j} \, ,
\end{align}
so any group-algebra product decomposes into $\ell^2_2$ cyclic convolutions of length $\ell_1$, twisted by the substitutions $m_j(x)\mapsto m_j(x^{q^\beta})$. Under a discrete Fourier transform of length $\ell_1$, each substitution acts as a mere index permutation of the Fourier coefficients, $\hat{m}_j(k) \mapsto \hat{m}_j(q^\beta k)$, where indices are modulo $\ell_1$. An encoder based on the fast Fourier transform (FFT) therefore requires $\ell_2$ forward transforms, $\ell^2_2$ pointwise multiplications and $\ell_2$ inverse transforms, costing $O\big(\ell_1\ell_2(\log\ell_1+\ell_2)\big) = O\big(n(\log\ell_1+\ell_2)\big)$ operations in total, where a lift to integer arithmetic is required when $\ell_1$ is even due to lack of roots of unity. The FFT encoder becomes advantageous over the sparse orbit encoder when $d\gg\log\ell_1+\ell_2$.

%%%%%%%%%%%%%%%%%%%%%%%%%%%%%%%%%%%%%%%%%%%%%%%%%%%%%%%%%

\section{Quantum lifted/balanced product codes}
\label{app:ZSZ-LP codes}

In this appendix, we present formal arguments to back up some of the intuition stated in the main text. Appendix \ref{app:abelian-low-weight-logicals} formalizes the notion that going to quantum cannot save the code from the low-weight abelian syzygies in the classical codes. Appendix \ref{app:tanner graph cycles} presents some upper bounds on the Tanner girth of ZSZ-LP codes; Theorem \ref{thm:symmetric-girth-4} in particular is the reason why we choose a non-symmetric balanced product and use two classical seed codes with differing trinomials; otherwise we have unavoidable short 4-cycles in our Tanner graphs which could lead to degraded BP decoding performance. Appendix \ref{app:fold-symmetric-low-distance} identifies a separate distance-three obstruction for the trivial automorphism fold. In Appendix \ref{app:codeword orbits}, we present our minimum-weight codewords derived from automorphism orbits.

\subsection{Abelian lifts and low distances}
\label{app:abelian-low-weight-logicals}

Let $G$ be a finite abelian group of order $\ell$, let $\Ring := \F_2[G]$ denote its binary group algebra, and let $a,b,c,d\in\Ring$. Let $A,B,C,D\in\F_2^{\ell\times\ell}$ denote their regular representations. Since $G$ is abelian, these matrices commute, and hence the matrices $H_X,H_Z$ in \eqref{eq:5-block H_X,H_Z} define a CSS code, which we denote by $\mathcal{Q}(a,b,c,d)$. Recall that the transpose of the binary lift corresponds to the antipode map $p\mapsto \bar{p}$ on $\Ring$. For $p,q\in\Ring$, define
\begin{align}
    \operatorname{Ann}(p) &:= \{r\in\Ring: rp=0\} \, , \\
    K(p,q) &:= p\,\operatorname{Ann}(q)\cap q\,\operatorname{Ann}(p) \, .
\end{align}

We now present two short lemmas which will be useful for our eventual Theorem \ref{thm:abelian-canonical-syzygies} which shows that not every low-weight abelian syzygy can be stabilized away.

\begin{lem}[Nakayama's lemma for local rings \cite{Nakayama_1951}]
\label{lem:nakayama}
    Let $\mathsf{S}$ be a local commutative ring with unique maximal ideal $\mathfrak{m}$, and let $M$ be a finitely generated $\mathsf{S}$-module. If
    \begin{align}
        M=\mathfrak{m}M \, ,
    \end{align}
    then $M=0$. In particular, if $x\in\mathfrak{m}$ and $M=xM$, then $M=0$.
\end{lem}

For brevity, let $(\mathsf{S},\mathfrak{m})$ denote a local commutative ring $\mathsf{S}$ with unique maximal ideal $\mathfrak{m}$.

\begin{lem}[Mutual annihilator syzygies]
\label{lem:mutual-annihilator-syzygies}
    Let $p,q,r,s\in\Ring$ with abelian $G$ as above. If
    \begin{align}
        r,s\in K(p,q)\, ,\qquad p,q\in K(r,s) \, ,
    \end{align}
    then $p=q=r=s=0$.
\end{lem}

\begin{proof}
    Since $\Ring$ is a finite commutative ring, it is Artinian and hence decomposes as a finite product of Artinian local rings. It is enough to prove the claim after projection to each local factor. Thus let $(\mathsf{S},\mathfrak{m})$ be one local factor, and with slight abuse of notation write again $p,q,r,s$ for the projected elements, with $K$ and $\operatorname{Ann}$ now computed in $\mathsf{S}$.

    Suppose, on the contrary, that $p\neq0$. Since $s\in K(p,q)$, there exists $u\in \operatorname{Ann}(q) \subseteq \mathsf{S}$ such that
    \begin{align}
        s=pu\, ,\qquad uq=0 \, .
    \end{align}
    Since $p\in K(r,s)$, there is $t\in \mathsf{S}$ such that
    \begin{align}
        p=st\, ,\qquad tr=0 \, .
    \end{align}
    Hence $p=put$, so the principal ideal $\braket{p}$ satisfies
    \begin{align}
        \braket{p} = \braket{p}ut \, .
    \end{align}
    If $ut\in\mathfrak{m}$, then Nakayama's lemma (Lemma~\ref{lem:nakayama}), applied to the finitely generated $\mathsf{S}$-module $\braket{p}$, gives $\braket{p}=0$, contradicting the assumption that $p\neq0$. Therefore $ut\notin\mathfrak{m}$, and hence $ut$ is a unit. Since $\mathsf{S}$ is commutative, $u$ and $t$ are accordingly also units. It follows from $uq=0$ and $tr=0$ that $q=r=0$. But then $p\in K(r,s)\subseteq r\mathsf{S}=0$, again a contradiction. Thus we must have $p=0$.

    Interchanging the roles of $p$ and $q$, and of $s$ and $r$, gives $q=0$. Then $r,s\in K(0,0)=0$, so $r=s=0$. Since this result holds in each local factor, we conclude that $p=q=r=s=0$ in $\Ring$.
\end{proof}

Recall the five-block form of $H_X$ and $H_Z$, which we recite below in ring notation:
\begin{subequations}\label{eq:5-block H_X,H_Z ring}
\begin{align}
    H_X &= \begin{pmatrix}\,
        a & 0 & b & 0 & \bar{c} \\
        0 & a & 0 & b & \bar{d} \,
    \end{pmatrix}  \\
    H_Z &= \begin{pmatrix}\,
        c & d & 0 & 0 & \bar{a} \\
        0 & 0 & c & d & \bar{b} \,
    \end{pmatrix} \, .
\end{align}
\end{subequations}

\begin{thm}[Abelian syzygies and low distance]
\label{thm:abelian-canonical-syzygies}
    Assume $a,b,c,d\in\Ring$ are not all zero. Then at least one of the following four abelian syzygies gives a nontrivial logical operator of $\mathcal{Q}(a,b,c,d)$:
    \begin{align}
        \mathbf{x}_1&:=(\bar{d}\;\,\bar{c}\;\,0\;\,0\;\,0) \; , \qquad
        \mathbf{x}_2:=(0\;\,0\;\,\bar{d}\;\,\bar{c}\;\,0) \, , \\
        \mathbf{z}_1&:=(\bar{b}\;\,0\;\,\bar{a}\;\,0\;\,0) \; , \qquad
        \mathbf{z}_2:=(0\;\,\bar{b}\;\,0\;\,\bar{a}\;\,0) \, .
    \end{align}
    Consequently,
    \begin{align}
        d\big(\mathcal{Q}(a,b,c,d)\big) \leq
        \max\{\wt(a)+\wt(b),\,\wt(c)+\wt(d)\} \, .
    \end{align}
    In particular, if $\wt(a),\wt(b),\wt(c),\wt(d)\leq w$, then $d(\mathcal{Q}(a,b,c,d))\leq 2w=O(1)$.
\end{thm}

\begin{proof}
    The four displayed rows are cycles by commutativity of $\Ring$, since
    \begin{align}
        \mathbf{x}_1H_Z^\dagger
        = (\bar{d}\,\bar{c}+\bar{c}\,\bar{d}\;\;0)=0 \, ,
    \end{align}
    and the other three identities are analogous. We prove that they cannot all be stabilizers unless $a=b=c=d=0$, i.e. the code is trivial.

    Suppose first that $\mathbf{x}_1, \mathbf{x}_2 \in \rs H_X$. From $\mathbf{x}_1\in\rs H_X$, applied to the first four block-columns, there exist $\alpha,\beta\in\Ring$ such that
    \begin{align}
        \alpha a=\bar{d}\, ,\qquad
        \beta a=\bar{c}\, ,\qquad
        \alpha b=0\, ,\qquad
        \beta b=0 \, .
    \end{align}
    The last two equations imply that $\alpha,\beta \in \operatorname{Ann}(b)$ and, combined with the first two, then imply that $\bar{c},\bar{d}\in a\operatorname{Ann}(b)$. Similarly, $\mathbf{x}_2\in\rs(H_X)$ implies $\bar{c},\bar{d}\in b\operatorname{Ann}(a)$. Applying the antipode gives
    \begin{align}
        c,d\in K(\bar{a},\bar{b}) \, .
    \end{align}
    Likewise, if $\mathbf{z}_1, \mathbf{z}_2\in\rs(H_Z)$, then the same block comparison gives
    \begin{align}
        \bar{a},\bar{b}\in K(c,d) \, .
    \end{align}

    By Lemma~\ref{lem:mutual-annihilator-syzygies}, applied with
    \begin{align}
        p=\bar{a}\, ,\qquad q=\bar{b}\, ,\qquad r=c\, ,\qquad s=d \, ,
    \end{align}
    simultaneous containment
    \begin{align}
        \mathbf{x}_1, \mathbf{x}_2\in\rs(H_X)\; ,\qquad
        \mathbf{z}_1, \mathbf{z}_2\in\rs(H_Z)
    \end{align}
    would force $\bar{a}=\bar{b}=c=d=0$ and hence $a=b=c=d=0$, contrary to assumption. Hence at least one nonzero syzygy is not a stabilizer and therefore gives a nontrivial logical operator.

    The $X$-type syzygies have weight $\wt(c)+\wt(d)$, while the $Z$-type syzygies have weight $\wt(a)+\wt(b)$. The displayed upper bound on $d(\mathcal{Q}(a,b,c,d))$ follows immediately.
\end{proof}

\subsection{Tanner graph cycles and girth bounds}
\label{app:tanner graph cycles}

\begin{thm}[Tanner-girth bound for 5-block codes]
\label{thm:5-block-tanner-girth}
    Let $G$ be a finite group and let $a,b,c,d\in\F_2[G]$ be trinomials with distinct monomial terms. Let
    \begin{align}
        A=L[a]\, ,\qquad B=L[b]\, ,\qquad C=R[c]\, ,\qquad D=R[d]\, ,
    \end{align}
    and define $H_X,H_Z$ according to \eqref{eq:5-block H_X,H_Z}. Then the Tanner graphs of $H_X$ and $H_Z$ obey
    \begin{align}
        \operatorname{girth}(H_X)\leq 8\; ,\qquad
        \operatorname{girth}(H_Z)\leq 8 \, .
    \end{align}
\end{thm}

\begin{proof}
    We first prove the claim for $H_X$. Write
    \begin{align}
        a=a_1+a_2+a_3\, ,\qquad
        c=c_1+c_2+c_3\, ,\qquad
        d=d_1+d_2+d_3
    \end{align}
    as sums of distinct group elements, and let
    \begin{align}
        A_i=L[a_i]\, ,\qquad C_i=R[c_i]\, ,\qquad D_i=R[d_i]
    \end{align}
    denote the corresponding permutation matrices. Define
    \begin{align}
        U := A_1A_2^\transp\, ,\qquad
        V := D_1^\transp C_1 \, .
    \end{align}
    The matrix $U$ is left-regular, while $V$ is right-regular. Since left and right regular representations commute, we have $UV=VU$, and hence
    \begin{align}\label{eq:HX-girth-8-identity}
        V^\transp U^\transp VU = 1 \, .
    \end{align}

    We now interpret \eqref{eq:HX-girth-8-identity} as a closed walk in the Tanner graph of $H_X$. Let $C^X_1,C^X_2$ denote the two check-row blocks of $H_X$, and let $Q_1,\ldots,Q_5$ denote the five data block-columns. At the level of check-row blocks, the four check-to-check moves are
    \begin{align}\label{eq:H_X 8-cycle}
        C^X_1
        \xrightarrow[\;Q_1\;]{A_1A_2^\transp}
        C^X_1
        \xrightarrow[\;Q_5\;]{D_1^\transp C_1}
        C^X_2
        \xrightarrow[\;Q_2\;]{A_2A_1^\transp}
        C^X_2
        \xrightarrow[\;Q_5\;]{C_1^\transp D_1}
        C^X_1 \, .
    \end{align}
    Indeed, the product of the four transition matrices is
    \begin{align}
        (C_1^\transp D_1)(A_2A_1^\transp)(D_1^\transp C_1)(A_1A_2^\transp)
        = V^\transp U^\transp VU
        = 1 \, .
    \end{align}
    Thus, starting from any check vertex in $C^X_1$, these four check-to-check moves return to the starting check vertex. Each check-to-check move corresponds to two edges in the $X$-Tanner graph, so \eqref{eq:H_X 8-cycle} gives a closed walk of length $8$ on the $X$-Tanner graph. We now verify that the eight visited vertices are pairwise distinct, so that this closed walk is a simple cycle; this is where the distinctness of the monomial terms enters. Denote the visited check vertices by $u_0,u_1\in C^X_1$ and $u_2,u_3\in C^X_2$, where $u_1=A_1A_2^\transp u_0$ and $u_3=A_2A_1^\transp u_2$. Since $a_1\neq a_2$, the permutation $A_1A_2^\transp=L[a_1a_2^{-1}]$ is left multiplication by a non-identity group element and is therefore fixed-point free, so $u_1\neq u_0$ and $u_3\neq u_2$; all other check pairs are distinct because $C^X_1$ and $C^X_2$ are different blocks. Among the data vertices, those visited in $Q_1$ and $Q_2$ are distinct from all others because they lie in different data blocks, while the two data vertices visited in $Q_5$ are $w_1=C_1u_1$ (between $u_1$ and $u_2$) and $w_2=C_1u_0$ (between $u_3$ and $u_0$), which are distinct because $C_1$ is a permutation and $u_1\neq u_0$. Note that $c_1\neq d_1$ is not required: even when the transition monomial $D_1^\transp C_1$ is trivial, the two edges of each $Q_5$ move lie in the distinct blocks $C^\transp$ and $D^\transp$. Hence \eqref{eq:H_X 8-cycle} traces a simple Tanner $8$-cycle, and we conclude that $\operatorname{girth}(H_X) \leq 8$.

    The proof for $H_Z$ is the same with left and right interchanged. Write
    \begin{align}
        c=c_1+c_2+c_3\, ,\qquad
        a=a_1+a_2+a_3\, ,\qquad
        b=b_1+b_2+b_3
    \end{align}
    and set
    \begin{align}
        C_i=R[c_i]\, ,\qquad A_i=L[a_i]\, ,\qquad B_i=L[b_i] \, .
    \end{align}
    Define
    \begin{align}
        U:=C_1C_2^\transp\, ,\qquad
        V:=B_1^\transp A_1 \, .
    \end{align}
    Now $U$ is right-regular and $V$ is left-regular, so $UV=VU$. Hence
    \begin{align}\label{eq:HZ-girth-8-identity}
        V^\transp U^\transp VU = 1 \, .
    \end{align}
    Let $C^Z_1,C^Z_2$ denote the two check-row blocks of $H_Z$. The corresponding closed check-walk is
    \begin{align}
        C^Z_1
        \xrightarrow[\;Q_1\;]{C_1C_2^\transp}
        C^Z_1
        \xrightarrow[\;Q_5\;]{B_1^\transp A_1}
        C^Z_2
        \xrightarrow[\;Q_3\;]{C_2C_1^\transp}
        C^Z_2
        \xrightarrow[\;Q_5\;]{A_1^\transp B_1}
        C^Z_1 \, .
    \end{align}
    The product of the four transition matrices is
    \begin{align}
        (A_1^\transp B_1)(C_2C_1^\transp)(B_1^\transp A_1)(C_1C_2^\transp)
        = V^\transp U^\transp VU
        = 1 \, .
    \end{align}
    Hence the four check-to-check moves return to the starting check vertex, giving a closed walk of length $8$ on the $Z$-Tanner graph. Vertex distinctness follows exactly as for $H_X$, with left and right interchanged: denoting the visited check vertices by $u_0,u_1\in C^Z_1$ and $u_2,u_3\in C^Z_2$, the permutation $C_1C_2^\transp$ is right multiplication by the non-identity element $c_2^{-1}c_1$ (recall $c_1\neq c_2$) and is therefore fixed-point free, so $u_1\neq u_0$ and $u_3\neq u_2$; the data vertices visited in $Q_1$ and $Q_3$ lie in different data blocks from all others, and the two data vertices visited in $Q_5$, namely $w_1=A_1u_1$ and $w_2=A_1u_0$, are distinct because $u_1\neq u_0$. Hence the walk traces a simple Tanner $8$-cycle, and we conclude that $\operatorname{girth}(H_Z)\leq 8$.
\end{proof}

The next theorem asserts that if the same trinomials are used for both $H_{\rm left}$ and $H_{\rm right}$, i.e. a symmetric lifted/balanced product, then there will be unavoidable 4-cycles in the balanced product code's Tanner graph, and so its girth will be 4.

\begin{thm}[Tanner-girth bound for symmetric 5-block codes]
\label{thm:symmetric-girth-4}
    Let $G$ be a finite group and let $a,b\in\F_2[G]$ be trinomials with distinct monomial terms. Suppose the two classical seed codes use the same trinomials, i.e.
    \begin{align}
        A=L[a]\, ,\qquad B=L[b]\, ,\qquad C=R[a]\, ,\qquad D=R[b]\, .
    \end{align}
    Then the Tanner graphs of $H_X$ and $H_Z$ in \eqref{eq:5-block H_X,H_Z} contain $4$-cycles. In particular,
    \begin{align}
        \operatorname{girth}(H_X)=4\, ,\qquad
        \operatorname{girth}(H_Z)=4 \, .
    \end{align}
\end{thm}

\begin{proof}
    Let $1$ denote the identity element of $G$. We use the convention that $L[g]$ connects a data vertex labeled $h$ to the check vertex labeled $gh$, while $R[g]$ connects a data vertex labelled $h$ to the check vertex labeled $hg$. Thus $L[g]^\transp$ connects $h$ to $g^{-1}h$, and $R[g]^\transp$ connects $h$ to $hg^{-1}$.

    Write
    \begin{align}
        a=a_1+a_2+a_3
    \end{align}
    as a sum of distinct group elements, and choose two distinct monomial terms $a_i,a_j$. We first construct a $4$-cycle in the Tanner graph of $H_X$. Let $C^X_1,C^X_2$ denote the two check-row blocks of $H_X$, and let $Q_1,\ldots,Q_5$ denote the five data block-columns. Since each data or check subblock internally contains $|G|$ vertices, we will use the tuple $(T,g)$ to indicate the vertex labeled by $g\in G$ in subblock $T$. Now, consider the four Tanner vertices
    \begin{align}
        (C^X_1,1)\, ,\qquad
        (Q_1,a_j^{-1})\, ,\qquad
        (C^X_1,a^{}_i a_j^{-1})\, ,\qquad
        (Q_5,a_i) \, .
    \end{align}
    We claim that these vertices form a cycle. Indeed, the block $A=L[a]$ between $C^X_1$ and $Q_1$ contains the monomial matrices $L[a_i]$ and $L[a_j]$, so the data vertex $(Q_1,a_j^{-1})$ is adjacent to the two check vertices
    \begin{align}
        (C^X_1,a^{}_j a_j^{-1})=(C^X_1,1)\, ,\qquad
        (C^X_1,a^{}_i a_j^{-1}) \, .
    \end{align}
    On the other hand, the block $C^\transp=R[a]^\transp$ between $C^X_1$ and $Q_5$ contains $R[a_i]^\transp$ and $R[a_j]^\transp$, so the data vertex $(Q_5,a_i)$ is adjacent to
    \begin{align}
        (C^X_1,a^{}_i a_i^{-1})=(C^X_1,1)\, ,\qquad
        (C^X_1,a^{}_i a_j^{-1}) \, .
    \end{align}
    Hence we have the $4$-cycle
    \begin{align}
        (C^X_1,1)
        \;-\;
        (Q_1,a_j^{-1})
        \;-\;
        (C^X_1,a^{}_i a_j^{-1})
        \;-\;
        (Q_5,a_i)
        \;-\;
        (C^X_1,1) \, .
    \end{align}
    The two check vertices are distinct because $a_i\neq a_j$, and the two data vertices are distinct because they lie in different data blocks $Q_1$ and $Q_5$. Therefore this is a simple Tanner $4$-cycle, and so $\operatorname{girth}(H_X)=4$.

    The proof for $H_Z$ is analogous, with left and right interchanged. Let $C^Z_1,C^Z_2$ denote the two check-row blocks of $H_Z$. In the first check-row block $C^Z_1$, the data block $Q_1$ is connected by $C=R[a]$, while the data block $Q_5$ is connected by $A^\transp=L[a]^\transp$. Consider the four Tanner vertices
    \begin{align}
        (C^Z_1,1)\, ,\qquad
        (Q_1,a_j^{-1})\, ,\qquad
        (C^Z_1,a_j^{-1}a^{}_i)\, ,\qquad
        (Q_5,a_i) \, .
    \end{align}
    Since $C=R[a]$, the data vertex $(Q_1,a_j^{-1})$ is adjacent to
    \begin{align}
        (C^Z_1,a_j^{-1}a^{}_j)=(C^Z_1,1)\, ,\qquad
        (C^Z_1,a_j^{-1}a^{}_i) \, .
    \end{align}
    Since $A^\transp=L[a]^\transp$, the data vertex $(Q_5,a_i)$ is adjacent to
    \begin{align}
        (C^Z_1,a_i^{-1}a^{}_i)=(C^Z_1,1)\, ,\qquad
        (C^Z_1,a_j^{-1}a^{}_i) \, .
    \end{align}
    Hence we have the $4$-cycle
    \begin{align}
        (C^Z_1,1)
        \;-\;
        (Q_1,a_j^{-1})
        \;-\;
        (C^Z_1,a_j^{-1}a^{}_i)
        \;-\;
        (Q_5,a_i)
        \;-\;
        (C^Z_1,1) \, .
    \end{align}
    Again the two check vertices are distinct because $a_i\neq a_j$, and the two data vertices lie in different data blocks. Therefore this is a simple Tanner $4$-cycle, and so $\operatorname{girth}(H_Z)=4$.

    The same construction can also be made with the trinomial $b$, using the pairs of data blocks $(Q_4,Q_5)$ for $H_X$ and $(Q_4,Q_5)$ for $H_Z$ in the second check-row blocks.
\end{proof}

So if we wish to obtain ZSZ-LP codes with girth\;$\geq6$, then we cannot recycle the same trinomial pair for both $H_{\rm left}$ and $H_{\rm right}$. Table \ref{tab:tanner graph cycles} reports the Tanner girths and the number of short cycles for the candidate ZSZ-LP codes. For each quantum code, the reported girth is the minimum of the Tanner-graph girths of $H_X$ and $H_Z$. The $4$-cycle and $6$-cycle columns count exact simple cycles in the Tanner graphs and sum the contributions from $H_X$ and $H_Z$.

\begin{table}[t]
\centering
\begin{tabular}{cccc}
\toprule
\textbf{Code name} & \textbf{Girth} & \textbf{4-cycles} & \textbf{6-cycles} \\ \midrule
ZSZ-LP-160 & 4 & 64 & 4608 \\
ZSZ-LP-240 & 6 & 0 & 3936 \\
ZSZ-LP-320 & 4 & 32 & 4000 \\
ZSZ-LP-390 & 4 & 84 & 4056 \\
ZSZ-LP-550 & 6 & 0 & 2864 \\
ZSZ-LP-775 & 4 & 160 & 3520 \\ \bottomrule
\end{tabular}
\caption{Tanner graph girths and short-cycle counts for the candidate ZSZ-LP codes. The girth is $\min\{\operatorname{girth}(H_X),\operatorname{girth}(H_Z)\}$, while the 4-cycle and 6-cycle counts are the sums of those for $H_X$ and $H_Z$.}
\label{tab:tanner graph cycles}
\end{table}

\subsection{Residual center automorphisms}
\label{app:residual center aut}

Although the balanced product quotients out the full diagonal ZSZ action, a
small residual quantum code automorphism can remain whenever the underlying ZSZ group has a nontrivial center. Let $z\in Z(\ZSZ_{\ell_1,\ell_2,q})$ and let $P_z=L[z]=R[z]$ be its regular-representation permutation matrix, where the last equality follows from centrality.  Since $z$ commutes with every group element, $P_z$ commutes with all lifted blocks $A=L[a]$, $B=L[b]$, $C=R[c]$, and $D=R[d]$, as well as their transposes.  Therefore the simultaneous permutation $P_z^{\oplus 5}$ of the five data-qubit blocks satisfies
\begin{align}
    H_X P_z^{\oplus 5} = P_z^{\oplus 2} H_X \, ,\qquad
    H_Z P_z^{\oplus 5} = P_z^{\oplus 2} H_Z \, .
\end{align}
Thus every element of the group center gives a residual code (Tanner-graph) automorphism of the ZSZ-LP code.  These residual symmetries are generally much smaller than the original ZSZ group itself, unlike the full symmetry available in abelian LP codes; for instance, the largest center among the candidate codes below has size $12$, and ZSZ-LP-775 has a trivial center.

For $G=\ZSZ_{\ell_1,\ell_2,q}$, write
$t=\operatorname{ord}_{\ell_1}(q)$ and $u=\gcd(\ell_1,q-1)$.  The center is
\begin{align}
    Z(G) = \left\{x^{i\ell_1/u} y^{jt} \,:\, 0\leq i<u,\;0\leq j<\ell_2/t\right\} \, ,
\end{align}
and hence $|Z(G)|=u\ell_2/t$.  Table \ref{tab:ZSZ-LP centers} lists the residual center automorphisms for the candidate ZSZ-LP codes.

\begin{table}[t]
\centering\renewcommand{\arraystretch}{1.1}
\begin{tabular}{ccccc}
\toprule
\textbf{Code name} & $(\ell_1,\ell_2,q)$ & $\operatorname{ord}_{\ell_1}(q)$ & $Z(G)$ & $|Z(G)|$ \\ \midrule
ZSZ-LP-160 & $(16,2,9)$ & $2$ & $\langle x^2\rangle=\{x^{2r}:0\leq r<8\}$ & $8$ \\
ZSZ-LP-240 & $(12,4,7)$ & $2$ & $\langle x^2,y^2\rangle=\{x^{2r}y^{2s}:0\leq r<6,\;0\leq s<2\}$ & $12$ \\
ZSZ-LP-320 & $(16,4,3)$ & $4$ & $\langle x^8\rangle=\{1,x^8\}$ & $2$ \\
ZSZ-LP-390 & $(26,3,3)$ & $3$ & $\langle x^{13}\rangle=\{1,x^{13}\}$ & $2$ \\
ZSZ-LP-550 & $(22,5,3)$ & $5$ & $\langle x^{11}\rangle=\{1,x^{11}\}$ & $2$ \\
ZSZ-LP-775 & $(31,5,2)$ & $5$ & $\{1\}$ & $1$ \\ \bottomrule
\end{tabular}
\caption{Centers of the ZSZ groups used by the candidate ZSZ-LP codes. Each central element gives a residual code automorphism by applying the corresponding regular-representation permutation simultaneously to all five data-qubit blocks.}
\label{tab:ZSZ-LP centers}
\end{table}

\subsection{Distance bounds on fold-symmetric codes}
\label{app:fold-symmetric-low-distance}

We now isolate a low-distance obstruction specific to the automorphism-fold construction of Section \ref{sec:ZX dualities}. Recall that an automorphism fold chooses $c=\phi(\bar{a})$ and $d=\phi(\bar{b})$ according to \eqref{eq:automorphism fold c d}. The simplest choice is the trivial group automorphism $\phi=1$, for which its inverse-automorphism $\theta$ is ordinary group inversion and
\begin{align}\label{eq:trivial aut fold c,d}
    c=\bar{a}\, ,\qquad d=\bar{b} \, .
\end{align}
Although this choice gives a simple $ZX$-duality for every finite group and every pair $a,b\in\Ring$, the resulting quantum code will have poor distance. In fact, it produces a canonical weight-three logical pair independently of whether the group is abelian, of the weights of $a,b$, and of their invertibility in the group algebra.

\begin{thm}[Low distance for the trivial automorphism fold]
\label{thm:trivial-automorphism-fold-low-distance}
    Let $G$ be any finite group of order $\ell$, let $\Ring=\F_2[G]$, and let $a,b\in\Ring$. Set
    \begin{align}
        c=\bar{a}\, ,\qquad d=\bar{b}
    \end{align}
    and define the CSS code $\mathcal{Q}(a,b,\bar{a},\bar{b})$ by \eqref{eq:5-block H_X,H_Z}. Then the weight-3 row vector
    \begin{align}\label{eq:trivial-fold-weight-three-logical}
        \boldsymbol{\lambda}:=(1\;\;0\;\;0\;\;1\;\;1)
    \end{align}
    gives both a nontrivial logical $\bar{X}$ operator and a nontrivial logical $\bar{Z}$ operator. Consequently,
    \begin{align}
        d_X\big(\mathcal{Q}(a,b,\bar{a},\bar{b})\big)\leq3\, ,\qquad
        d_Z\big(\mathcal{Q}(a,b,\bar{a},\bar{b})\big)\leq3 \, .
    \end{align}
\end{thm}

\begin{proof}
    It is first straightforward to see that $\boldsymbol{\lambda}$ lies in $\ker H_Z$ and $\ker H_X$, since
    \begin{subequations}\label{eq:lambda H^T_Z,H^T_X}
    \begin{align}
        \boldsymbol{\lambda}H_Z^\transp
        &=(c+\bar{a}\;\;\;d+\bar{b}) = (2c\;\;2d) = 0\, , \\
        \boldsymbol{\lambda}H_X^\transp
        &=(a+\bar{c}\;\;\;b+\bar{d}) = (2a\;\;2b) = 0 \, .
    \end{align}
    \end{subequations}
    It remains to show that it is not a stabilizer.

    Suppose first that $\boldsymbol{\lambda}\in\rs(H_X)$. Comparing the first, second, and fourth block-columns gives $\alpha,\beta\in\Ring$ such that
    \begin{align}\label{eq:trivial-fold-X-stabilizer-contradiction}
        \bar{a}\alpha=1\, ,\qquad
        \bar{a}\beta=0\, ,\qquad
        \bar{b}\beta=1 \, .
    \end{align}
    The first equation says that $\bar{a}$ admits a right-inverse. Since $\Ring$ is finite, a one-sided inverse is automatically two-sided, so $\bar{a}$ is a unit. The second equation in \eqref{eq:trivial-fold-X-stabilizer-contradiction} therefore forces $\beta=0$, contradicting the third equation. Hence $\boldsymbol{\lambda}\notin\rs(H_X)$.

    Similarly, if $\boldsymbol{\lambda}\in\rs(H_Z)$, then comparison of the first, third, and fourth block-columns gives $\alpha,\beta\in\Ring$ satisfying
    \begin{align}\label{eq:trivial-fold-Z-stabilizer-contradiction}
        \alpha\bar{c}=1\, ,\qquad
        \beta\bar{c}=0\, ,\qquad
        \beta\bar{d}=1 \, .
    \end{align}
    The first equation makes $\bar{c}$ a unit, so the second forces $\beta=0$, again contradicting the third. Therefore $\boldsymbol{\lambda}\notin\rs(H_Z)$, and we conclude that it gives both a nontrivial logical $\bar{X}$ operator and a nontrivial logical $\bar{Z}$ operator.
\end{proof}

Theorem \ref{thm:trivial-automorphism-fold-low-distance} is related but distinct from the abelian obstruction of Theorem \ref{thm:abelian-canonical-syzygies}. The latter gives a constant upper bound essentially because an abelian group algebra provides the classical and quantum codes with ``too much symmetry'' which leads to low-weight codewords/logical operators, regardless of the choice of trinomials. The present obstruction says that the trivial automorphism \eqref{eq:trivial aut fold c,d} provides ``too much symmetry'' and leads to low-weight logical operators, and the no-go result persists for arbitrary finite groups, including non-abelian groups. It also is independent of the weight of $a$ and $b$: changing their weights cannot remove the low-weight logical \eqref{eq:trivial-fold-weight-three-logical}.

A nontrivial automorphism $\phi$ circumvents this particular obstruction by replacing $c=\bar{a}$ and $d=\bar{b}$ with $c=\phi(\bar{a})$ and $d=\phi(\bar{b})$. \eqref{eq:lambda H^T_Z,H^T_X} then no longer vanishes unless the relevant seed polynomials are invariant under $\phi$. This corroborates the observation that the nontrivial automorphism-folds in Table \ref{tab:ZSZ-LP ZX dualities} can retain the desired $ZX$-duality while avoiding the distance-3 obstruction above. Of course, a nontrivial $\phi$ does not by itself guarantee large distance; it only removes this particular low-distance mechanism.

%%%%%%%%%%%%%%%%%%%%%%%%%%%%%%%%%%%%%%%%%%%%%%%%%%%%%%%%

\section{Greedy AOD scheduler}
\label{app:greedy-aod-scheduler}

We give a compact version of the greedy scheduler used for the AOD-routing calculations in Algorithm \ref{alg:greedy-aod-scheduler}. The input jobs are monomial interactions of the form $(P,C,D,g)$, where $P\in\{L,R\}$ is the left/right module, $C$ is a check subblock chosen from $C^X_1,C^X_2,C^Z_1,C^Z_2$, $Q$ is a data subblock chosen from $Q_1,\ldots,Q_5$, and $g\in\ZSZ$ is the monomial.

\begin{algorithm}[t]
\caption{Greedy AOD scheduler}
\label{alg:greedy-aod-scheduler}
\Input{Job set $\mathcal{J}$ with jobs $(P,C,Q,g)$, group parameters $(\ell_1,\ell_2,q)$}
\Return{Merged move/gate schedule $\mathsf{Sched}$}
\ForEach{syndrome subblock $Q$ appearing in $\mathcal{J}$}{
    $F_Q \leftarrow (1,1)$ \tcp*{left and right reference-frame components}
    $\mathcal{R}_Q \leftarrow [\;]$ \tcp*{route for this syndrome subblock}
    $\mathcal{J}_Q \leftarrow \{J\in\mathcal{J}: J.\mathrm{check}=Q\}$\;
    \While{$\mathcal{J}_Q\neq\emptyset$}{
        \ForEach{$J=(P,C,Q,g)\in\mathcal{J}_Q$}{
            $\Delta_J \leftarrow \mathrm{TransitionMonomial}(F_Q,P,g)$\;
            $c_J \leftarrow \mathrm{AODCost}(\Delta_J)$\;
        }
        Choose $J^\star=(P^\star,C,Q^\star,g^\star)$ minimizing $c_J$\;
        Break ties deterministically\;
        Append the moves for $\Delta_{J^\star}$ to $\mathcal{R}_Q$\;
        Append the monomial gate layer $(P^\star,C,Q^\star,g^\star)$ to $\mathcal{R}_Q$\;
        Update $F_Q$ to the frame required by $(P^\star,g^\star)$\;
        $\mathcal{J}_Q \leftarrow \mathcal{J}_Q\setminus\{J^\star\}$\;
    }
}
$\mathsf{Sched}\leftarrow [\;]$\;
\While{some route $\mathcal{R}_Q$ is nonempty}{
    Greedily merge the next compatible move layers from the routes\;
    Greedily merge the next compatible gate layers from the routes\;
}
\Return{$\mathsf{Sched}$}
\end{algorithm}

For a full $X+Z$ round, the $C^X_i$ routes are scheduled before the $C^Z_i$ routes. We do not insert a physical reset of the final $C^X_i$ frames prior to measurement; this final frame can instead be tracked in software after measurement.

For completeness, Table \ref{tab:greedy-aod-gate-schedule} lists the resulting gate schedules for all simulated ZSZ-LP codes (all but ZSZ-LP-775) in a format that is sufficient to reconstruct the monomial ordering used in the circuit-level simulations. For each code, we write the trinomials in Table \ref{tab:ZSZ-LP polynomials} as $a=a_1+a_2+a_3$, and likewise for $b,c,d$, in the displayed order. The $X$-check table uses the monomials themselves, while the $Z$-check table uses inverse monomials because $H_Z$ contains the antipode-transposed blocks. In the circuit, $X$-check gates have CNOT direction $C^X_i\to Q_j$, while $Z$-check gates have CNOT direction $Q_j\to C^Z_i$.

% $X$-check gate layers
\begin{table}[t]
\centering
\resizebox{\textwidth}{!}{%
\begin{tabular}{cccccc}
\toprule
\textbf{Layer} & \textbf{ZSZ-LP-160} & \textbf{ZSZ-LP-240} & \textbf{ZSZ-LP-320} & \textbf{ZSZ-LP-390} & \textbf{ZSZ-LP-550} \\
\midrule
1 & $(C^X_1,Q_1,a_1),(C^X_2,Q_2,a_1)$ & $(C^X_1,Q_1,a_1),(C^X_2,Q_2,a_1)$ & $(C^X_1,Q_1,a_1),(C^X_2,Q_2,a_1)$ & $(C^X_1,Q_1,a_1),(C^X_2,Q_2,a_1)$ & $(C^X_1,Q_1,a_1),(C^X_2,Q_2,a_1)$ \\
2 & $(C^X_1,Q_3,b_1),(C^X_2,Q_4,b_1)$ & $(C^X_1,Q_3,b_1),(C^X_2,Q_4,b_1)$ & $(C^X_1,Q_3,b_1),(C^X_2,Q_4,b_1)$ & $(C^X_1,Q_3,b_1),(C^X_2,Q_4,b_1)$ & $(C^X_1,Q_3,b_1),(C^X_2,Q_4,b_1)$ \\
3 & $(C^X_1,Q_5,c_1)$ & $(C^X_1,Q_3,b_2),(C^X_2,Q_4,b_2)$ & $(C^X_1,Q_5,c_1)$ & $(C^X_1,Q_5,c_1)$ & $(C^X_1,Q_5,c_1)$ \\
4 & $(C^X_1,Q_1,a_2),(C^X_2,Q_5,d_1)$ & $(C^X_1,Q_3,b_3),(C^X_2,Q_4,b_3)$ & $(C^X_1,Q_1,a_3),(C^X_2,Q_5,d_1)$ & $(C^X_1,Q_3,b_3),(C^X_2,Q_5,d_1)$ & $(C^X_1,Q_1,a_3),(C^X_2,Q_5,d_1)$ \\
5 & $(C^X_2,Q_2,a_2)$ & $(C^X_1,Q_1,a_3),(C^X_2,Q_2,a_3)$ & $(C^X_1,Q_3,b_2),(C^X_2,Q_2,a_3)$ & $(C^X_1,Q_1,a_2),(C^X_2,Q_4,b_3)$ & $(C^X_1,Q_3,b_3),(C^X_2,Q_2,a_3)$ \\
6 & $(C^X_1,Q_5,c_2),(C^X_2,Q_2,a_3)$ & $(C^X_1,Q_1,a_2),(C^X_2,Q_2,a_2)$ & $(C^X_1,Q_1,a_2),(C^X_2,Q_4,b_2)$ & $(C^X_1,Q_1,a_3),(C^X_2,Q_2,a_2)$ & $(C^X_1,Q_1,a_2),(C^X_2,Q_4,b_3)$ \\
7 & $(C^X_1,Q_5,c_3),(C^X_2,Q_4,b_2)$ & $(C^X_1,Q_5,c_1)$ & $(C^X_1,Q_3,b_3),(C^X_2,Q_2,a_2)$ & $(C^X_1,Q_3,b_2),(C^X_2,Q_2,a_3)$ & $(C^X_1,Q_3,b_2),(C^X_2,Q_2,a_2)$ \\
8 & $(C^X_2,Q_4,b_3)$ & $(C^X_1,Q_5,c_2)$ & $(C^X_2,Q_4,b_3)$ & $(C^X_2,Q_4,b_2)$ & $(C^X_2,Q_4,b_2)$ \\
9 & $(C^X_1,Q_1,a_3)$ & $(C^X_1,Q_5,c_3)$ & $(C^X_1,Q_5,c_2)$ & $(C^X_1,Q_5,c_2)$ & $(C^X_1,Q_5,c_3)$ \\
10 & $(C^X_1,Q_3,b_2),(C^X_2,Q_5,d_2)$ & $(C^X_2,Q_5,d_1)$ & $(C^X_1,Q_5,c_3)$ & $(C^X_1,Q_5,c_3)$ & $(C^X_1,Q_5,c_2)$ \\
11 & $(C^X_1,Q_3,b_3),(C^X_2,Q_5,d_3)$ & $(C^X_2,Q_5,d_3)$ & $(C^X_2,Q_5,d_2)$ & $(C^X_2,Q_5,d_2)$ & $(C^X_2,Q_5,d_3)$ \\
12 & -- & $(C^X_2,Q_5,d_2)$ & $(C^X_2,Q_5,d_3)$ & $(C^X_2,Q_5,d_3)$ & $(C^X_2,Q_5,d_2)$ \\
\bottomrule
\end{tabular}%
} \\ \vspace{1em}
\resizebox{\textwidth}{!}{%
\begin{tabular}{cccccc}
\toprule
\textbf{Layer} & \textbf{ZSZ-LP-160} & \textbf{ZSZ-LP-240} & \textbf{ZSZ-LP-320} & \textbf{ZSZ-LP-390} & \textbf{ZSZ-LP-550} \\
\midrule
1 & $(C^Z_1,Q_1,c_1^{-1}),(C^Z_2,Q_3,c_1^{-1})$ & $(C^Z_1,Q_1,c_1^{-1}),(C^Z_2,Q_3,c_1^{-1})$ & $(C^Z_1,Q_1,c_1^{-1}),(C^Z_2,Q_3,c_1^{-1})$ & $(C^Z_1,Q_1,c_1^{-1}),(C^Z_2,Q_3,c_1^{-1})$ & $(C^Z_1,Q_1,c_1^{-1}),(C^Z_2,Q_3,c_1^{-1})$ \\
2 & $(C^Z_1,Q_2,d_1^{-1}),(C^Z_2,Q_4,d_1^{-1})$ & $(C^Z_1,Q_2,d_1^{-1}),(C^Z_2,Q_4,d_1^{-1})$ & $(C^Z_1,Q_2,d_1^{-1}),(C^Z_2,Q_4,d_1^{-1})$ & $(C^Z_1,Q_2,d_1^{-1}),(C^Z_2,Q_4,d_1^{-1})$ & $(C^Z_1,Q_2,d_1^{-1}),(C^Z_2,Q_4,d_1^{-1})$ \\
3 & $(C^Z_1,Q_5,a_1^{-1})$ & $(C^Z_1,Q_5,a_1^{-1})$ & $(C^Z_1,Q_5,a_1^{-1})$ & $(C^Z_1,Q_5,a_1^{-1})$ & $(C^Z_1,Q_5,a_1^{-1})$ \\
4 & $(C^Z_1,Q_1,c_2^{-1}),(C^Z_2,Q_5,b_1^{-1})$ & $(C^Z_1,Q_2,d_3^{-1}),(C^Z_2,Q_5,b_1^{-1})$ & $(C^Z_1,Q_5,a_3^{-1})$ & $(C^Z_1,Q_2,d_2^{-1}),(C^Z_2,Q_5,b_1^{-1})$ & $(C^Z_1,Q_5,a_3^{-1})$ \\
5 & $(C^Z_1,Q_1,c_3^{-1}),(C^Z_2,Q_3,c_2^{-1})$ & $(C^Z_1,Q_1,c_2^{-1}),(C^Z_2,Q_4,d_3^{-1})$ & $(C^Z_2,Q_5,b_1^{-1})$ & $(C^Z_1,Q_1,c_3^{-1}),(C^Z_2,Q_4,d_2^{-1})$ & $(C^Z_1,Q_5,a_2^{-1})$ \\
6 & $(C^Z_1,Q_2,d_3^{-1}),(C^Z_2,Q_3,c_3^{-1})$ & $(C^Z_1,Q_1,c_3^{-1}),(C^Z_2,Q_3,c_2^{-1})$ & $(C^Z_1,Q_2,d_2^{-1}),(C^Z_2,Q_5,b_2^{-1})$ & $(C^Z_1,Q_2,d_3^{-1})$ & $(C^Z_2,Q_5,b_1^{-1})$ \\
7 & $(C^Z_1,Q_2,d_2^{-1}),(C^Z_2,Q_4,d_3^{-1})$ & $(C^Z_1,Q_2,d_2^{-1}),(C^Z_2,Q_3,c_3^{-1})$ & $(C^Z_1,Q_1,c_3^{-1})$ & $(C^Z_1,Q_1,c_2^{-1}),(C^Z_2,Q_5,b_3^{-1})$ & $(C^Z_1,Q_1,c_3^{-1}),(C^Z_2,Q_5,b_3^{-1})$ \\
8 & $(C^Z_2,Q_4,d_2^{-1})$ & $(C^Z_2,Q_4,d_2^{-1})$ & $(C^Z_1,Q_1,c_2^{-1}),(C^Z_2,Q_4,d_2^{-1})$ & $(C^Z_2,Q_5,b_2^{-1})$ & $(C^Z_1,Q_2,d_3^{-1}),(C^Z_2,Q_5,b_2^{-1})$ \\
9 & $(C^Z_1,Q_5,a_2^{-1})$ & $(C^Z_1,Q_5,a_3^{-1})$ & $(C^Z_1,Q_2,d_3^{-1}),(C^Z_2,Q_3,c_3^{-1})$ & $(C^Z_1,Q_5,a_2^{-1})$ & $(C^Z_1,Q_1,c_2^{-1})$ \\
10 & $(C^Z_1,Q_5,a_3^{-1})$ & $(C^Z_1,Q_5,a_2^{-1})$ & $(C^Z_2,Q_3,c_2^{-1})$ & $(C^Z_1,Q_5,a_3^{-1}),(C^Z_2,Q_3,c_2^{-1})$ & $(C^Z_1,Q_2,d_2^{-1}),(C^Z_2,Q_3,c_3^{-1})$ \\
11 & $(C^Z_2,Q_5,b_2^{-1})$ & $(C^Z_2,Q_5,b_2^{-1})$ & $(C^Z_1,Q_5,a_2^{-1}),(C^Z_2,Q_4,d_3^{-1})$ & $(C^Z_2,Q_4,d_3^{-1})$ & $(C^Z_2,Q_4,d_3^{-1})$ \\
12 & $(C^Z_2,Q_5,b_3^{-1})$ & $(C^Z_2,Q_5,b_3^{-1})$ & $(C^Z_2,Q_5,b_3^{-1})$ & $(C^Z_2,Q_3,c_3^{-1})$ & $(C^Z_2,Q_3,c_2^{-1})$ \\
13 & -- & -- & -- & -- & $(C^Z_2,Q_4,d_2^{-1})$ \\
\bottomrule
\end{tabular}%
}
\caption{$X$-check (top) and $Z$-check (bottom) gate layers for candidate ZSZ-LP codes under the greedy AOD scheduler. Each entry is $(\text{check subblock},\text{data subblock},\text{monomial})$.}
\label{tab:greedy-aod-gate-schedule}
\end{table}

%%%%%%%%%%%%%%%%%%%%%%%%%%%%%%%%%%%%%%%%%%%%%%%%%%%%%%%%

\section{Additional numerical simulation details}
\label{app:more numerics}

Table \ref{tab:relay-gamma-intervals} lists the relay gamma intervals used for the memory simulations in the main text as well as an upper-bound estimate on the circuit distance of the greedy-scheduled syndrome extraction circuit. To estimate the circuit $X$($Z$)-distance, we construct a shallow circuit consisting of initialization in $\ket{0}^{\otimes n}$ ($\ket{+}^{\otimes n}$), one round of $X$($Z$)-syndrome extraction, and finally transversal measurement in $Z$ ($X$). The detector error model (DEM) of this single-round circuit is sufficient to capture correlated hook errors from syndrome to data qubits. We estimate the circuit distances with $10^6$ iterations of \textsf{QDistEvol} on each corresponding DEM. For ZSZ-LP-390 and ZSZ-LP-550, we did not obtain any circuit-distance estimates lower than the minimum distances of the codes themselves, which are likely just due to the DEMs being too large to obtain reliable estimates with \textsf{QDistEvol}. Following the apparent trend of $d_{\rm circ} \leq d-2$ for ZSZ-LP-160 through ZSZ-LP-320, we set $d_{\rm circ}=d-2$ in the phenomenological fitting function used in the main text.

Table \ref{tab:zsz-lp-390-gpu-latencies} lists the GPU-accelerated Relay-BP decoding latencies of several different NVIDIA GPUs for ZSZ-LP-390 at $p=0.1\%$. Recall that we decode with respect to the circuit-level $X$-type and $Z$-type DEMs over 16 repeated rounds of syndrome extraction for ZSZ-LP-390. For each GPU, the timing data is averaged over $10^5$ shots of $X$-memory and $Z$-memory experiments, with the GPU only decoding one shot at a time (non-batched). Curiously enough, we find that the L40S GPU gives the lowest latency despite being the oldest generation of the bunch. We attribute this peculiar observation to its fast L2 cache latency, which is likely the dominant bottleneck for sequential-shot decoding. When performing batched-shot decoding, VRAM capacity and bandwidth become the dominant bottlenecks, so the more modern GPUs become faster due to their increased VRAM capacity and bandwidth.

\begin{table}[t]
\centering
\begin{tabular}{cccc}
\toprule
\textbf{Code name} & $p_{\rm sweep}$ & $[\gamma_{\rm min},\gamma_{\rm max}]$ & $d_{\rm circ}$ \\ \midrule
ZSZ-LP-160 & $0.2\%$ & $[-0.15,0.61]$ & $\leq8$ \\
ZSZ-LP-240 & $0.2\%$ & $[-0.1,0.6]$ & $\leq10$ \\
ZSZ-LP-320 & $0.2\%$ & $[-0.16,0.66]$ & $\leq12$ \\
ZSZ-LP-390 & $0.2\%$ & $[-0.16,0.7]$ \\
ZSZ-LP-550 & $0.3\%$ & $[-0.12,0.62]$ \\ \bottomrule
\end{tabular}
\caption{Relay-BP gamma intervals used for circuit-level memory decoding using the greedy AOD schedules for syndrome extraction. The intervals are selected by a grid search at $p=p_{\rm sweep}$ for the lowest logical error rates. We also report estimates of the circuit distances $d_{\rm circ}$ which includes correlated hook-error mechanisms.}
\label{tab:relay-gamma-intervals}
\end{table}

\begin{table}[t]
\centering\renewcommand{\arraystretch}{1.1}
\resizebox{\textwidth}{!}{
\begin{tabular}{ccccccc}
\toprule
\textbf{GPU} & \textbf{Generation} & \;\textbf{Socket}\; & \textbf{VRAM} &
\textbf{\makecell{Mean\\ latency}} &
\textbf{\makecell{p99\\ latency}} &
\textbf{\makecell{p99.9\\ latency}} \\ \midrule
B200 & Blackwell & SXM6 & 192\,GB HBM3e & \;2.068 ms\; & \;3.554 ms\; & \;5.248 ms\; \\
RTX PRO 6000 & Blackwell & PCIe5 & 96\,GB GDDR7 & 2.304 ms & 4.033 ms & 6.290 ms \\
H200 & Hopper & SXM5 & 141\,GB HBM3e & 1.794 ms & 3.131 ms & 4.918 ms \\
H100 & Hopper & SXM5 & 80\,GB HBM3 & 1.856 ms & 3.208 ms & 4.875 ms \\
L40S & Lovelace & PCIe4 & 48\,GB GDDR6 & 1.494 ms & 2.572 ms & 3.903 ms \\ \bottomrule
\end{tabular}
}
\caption{Relay-BP sequential-shot decoding latencies of several different NVIDIA GPUs for ZSZ-LP-390 at $p=0.1\%$. The reported mean, 99th percentile (p99), and 99.9th percentile (p99.9) latencies are from $10^5$ shots, averaged over $X$-memory and $Z$-memory experiments.}
\label{tab:zsz-lp-390-gpu-latencies}
\end{table}

\end{document}